\pdfoutput=1
\documentclass[a4paper,USenglish]{lipics-v2019}

\nolinenumbers
\hideLIPIcs

\usepackage{thmtools,thm-restate}

\usepackage{amsmath}
\usepackage{amsthm}
\usepackage{amssymb}
\usepackage{tikz}
\usetikzlibrary{positioning,automata,fit,shapes,calc}
\usepackage{dsfont}
\usepackage{caption}
\usepackage{subcaption}
\usepackage{wrapfig}

\usepackage[section]{placeins}
\newcommand{\PE}{\mathit{PE}}
\newcommand{\CE}{\mathit{CE}}
\newcommand{\PMP}{\mathit{WLF}}

\newcommand{\VaR}{\mathit{VaR}}
\newcommand{\CVaR}{\mathit{CVaR}}

\newcommand{\tudparagraph}[2]{%
\vspace*{#1}

\noindent
{\bf #2}
}

\newcommand{\cA}{\mathcal{A}}
\newcommand{\cB}{\mathcal{B}}
\newcommand{\cC}{\mathcal{C}}

\newcommand{\cE}{\mathcal{E}}

\newcommand{\cK}{\mathcal{K}}
\newcommand{\cL}{\mathcal{L}}
\newcommand{\cM}{\mathcal{M}}
\newcommand{\cN}{\mathcal{N}}

\newcommand{\eqdef}{\ensuremath{\stackrel{\text{\tiny def}}{=}}}

\renewcommand{\Pr}{\mathrm{Pr}}

\newcommand{\threshold}{\vartheta}

\newcommand{\flrprob}[2]{\mathit{lrp}^{#1}_{\scriptscriptstyle #2}}

\newcommand{\LPr}[3]{\mathbb{LP}^{#1}_{#2}(#3)}

\newcommand{\after}[2]{\residual{#1}{#2}}
\newcommand{\SatPoint}{K}

\newcommand{\FM}{\mathrm{FM}}
\newcommand{\fmstate}[2]{\mathfrak{#1}} 
\newcommand{\lgth}{n}

\newcommand{\eap}{\mathit{eaw}}
\newcommand{\ens}{\mathit{ens}}

\newcommand{\AP}{\mathsf{AP}}

\newcommand{\sinit}{s_{\mathit{\scriptscriptstyle init}}}

\newcommand{\Act}{\mathit{Act}}

\newcommand{\act}{\alpha}

\newcommand{\fpath}{\pi}
\newcommand{\finpath}{\varrho}

\newcommand{\infpath}{\varsigma}

\newcommand{\last}{\mathit{last}}

\newcommand{\fragment}[3]{#1[#2 \ldots #3]}
\newcommand{\prefix}[2]{\mathit{pref}(#1,#2)}

\newcommand{\sched}{\mathfrak{S}}
\newcommand{\tsched}{\mathfrak{T}}

\newcommand{\fsched}{\mathfrak{F}}
\newcommand{\qsched}{\mathfrak{Q}}
\newcommand{\rsched}{\mathfrak{R}}

\newcommand{\residual}[2]{#1 {\uparrow} {#2}}

\newcommand{\wgt}{\mathit{wgt}}

\newcommand{\goal}{\mathit{goal}}
\newcommand{\fail}{\mathit{fail}}

\newcommand{\Fail}{\mathit{Fail}}

\newcommand{\Goal}{\mathit{Goal}}

\newcommand{\neXt}{\bigcirc}

\DeclareMathOperator{\Until}{\ensuremath{\mathrm{U}}}
\DeclareMathOperator{\until}{\Until}

\newcommand{\Nat}{\mathbb{N}}
\newcommand{\Rational}{\mathbb{Q}}

\newcommand{\LP}{\mathbb{LP}}

\newcommand{\CiteAppendix}[1]{}

\usepackage{amssymb}
\usepackage{tikz}

\newcommand{\rawdiaplus}{%
  \begin{tikzpicture}
    \useasboundingbox (-0.7ex, -0.9ex) rectangle (0.7ex, 0.9ex);
    \node (w) at (-0.7ex,0) {};
    \node (e) at (+0.7ex,0) {};
    \node (s) at (0,-0.9ex) {};
    \node (n) at (0,+0.9ex) {};
    \draw (n.center) -- (e.center) -- (s.center) -- (w.center) -- (n.center);
    \draw (n.center) -- (s.center);
    \draw (e.center) -- (w.center);
  \end{tikzpicture}}

\newsavebox{\diamondplusbox}
\savebox{\diamondplusbox}{\rawdiaplus}

\newcommand{\rawdiaminus}{%
  \begin{tikzpicture}
    \useasboundingbox (-0.7ex, -0.9ex) rectangle (0.7ex, 0.9ex);
    \node (w) at (-0.7ex,0) {};
    \node (e) at (+0.7ex,0) {};
    \node (s) at (0,-0.9ex) {};
    \node (n) at (0,+0.9ex) {};
    \draw (n.center) -- (e.center) -- (s.center) -- (w.center) -- (n.center);
    \draw (e.center) -- (w.center);
  \end{tikzpicture}}

\newsavebox{\diamondminusbox}
\savebox{\diamondminusbox}{\rawdiaminus}

\theoremstyle{definition}
\newtheorem{mydef}{Definition}

\theoremstyle{plain}
\newtheorem{myprop}[mydef]{Proposition}

\theoremstyle{definition}
\newtheorem{myex}[mydef]{Example}

\theoremstyle{plain}
\newtheorem{mylem}[mydef]{Lemma}

\theoremstyle{plain}
\newtheorem{mycor}[mydef]{Corollary}

\theoremstyle{plain}
\newtheorem{mythm}[mydef]{Theorem}

\theoremstyle{definition}
\newtheorem{myrem}[mydef]{Remark}

\usepackage{tabularx}
    \newcolumntype{L}{>{\raggedright\arraybackslash}X}

\newtheorem*{theorem*}{Main result}


\title{On Skolem-hardness and saturation points in Markov decision processes\\
\vspace{12pt}
\Large{(Extended Version)}}%
\author{Jakob Piribauer}
{Technische Universit{\"a}t Dresden}
{jakob.piribauer@tu-dresden.de}
{}
{}
\author{Christel Baier}
{Technische Universit{\"a}t Dresden}
{christel.baier@tu-dresden.de}
{}
{}
\authorrunning{J. Piribauer and C. Baier}

\funding{ The authors are supported by the DFG through
        the DFG grant 389792660 as part of {TRR~248},
        the Cluster of Excellence EXC 2050/1 (CeTI, project ID 390696704, as part of Germany's Excellence Strategy), 
        the Research Training Group QuantLA (GRK 1763),
        the Collaborative Research Centers CRC 912 (HAEC),
        	the DFG-project BA-1679/11-1, and
        the DFG-project BA-1679/12-1.}

\Copyright{Jakob Piribauer, Christel Baier}


\EventEditors{Artur Czumaj, Anuj Dawar, and Emanuela Merelli}
\EventNoEds{3}
\EventLongTitle{47th International Colloquium on Automata, Languages, and Programming (ICALP 2020)}
\EventShortTitle{ICALP 2020}
\EventAcronym{ICALP}
\EventYear{2020}
\EventDate{July 8--11, 2020}
\EventLocation{Saarbrücken, Germany (virtual conference)}
\EventLogo{}
\SeriesVolume{168}
\ArticleNo{138}

\begin{document}

\maketitle

\begin{abstract}

The Skolem problem and the related Positivity problem for linear recurrence sequences are outstanding number-theoretic problems whose decidability has been open for many decades. 
In this paper, the inherent mathematical difficulty of a series of optimization problems on Markov decision processes (MDPs) is shown by a reduction from the Positivity problem to the associated decision problems which establishes that the problems are also at least as hard as the Skolem problem as an immediate consequence.
The optimization problems under consideration are two non-classical variants of the stochastic shortest path problem (SSPP) in terms of expected partial or conditional 
accumulated weights, the optimization of the conditional value-at-risk for accumulated weights, and two problems addressing the long-run satisfaction of path properties, namely the optimization of long-run probabilities of regular co-safety properties and the model-checking problem of the logic frequency-LTL. To prove the Positivity- and  hence Skolem-hardness  for the latter two problems, a new  auxiliary 
path measure, called weighted long-run frequency, is introduced  and the Positivity-hardness of the corresponding decision problem is shown as an intermediate step.
For the partial and conditional SSPP on MDPs with non-negative weights and for the optimization of long-run probabilities of constrained reachability properties ($a\Until b$),
solutions
are known 
that 
rely on the identification of a bound on the accumulated weight or the number of consecutive visits to certain sates, called a saturation point, from which on
optimal schedulers behave memorylessly. In this paper, it is shown that  also the optimization of the conditional value-at-risk for the classical SSPP and of weighted long-run frequencies on MDPs with non-negative weights can  be solved in pseudo-polynomial time exploiting the existence of a saturation point. As a consequence, one obtains the decidability of the qualitative model-checking problem of a frequency-LTL formula that is not included in the fragments with known solutions.

\end{abstract}	

\newpage

\section{Introduction}


Markov decision processes (MDPs) (see, e.g., \cite{puterman1994}) constitute one of the most prominent classes of 
operational models combining randomization and non-determinism and are widely used in verification, articifical intelligence, robotics and operations research.
Consequently, a vast landscape of optimization problems on MDPs has been studied. The  task usually is to find a strategy resolving the non-deterministic choices, called a \emph{scheduler}, such that a certain objective quantity is optimized or to decide whether the optimal value exceeds a given rational threshold (\emph{threshold problem}).

\emph{Stochastic shortest path problems (SSPPs)} are one important type of such optimization problems on MDPs equipped with weights. These problems ask for a scheduler maximizing or minimizing the expected accumulated weight before reaching a designated goal state. In the classical setting, only schedulers reaching the goal almost surely are taken into consideration. This \emph{classical SSPP} is known to be solvable in polynomial time using graph-based algorithms and linear-programming techniques \cite{bertsekas1991,deAlfaro1999,lics2018}. 
For various purposes, the requirement that the goal has to be reached almost surely, however, is not appropriate. 
This applies, e.g., to work on  the semantics of probabilistic programs when no guarantee on almost sure termination can be given \cite{gretz2014,katoen2015,barthe2016,chatterjee2016,olmedo2018}, to the analysis of the behavior of fault-tolerant systems in error scenarios which occur with low probability, or to the trade-off analysis when combinations of utility and cost constraints can be achieved with positive probability, but not almost surely (see, e.g., \cite{baier2014}).
This motivates a switch to non-classical variants of the SSPP: The  \emph{conditional SSPP}
\cite{tacas2017}  asks for a scheduler optimizing the conditional expected accumulated weight before reaching the goal under the condition that the goal will indeed be reached and  the \emph{partial SSPP} \cite{chen2013,fossacs2019}  assigns weight $0$ to all executions not reaching the goal. Both variants increase the algorithmic difficulties.  In the special case of MDPs with non-negative weights, exponential-time algorithms for the partial and conditional SSPP  exploit the monotonicity of accumulated weights and rely on the existence of a \emph{saturation point} (a bound for the accumulated weight) from which on optimal schedulers behave memorylessly. Apart from a PSPACE lower bound and approximation algorithms \cite{fossacs2019}, no algorithms are known for solving the partial or conditional SSPP in integer-weighted MDPs.

Conditional expectations also play a crucial role in risk management: The \emph{conditional value-at-risk} is an established risk measure quanitfying the expected loss in bad cases \cite{Uryasev00,AcerbiTasche02}. Given a probability value $p$, the \emph{value-at-risk} of a random variable $X$ is defined as the worst $p$-quantile. Quantile queries on the distribution of path lengths have been studied in \cite{UB13}. The conditional value-at-risk is  the expectation of $X$ under the condition that the outcome is worse than the value-at-risk. For MDPs, the conditional value-at-risk has been studied for mean-payoffs and for  weighted reachability where on each run only once a terminal weight is collected when a target state is reached \cite{kretinsky2018}. In this paper, we consider the conditional value-at-risk for the more general accumulated weight before reaching the goal, i.e. for the classical SSPP. To the best of our knowledge, this problem has not been studied.

Other typical optimization problems arise in the context of verification, asking for worst-case schedulers that minimize or maximize the probability of a given path property. 
While such problems are well-understood, e.g., for properties given by linear temporal logic (LTL)-formulas or non-deterministic B\"uchi-automata \cite{CY95}, there has been increasing interest in ways to quantify the degree to which a property is satisfied not only by the probability (see \cite{Henzinger2013}). Approaches in this direction include the work on
 robust satisfaction of temporal specifications \cite{KupfVar99,TabNeider16},  coverage semantics \cite{ChockKupfVardi06}, robustness distances \cite{CerHenRad12},
and the more general model-measurement semantics \cite{HenOtop13} among others.
Furthermore, this has  lead to different notions quantifying to which degree a property is satisfied in the long-run: 
\emph{Frequency-LTL} 
has been introduced in \cite{ForejtK15,ForejtKK15} 
as an extension of LTL by a frequency modality that makes assertions on the portion of time (or relative frequency of positions in paths) where a given event holds.
While \cite{ForejtK15,ForejtKK15} presents model-checking algorithms for Markov chains and arbitrary frequency-LTL formulas, the presented model checking algorithms for MDPs are restricted to fragments of frequency-LTL. We address the model checking problem for frequency-LTL formulas not contained in these fragments. Further, the concept of \emph{long-run probabilities} \cite{lics2019} has been introduced for reasoning about the probabilities of path properties when the system is in equilibrium and can, e.g., be useful to formalize refined notions of long-run availability. 
In \cite{lics2019}, a pseudo-polynomial time algorithm that exploits the existence of a saturation point for the computation of optimal long-run probabilities of
 constrained reachability properties ($a\Until b$) is provided. 
Here, we study long-run probabilities of  general regular co-safety properties.

\tudparagraph{.3ex}{Contributions.} The main contribution of the paper is to provide evidence for the mathematical difficulty of the series of decision problems described above in terms of a reduction from the \emph{Positivity problem} of linear recurrence sequences. The Positivity problem is closely related to the \emph{Skolem problem}, a prominent number-theoretic decision problem for linear recurrence sequences, and the decidability of both problems has been open for many decades (see, e.g., \cite{halava2005}). As it is well-known that the Skolem problem is reducible to the Positivity problem, the provided reductions establish that the investigated decision problems are also at least as hard as the Skolem problem. In the middle column of Table \ref{tab:overview}, these Skolem-hardness results are listed:
\vspace{.8ex}

\begin {table}[h]
\caption {Overview of the results} \label{tab:overview} 
\begin{tabularx}{\textwidth}{|l|l|L|}
\hline
optimization problem   & threshold problem  Positivity-& exponential-time algorithm \\
on MDPs &and hence Skolem-hard for &using a saturation point for\\
\hline
\hline
partial SSPP (1)  &  weights in $\mathbb{Z}$, \emph{Thm. \ref{thm:threshold_PE}} & weights in $\mathbb{N}$ \cite{chen2013}\\
 & & (PSPACE-hard, \emph{Prop. \ref{prop:CE_to_PE}}) \\
\hline
conditional SSPP (2) & weights in $\mathbb{Z}$, \emph{Thm. \ref{thm:threshold_CE}} & weights in $\mathbb{N}$ \cite{tacas2017}\\
& & (PSPACE-hard \cite{tacas2017})\\
\hline
conditional value-at-risk & weights in $\mathbb{Z}$, \emph{Thm. \ref{thm:Skolem_cvar}} & weights in $\mathbb{N}$, \emph{Thm. \ref{thm:cvar_pos}}  \\
 for the classical SSPP (3) &  & \\
 \hline
long-run probability (4) & regular co-safety properties, & constrained reachability $a\Until b$ \cite{lics2019} \\
&   \emph{Thm. \ref{thm:Skolem_LP}} & (NP-hard \cite{lics2019}) \\
\hline
model checking of   & $\Pr^{\max}_\cM (G^{>\vartheta}_{\inf} (\varphi))=1$? & $\Pr^{\max}_\cM (G^{>\vartheta}_{\inf} (a\Until b))=1$? \\
frequency-LTL (5) &for an LTL-formula $\varphi$, \emph{Thm. \ref{thm:fLTL}}& \emph{Cor. \ref{cor:frequency-LTL}}\\
\hline

\end{tabularx}
\vspace{-12pt}
\end{table}
\vspace{.3ex}

 To obtain these results, we construct an MDP-gadget in which a linear recurrence relation can be encoded. Together with different gadgets encoding initial values of a linear recurrence sequence, we use this gadget  to establish Positivity-hardness for problems (1)-(3). Afterwards, we introduce a notion of \emph{weighted long-run frequency} for constrained reachability properties that can be seen as a generalization of classical limit-average weights and serves here as a technical vehicle to provide a connective link to long-run probabilities and the model-checking problem of frequency-LTL. 
The Positivity-hardness for problems (4) and (5) is obtained via the Positivity-hardness of the threshold problem for weighted long-run frequencies by showing how to encode  integer weights  in terms of the satisfaction of a fixed  co-safety property.
The Positivity-hardness of (4) and (5) is somehow surprising: The non-probabilistic variant (4) is shown to be decidable in \cite{lics2019}, while our results show that Positivity-hardness of (4) holds even for a simple fixed co-safety property given by a very small counter-free non-deterministic finite automaton.
Likewise, Positivity-hardness of (5) is established already for the restriction to
 the almost-sure satisfaction problem of a simple fixed frequency-LTL formula.
 
For special cases of some of  the problems studied here it is known that optimal values can be computed in exponential time exploiting a saturation point. We extend this picture by showing  analogous results for problems (3) and (5) (see Table \ref{tab:overview}). In particular, we provide a simple exponential time algorithm for the computation of the optimal conditional value-at-risk for the classical SSPP. Further, we pinpoint where the Positivity-hardness of the model checking problem of frequency-LTL arises:
We observe that the techniques of \cite{lics2019} allow to solve the qualitative model-checking problem for a frequency-LTL formula with only one constrained reachability ($a\Until b$) property under a frequency-globally modality. 
Our Positivity-hardness result for  model checking  frequency-LTL uses an only slightly more complicated fixed formula where a Boolean combination of atomic propositions and constrained reachability properties occurs in the scope of the frequency-globally modality. In particular, the Positivity-hardness does not require deeper nesting of temporal operators.

\tudparagraph{.3ex}{Related work.}
Besides the above cited work that presents algorithms for special cases of 
the investigated problems, closest to our work is 
\cite{akshay2015} where Skolem-hardness for decision problems
for Markov chains have been established. The problems are to decide whether for given states $s$, $t$
and rational number
  $p$, there is a positive integer $n$ such that the 
probability to reach $t$ from $s$ in $n$ steps equals $p$
and
the model checking problem for a probabilistic variant of monadic logic 
and  a variant of LTL that treats Markov chains as linear transformers of
probability distributions.
These decision problems are of quite different nature than the
problems studied here, and so are the reductions from the Skolem problem.
In this context also the results of 
\cite{ChonOuakWor16} and \cite{MSS20} are remarkable
as they show the decidability (subject to Schanuel's conjecture)  of 
reachability problems in
continuous linear dynamical systems and continuous-time MDPs, respectively,
as  instances of the continuous Skolem problem.

A class of problems related to SSPPs concerns the optimization of 
 probabilities for weight-bounded reachability properties and also exhibits increasing algorithmic difficulty (for an overview see \cite{randour2015}):
For non-negative weights, schedulers optimizing the probability for reaching a target while the accumulated weight stays below a given bound are computable in pseudo-polynomial time and the corresponding probability-threshold problem is in P for qualitative probability thresholds (``${>}0$'' or ``${=}1$'') and PSPACE-hard in the general case \cite{UB13,haase2015}. For integer weights even in finite-state Markov chains, the probabilities for a weight-bounded reachability property can be irrational. Still, decidability  for analogous problems for integer-weighted MDPs have been established for certain cases. Examples are pseudo-polynomial algorithms for qualitative threshold problems in integer-weighted MDPs \cite{ChatDoy11,BKN16,MaySchTozWoj17,lics2018} or an exponential-time algorithm and a PSPACE lower bound for the almost-sure termination problem in one-counter MDPs~\cite{brazdil2010}. 

Switching to more expressive models typically leads to the undecidability of infinite-horizon verification problems. This applies, e.g., to recursive MDPs \cite{etessami2005}, MDPs with two or more weight functions \cite{BKKW14,randour2017} or partially observable MDPs \cite{madani1999,BaiGroeBer12}. However, we are not aware of natural decision problems for standard (finite-state) MDPs with a single weight function and single objective that are known to be undecidable.


\section{Preliminaries} \label{sec:preliminaries}


We give basic definitions and present our notation (for more details see, e.g., \cite{puterman1994}). We then formally define the quantitative objectives studied in this paper.
\tudparagraph{.3ex}{Notations for Markov decision processes.}
A \emph{Markov decision process} (MDP) is a tuple $\mathcal{M} = (S,\Act,P,\sinit,\wgt,\AP,L)$
where $S$ is a finite set of states,
$\Act$ a finite set of actions,
$\sinit \in S$ the initial state,
$P \colon S \times \Act \times S \to [0,1] \cap \Rational$ is the
transition probability function,
$\wgt \colon S \times \Act \to \mathbb{Z}$ the weight function,
$\AP$ a finite set of atomic propositions, and
$L \colon S \to 2^{\AP}$ a labeling function.
If not needed, we might drop the weight function or the labeling. 
We require that
$\sum_{t\in S}P(s,\act,t) \in \{0,1\}$
for all $(s,\alpha)\in S\times \Act$.
We say that action $\alpha$ is enabled in state $s$ iff $\sum_{t\in S}P(s,\act,t) =1$.
 We assume that for all states $s$ there is an enabled action  and 
that all states are reachable from $s_{\mathit{init}}$. 
We call a state absorbing if there is only one enabled action, returning to the state  with probability $1$ and weight $0$.
The paths of $\cM$ are finite or
infinite sequences $s_0 \, \act_0 \, s_1 \, \act_1  \ldots$
where states and actions alternate such that
$P(s_i,\act_i,s_{i+1}) >0$ for all $i\geq0$.
For $\fpath =
    s_0 \, \act_0 \, s_1 \, \act_1 \,  \ldots \act_{k-1} \, s_k$,
  $\wgt(\fpath)=
   \wgt(s_0,\act_0) + \ldots + \wgt(s_{k-1},\act_{k-1})$
  denotes the accumulated weight of $\pi$,
    $P(\fpath) =
   P(s_0,\act_0,s_1) 
   \cdot \ldots \cdot P(s_{k-1},\act_{k-1},s_k)$
   its probability, and
$\last(\fpath)=s_k$ its last state. Further, we also write  $\fpath$ to denote the word $L(s_0),L(s_1),\dots$.
The \emph{size} of $\cM$ 
is the sum of the number of states 
plus the total sum of the logarithmic lengths of the non-zero
probability values
$P(s,\alpha,s')$ as fractions of co-prime integers and the weight values $\wgt(s,\alpha)$.
An {end component} of $\cM$ is a strongly connected sub-MDP.


\tudparagraph{.3ex}{Scheduler.}
A \emph{scheduler} for $\cM$
is a function $\sched$ that assigns to each finite path $\fpath$ 
a probability distribution over $\Act(\last(\fpath))$.
If there is a finite set $X$ of memory modes and a memory update function $U:S\times \Act \times S \times X \to X$ such that the choice of $\sched$ only depends on the current state after a finite path and the memory mode obtained from updating the memory mode according to $U$ in each step, we say that $\sched$ is a finite-memory scheduler. 
If the choice depends only on the current state, we say that $\sched$ is memoryless.
A scheduler $\sched$ is called deterministic if $\sched(\fpath)$ is a Dirac distribution
for each path $\fpath$.
Given a scheduler $\sched$,
$\zeta \, = \, s_0 \, \act_0 \, s_1 \, \act_1 \ldots$
is a $\sched$-path iff $\zeta$ is a path and
$\sched(s_0 \, \act_0  \ldots \act_{k-1} \, s_k)(\act_k)>0$
for all $k \geq 0$.

\tudparagraph{.3ex}{Probability measure.}
We write $\Pr^{\sched}_{\cM,s}$ or briefly $\Pr^{\sched}_{s}$
to denote the probability measure induced by $\sched$ and $s$.
For details, see \cite{puterman1994}.
We will use LTL-like formulas to denote measurable sets of paths.
Given a measurable set $\psi$ of infinite paths, we define
$\Pr^{\min}_{\cM,s}(\psi) = \inf_{\sched} \Pr^{\sched}_{\cM,s}(\psi)$
and
$\Pr^{\max}_{\cM,s}(\psi) = \sup_{\sched} \Pr^{\sched}_{\cM,s}(\psi)$
where $\sched$ ranges over all schedulers for $\cM$.
For a random variable $X$ defined on infinte paths in $\cM$, we denote the expected value of $X$ under the probability measure induced by a scheduler $\sched$ and state $s$ by $\mathbb{E}^{\sched}_{\cM,s}(X)$. Furthermore, for a measurable set of paths $\psi$ with positive probability, $\mathbb{E}^{\sched}_{\cM,s}(X|\psi)$ denotes the conditional expectation of $X$ under $\psi$. If $s=\sinit$, we sometimes drop the subscript $s$.

\tudparagraph{.3ex}{Partial and conditional SSPP.}
Let $\cM$ be an MDP with an absorbing state $\goal$. On infinite paths $\zeta$, 
we define the random variable $\oplus \goal (\zeta)$ to be 
$ 
\wgt(\zeta)$ if $\zeta \vDash \Diamond \goal$, and to be $0$ otherwise.
 The \emph{partial expectation} ${\PE}^\sched_{\cM,s}$ of a scheduler $\sched$ is defined as $\mathbb{E}^\sched_{\cM,s}(\oplus \goal)$. The maximal partial expectation is ${\PE}^{\max}_{\cM,s}=\sup_\sched {\PE}^\sched_{\cM,s}$. 
 The \emph{conditional expectation} ${\CE}^\sched_{\cM,s}$ is defined as the conditional expected value $\mathbb{E}^\sched_{\cM,s}(\oplus \goal|\Diamond \goal)$ for all schedulers reaching $\goal$ with positive probability, 
 and the maximal conditional expectations is ${\CE}^{\max}_{\cM,s}=\sup_\sched {\CE}^\sched_{\cM,s}$ 
 where $\sched$ ranges over all schedulers $\sched$ with $\Pr^\sched_{\cM,s}(\Diamond \goal)>0$. 
 The \emph{partial SSPP} asks for the maximal partial expectations and the \emph{conditional SSPP} for the maximal conditional expectation. These problems were first considered in \cite{chen2013} and \cite{tacas2017}.
 For more details see \cite{tacas2017,fossacs2019}.

\tudparagraph{.3ex}{Conditional value-at-risk.}
Given an MDP $\cM$ with a scheduler $\sched$, a random variable $X$ defined on runs of the MDP with values in $\mathbb{R}$ and a value $p\in [0,1]$, we define the \emph{value-at-risk} as 
$\VaR^{\sched}_{p}(X) = \sup \{r\in \mathbb{R}| \Pr_{\cM}^\sched (X\leq r)\leq p\}$.
So, the value-at-risk is the point at which the cumulative distribution function of $X$ reaches or exceeds $p$. 
Denote $\VaR_p^\sched(X)$ by $v$. 
The \emph{conditional value-at-risk} is now the expectation of $X$ under the condition that the outcome belongs to the $p$ worst outcomes.
Following the treatment of  random variables that are not continuous in general in \cite{kretinsky2018}, we define the conditional value-at-risk as follows:
\[\CVaR_p^\sched(X) = 1/p (  \Pr_{\cM}^{\sched}(X<v) \cdot \mathbb{E}_{\cM}^\sched(X|X<v) + (p- \Pr_{\cM}^{\sched}(X<v))\cdot v  ).\]
Outcomes of $X$ which are less than $v$  are treated differently to outcomes equal to $v$ as it is possible that the outcome $v$ has positive probability and we only want to 
 account exactly for the $p$ worst outcomes. Hence,  we take only $p- \Pr_{\cM}^{\sched}(X<v)$ of the outcomes which are exactly $v$ into account as well.
 
Threshold problems for the conditional value-at-risk in weighted MDPs have been studied in \cite{kretinsky2018} for two random variables: the mean-payoff and weighted reachability where a set of final states is equipped with  terminal weights obtained when reaching these states while all other transitions have weight $0$. In this paper, we will address the conditional value-at-risk for the accumulated weight before reaching $\goal$ in MDPs with an absorbing state $\goal$:
Define  $\rawdiaplus\goal (\zeta)$ to be 
$\wgt(\zeta)$ if $\zeta\vDash \Diamond \goal$ and leave it undefined otherwise.
The optimization of the expectation of $\rawdiaplus \goal$ is known as the $\emph{classical SSPP}$. Note that the expectation of this random variable is only defined under schedulers reaching $\goal$ with probability $1$.

\tudparagraph{.3ex}{Long-run probability.}
 Let $\cM$ be an MDP with states labeled by atomic propositions from $\AP$. Let $\varphi$ be a path property, i.e., a measurable set of paths.
The \emph{long-run probability} for $\varphi$ 
of a path $\zeta$ under a
scheduler $\sched$ is
$
  \flrprob{\sched}{\varphi}(\zeta) 
    =   
  \liminf_{n \to \infty} \
       \frac{1}{n{+}1} \cdot \sum_{i=0}^{n}
       \Pr^{\sched \after{} \zeta [0\dots i]}_{\cM,\zeta [i]}
        (\varphi) 
$.  
Here, $\zeta [0\dots i]$ denotes the prefix from position $0$ to $i$ of $\zeta$,   $\zeta [i]$ denotes the state after $i$ steps, and $\sched \after{} \zeta [0\dots i]$  denotes the residual scheduler defined by $\sched \after{} \zeta [0\dots i](\pi) = \sched(\zeta [0\dots i] \circ \pi)$ for all finite paths $\pi$ starting in $\zeta[i]$. 
The long-run probability of $\varphi$ under scheduler $\sched$ is
$
  \LPr{\sched}{\cM}{\varphi}  = 
     \mathbb{E}^{\sched}_{\cM}(\flrprob{\sched}{\varphi} )
$.
The maximal
long-run probability for $\varphi$ is 
$
  \LPr{\max}{\cM}{\varphi}  = 
  \sup _ \sched 
     \mathbb{E}^{\sched}_{\cM}(\flrprob{\sched}{\varphi} ).
$
This notion was introduced in \cite{lics2019}.
In this paper, we are interested in two kinds of path properties: Constrained reachability, $a\Until b$, where $a$ and $b$ are atomic propositions and the more general regular co-safety properties given by a finite non-deterministic automaton (NFA) $\cA$ accepting ``good'' prefixes of a run. For a co-safety property given by an NFA $\cA$, we also write  $\LPr{\max}{\cM}{\cA}$.


\section{Skolem-hardness} \label{sec:linear_recurrence}


The \emph{Skolem problem} and the closely related \emph{Positivity problem} are outstanding problems in the fields of number theory and theoretical computer science (see, e.g., \cite{halava2005,ouaknine2014}). Their decidability has been open for many decades. 
We call a problem to which the Skolem problem is reducible \emph{Skolem-hard}. This is a hardness result in the sense that a decision procedure would imply a major breakthrough by settling the decidability of the Skolem problem and it shows that a problem possesses an inherent mathematical difficulty.

\tudparagraph{.3ex}{Skolem problem.}  Given a natural number $k\geq 2$, and rationals $\alpha_i$ and $\beta_j$ with $1\leq i \leq k$ and $0\leq j \leq k-1$, let $(u_n)_{n\geq0}$ be defined  by the initial values $u_0=\beta_0$, \dots, $u_{k-1}=\beta_{k-1}$ and the linear recurrence relation
$ u_{n+k} = \alpha_1 u_{n+k-1} + \dots + \alpha_k u_n $
for all $n\geq 0$. The Skolem problem is to decide whether there is an $n\in\mathbb{N}$ with $u_n = 0$.

A closely related problem is the \emph{Positivity problem}. It asks whether $u_n\geq 0$ for all $n$.  It is folklore that the 
 Skolem problem is polynomial-time reducible to the positivity problem (see, e.g., \cite{everest2003}). We will  use the Positivity problem for our reductions leading to the main result:

\begin{theorem*}[Theorems \ref{thm:threshold_PE}, \ref{thm:threshold_CE}, \ref {thm:Skolem_cvar}, \ref{thm:Skolem_LP}, \ref{thm:fLTL}] The Positivity problem and hence the Skolem problem are polynomial-time reducible to the threshold problems for the partial and  conditional SSPP,   the conditional value-at-risk in the classical SSPP, and long-run probabilities of regular co-safety properties, as well as to the qualitative model checking problem of frequency-LTL. 
\end{theorem*}

For this purpose, we will construct an MDP gadget depicted in Figure \ref{fig:linear_recurrence} that encodes a linear recurrence relation in terms of the optimal values  of  different quantitative objectives. 
For the different problems, we then provide  gadgets encoding the initial values of a linear recurrence sequence. We can  plug these gadgets together to obtain an MDP and a scheduler
 $\sched$  such that $\sched$ maximizes the respective objective iff the linear recurrence sequence has no negative member. By computing the optimal values under $\sched$ in the MDPs -- which turn out to be  rational -- we provide  reductions from the positivity problem to the respective threshold problems with strict inequality (see also Remark \ref{rem:strict}).

\subsection{Partial and Conditional SSPP}\label{sub:Skolem_partial}
\label{subsec:hardness_threshold_PE}

Given a linear recurrence sequence, we construct an MDP in which the sequence is encoded in terms of optimal partial expectations. 
So let $k$ be a natural number and let $(u_n)_{n\geq 0}$ be the linear recurrence sequence given by rationals $\alpha_i$ for $1\leq i \leq k$ and $\beta_j$ for $0\leq j \leq k{-}1$ as above. 
As
$u_{n+k} = \alpha_1 u_{n+k-1} + \dots + \alpha_k u_n $
for all $n$, we see that for any positive $\lambda\in \mathbb{Q}$ the sequence $(v_n)_{n\geq 0}$ defined by $v_n = \lambda^{n+1} u_n$ satisfies
$v_{n+k} = \lambda^1  \alpha_1 v_{n+k-1} + \dots + \lambda^k \alpha_k v_n$ for all $n$. Furthermore, $v_n$ is non-negative if and only if $u_n$ is.
W.l.o.g., we hence can assume that   $\sum_i |\alpha_i|<\frac{1}{4}$ and that $0\leq \beta_j< \frac{1}{4k^{2k+2}}$ for all $j$ (see Appendix \ref{app:hardness_threshold_PE}). 

\FloatBarrier

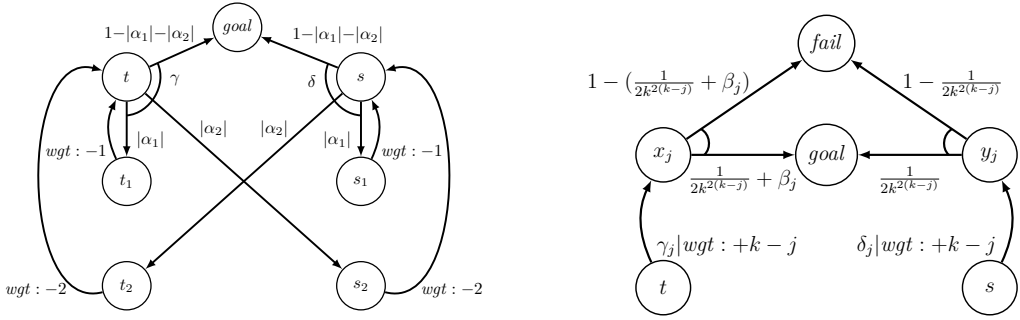
\begin{figure}[t]
     \centering
     \begin{subfigure}[b]{0.47\textwidth}
         \centering

          \resizebox{1\textwidth}{!}{
\begin{tikzpicture}[scale=1,auto,node distance=8mm,>=latex]
\large
    \tikzstyle{round}=[thick,draw=black,circle]

    \node[round,minimum size=30pt] (t) {$t$};
    \node[round,below=12mm of t, minimum size=30pt] (t1) {$t_1$};
    \node[round,below=12mm of t1, minimum size=30pt] (t2) {$t_2$};
    
    \node[round, above=0mm of t,xshift=24mm ,minimum size=30pt] (goal) {$\goal$};
        
    \node[round,right=40mm of t,minimum size=30pt] (s) {$s$};
    \node[round,below=12mm of s, minimum size=30pt] (s1) {$s_1$};
    \node[round,below=12mm of s1, minimum size=30pt] (s2) {$s_2$};

  \draw[color=black , very thick,->] (t) edge  node [ very near start, anchor=center] (h1) {} node [pos=0.0,above=11pt] {$1{-}|\alpha_1|{-}|\alpha_2|$} (goal) ;
  \draw[color=black , very thick,->] (t)  edge node [near start, anchor=center] (h2) {} node [pos=0.7,right=2pt] {$|\alpha_1|$} (t1) ;
  \draw[color=black , very thick,->] (t)  edge node [pos=0.2,right=5pt] {$|\alpha_2|$} (s2) ;
  \draw[color=black , very thick] (h1.center) edge [bend left=55] node [pos=0.25,right=2pt] {$\gamma$} (h2.center);
  
 \draw[color=black , very thick, ->] (t1) edge [bend left=25] node [pos=0.1,left=0pt] {$\wgt:-1$} (t);
  \draw[color=black , very thick, ->] (t2) edge [bend left=105] node [pos=0.1,left=4pt] {$\wgt:-2$} (t);

  \draw[color=black , very thick,->] (s) edge  node [ very near start, anchor=center] (g1) {} node [pos=0.0,above=11pt] {$1{-}|\alpha_1|{-}|\alpha_2|$} (goal) ;
  \draw[color=black , very thick,->] (s)  edge node [near start, anchor=center] (g2) {} node [pos=0.7,left=2pt] {$|\alpha_1|$} (s1) ;
  \draw[color=black , very thick,->] (s)  edge node [pos=0.2,left=5pt] {$|\alpha_2|$} (t2) ;
  \draw[color=black , very thick] (g1.center) edge [bend right=55] node [pos=0.25,left=2pt] {$\delta$} (g2.center);
  
 \draw[color=black , very thick, ->] (s1) edge [bend right=25] node [pos=0.1,right=0pt] {$\wgt:-1$} (s);
  \draw[color=black , very thick, ->] (s2) edge [bend right=105] node [pos=0.1,right=4pt] {$\wgt:-2$} (s);

\end{tikzpicture}
}
         
         \caption{In the depicted example, the recurrence depth is $2$, $\alpha_1>0$, and $\alpha_2<0$.
         }
         \label{fig:linear_recurrence}
     \end{subfigure}
     \hfill
     \begin{subfigure}[b]{0.47\textwidth}
         \centering
               \resizebox{.9\textwidth}{!}{
\begin{tikzpicture}[scale=1,auto,node distance=8mm,>=latex]
\Large
    \tikzstyle{round}=[thick,draw=black,circle]

    \node[round,minimum size=30pt] (t) {$t$};
    \node[round,above=15mm of t, minimum size=30pt] (xj) {$x_j$};

    \node[round, right=20mm of xj ,minimum size=30pt] (goal) {$\goal$};
         \node[round,right=20mm of goal, minimum size=30pt] (yj) {$y_j$};
    \node[round,below=15mm of yj,minimum size=30pt] (s) {$s$};
   
    \node[round,above=10mm of goal, minimum size=30pt] (fail2) {$\fail$};

  \draw[color=black , very thick,->] (xj) edge  node [ very near start, anchor=center] (h1) {} node [pos=0.7,left=5pt] {$1-(\frac{1}{2k^{2(k-j)}}+\beta_j)$} (fail2) ;
  \draw[color=black , very thick,->] (xj)  edge node [very near start, anchor=center] (h2) {} node [pos=0.5,below=2pt] {$\frac{1}{2k^{2(k-j)}}+\beta_j$} (goal) ;
  \draw[color=black , very thick] (h1.center) edge [bend left=55] node [pos=0.25,above=2pt] {} (h2.center);
  
  \draw[color=black , very thick, ->] (t) edge [bend left=25] node [pos=0.2,right=2pt] {$\gamma_j|\wgt:+k-j$} (xj);

    \draw[color=black , very thick,->] (yj) edge  node [ very near start, anchor=center] (g1) {} node [pos=0.7,right=5pt] {$1-\frac{1}{2k^{2(k-j)}}$} (fail2) ;
  \draw[color=black , very thick,->] (yj)  edge node [very near start, anchor=center] (g2) {} node [pos=0.5,below=2pt] {$\frac{1}{2k^{2(k-j)}}$} (goal) ;
  \draw[color=black , very thick] (g1.center) edge [bend right=55] node [pos=0.25,above=2pt] {} (g2.center);
  
  \draw[color=black , very thick, ->] (s) edge [bend right=25] node [pos=0.2,left=2pt] {$\delta_j|\wgt:+k-j$} (yj);

\end{tikzpicture}
}
         
         \caption{The gadget contains the depicted states and actions for each $0\leq j \leq k-1$.
        }
         \label{fig:initial_values}
     \end{subfigure}
      \vspace{12pt}
        \caption{The gadget (a) encoding  the linear recurrence relation in all reductions and (b) encoding the intial values in the reduction to the partial SSPP.    }
        \label{fig:gadgets}
        \vspace{-4pt}
\end{figure}

Now, we construct  an MDP-gadget with an example depicted in Figure \ref{fig:linear_recurrence}.
This gadget  contains states $\goal$, $s$, and $t$, as well as $s_1,\dots, s_k$ and $t_1,\dots, t_k$. In state $t$, an action $\gamma$ is enabled which has weight $0$ and leads to state $t_i$ with probability $\alpha_i$ if $\alpha_i>0$ and to state $s_i$ with probability $|\alpha_i|$ if $\alpha_i<0$ for all $i$. The remaining probability leads to $\goal$. From each state $t_i$, there is an action leading to $t$ with weight $-i$. The action $\delta$ enabled in $s$ as well as the actions leading from states $s_i$ to $s$ are constructed in the same way.
This gadget will be integrated into a larger MDP where there are no other outgoing edges from states $s_1,\dots,s_k,t_1,\dots,t_k$.
Now, for each state $q$ and each integer $w$, let $e(q,w)$ be the optimal partial expectation when starting in state $q$ with accumulated weight $w$. Further, let $d(w)=e(t,w)-e(s,w)$. The simple proof of the following lemma can be found in Appendix \ref{app:hardness_threshold_PE} and uses that optimal partial expectations satisfy that $e(q,w)=\sum_r P(q,\alpha,r) e(r,w{+}\wgt(q,\alpha))$ if an optimal scheduler chooses action $\alpha$ in state $q$ when the accumulated weight is $w$.
\begin{mylem} \label{lem:PE_recurrence}
 Let $w\in \mathbb{Z}$. If an optimal scheduler chooses action $\gamma$ in $t$ and $\delta$ in $s$ if the accumulated weight is $w$, then 
$d(w)=\alpha_1 d(w-1) + \dots + \alpha_k d(w-k)$.
\end{mylem}

Now we construct a gadget that encodes the initial values $\beta_0,\dots,\beta_{k{-}1}$.
The gadget  is depicted in Figure \ref{fig:initial_values} and contains states $t$, $s$, $\goal$, and $\fail$. For each $0\leq j\leq k-1$, it additionally contains states $x_j$ and $y_j$. In state $x_j$, there is one action enabled that leads to $\goal$ with probability $\frac{1}{2k^{2(k-j)}}+\beta_j$ and to $\fail$ otherwise. From state $y_j$, $\goal$ is reached with probability $\frac{1}{2k^{2(k-j)}}$ and $\fail$ otherwise. In state $t$, there is an action $\gamma_j$ leading to $x_j$ with weight $+k-j$ for each $0\leq j\leq k-1$. Likewise, in state $s$ there is an action $\delta_j$ leading to $y_j$ with weight $k{-}j$ for each $0\leq j\leq k-1$.
We now glue together the two gadgets  at states $s$, $t$, and $\goal$. The  cumbersome choices of probability values lead to the following lemma  via straight-forward computations  presented in Appendix~\ref{app:hardness_threshold_PE}.

\begin{mylem}\label{lem:optimal_scheduler}
Let $0\leq j \leq k-1$. Starting with  weight $-(k{-}1){+}j$ in state $t$ or $s$,  action $\gamma_j$ and $\delta_j$  maximize the partial expectation.  For positive starting weight, $\gamma$ and $\delta$ are optimal.
\end{mylem}

Comparing action $\gamma_j$ and $\delta_j$ for starting weight $-(k{-}1){+}j$, we conclude that the difference between optimal values $d(-(k{-}1){+}j)$ is equal to $\beta_j$, for $0\leq j\leq k-1$,  
and hence  $d({-}(k{-}1)+n)=u_n$ for all $n$.
Finally, we equip the MDP with a simple initial gadget (see Appendix \ref{app:hardness_threshold_PE}): From the initial state $\sinit$, one action with weight $+1$ is enabled. This action leads to a state $c$ with probability $\frac{1}{2}$ and loops back to $\sinit$ with probability $\frac{1}{2}$. In $c$, the decision between action $\tau$  leading to state $t$ and action $\sigma$ leading to state $s$ has to be made. So for any $n>0$, state $c$ is reached with accumulated weight $n$ with positive probability. An optimal scheduler now has to decide whether the partial expectation when starting with weight $n$ is better in state $s$ or $t$: 
 Action $\tau$ is optimal in $c$ for accumulated weight $w$ if and only if $d(w)\geq 0$.
Further, the scheduler $\sched$ always choosing $\tau$ in $c$ and actions $\gamma, \gamma_0, \dots, \gamma_{k-1}, \delta, \dots$ as described in Lemma \ref{lem:optimal_scheduler} is optimal iff the given linear recurrence sequence is non-negative.
We can compute the partial expectation of scheduler $\sched$ in the constructed MDP. The partial expectation turns out to be a rational. Hence, using this partial expectation as the threshold $\vartheta$, we obtain the first main result. The technical proof computing the  value of $\sched$ in the constructed MDP 
 is given in Appendix \ref{app:hardness_threshold_PE}.

\begin{mythm}\label{thm:threshold_PE}
The Positivity problem is polynomial-time reducible to the following problem: Given an MDP $\cM$ and a rational $\vartheta$, decide whether ${\PE}^{\max}_{\cM}>\vartheta$.
\end{mythm}

 \begin{myrem}\label{rem:strict}
 There is no obvious way to adjust the construction such that  the Skolem-hardness of the question whether  ${\PE}^{\max}_{\cM}\geq \vartheta$ would follow. One attempt would be to provide an $\varepsilon$ such that  ${\PE}^{\max}_{\cM}>\vartheta$ iff  ${\PE}^{\max}_{\cM}\geq \vartheta + \varepsilon$. This, however, probably requires a bound on the position at which the given linear recurrence sequence first becomes negative. But this question lies at the core of the positivity and the Skolem problem. All Skolem-hardness results in this paper hence concern only threshold problems with strict inequality.
 \end{myrem}


The Skolem-hardness of the threshold problem for the conditional SSPP is  obtained by a simple reduction showing that 
the threshold problems of the partial SSPP is polynomial-time reducible to the threshold problem of the  conditional SSPP (see Lemma \ref{lem_app:pe_ce_inter-reducible} in Appendix \ref{app:hardness_threshold_PE}).

\begin{mythm}\label{thm:threshold_CE}
The Positivity problem is reducible in polynomial time to the following problem: Given an MDP $\cM$ and a rational $\vartheta$, decide whether ${\CE}^{\max}_{\cM}>\vartheta$.
\end{mythm}

\subsection{Conditional value-at-risk for the classical SSPP}\label{sub:Skolem_cvar}

We reuse the gadget depicted in Figure \ref{fig:linear_recurrence} to prove the following result:

\begin{mythm}\label{thm:Skolem_cvar}
The Positivity problem is polynomial-time reducible to the following problem:
Given an MDP $\cM$ and  rationals $\vartheta$ and  $p\in(0,1)$, decide whether $\CVaR^{\max}_p (\rawdiaplus \goal)> \vartheta$.
\end{mythm}

We begin by the following consideration: Given an MDP $\cM$ with initial state $\sinit$,
we construct a new MDP $\cN$.
We add a new initial state $\sinit^\prime$. In $\sinit^\prime$, there is only one action with weight $0$  enabled 
leading to $\sinit$ with probability $\frac{1}{3}$ and to $\goal$ with probability $\frac{2}{3}$. So, at least two thirds of the paths accumulate weight $0$ before reaching the goal. 
Hence, we can already say that $\VaR^\sched_{1/2}(\rawdiaplus \goal)=0$ in $\cN$ under any scheduler $\sched$. Note that schedulers for $\cM$ can be seen as schedulers for $\cN$ and vice versa.
This considerably simplifies the computation of the conditional value-at-risk in $\cN$. Define the random variable 
$
\rawdiaminus \goal (\zeta)
$
to be 
$\rawdiaplus \goal (\zeta)$ if $ \rawdiaplus \goal \leq 0$ and to be $0$ otherwise.
Now, the conditional value-at-risk for the probability value $1/2$ under a scheduler $\sched$ in $\cN$ is  given by 
$
\CVaR^\sched_{1/2}(\rawdiaplus \goal)= 2 \cdot \mathbb{E}^\sched_{\cN,\sinit}(\rawdiaminus \goal) = \frac{2}{3}\cdot \mathbb{E}^\sched_{\cM,\sinit}(\rawdiaminus \goal) 
$.
So, the result follows from the following lemma:

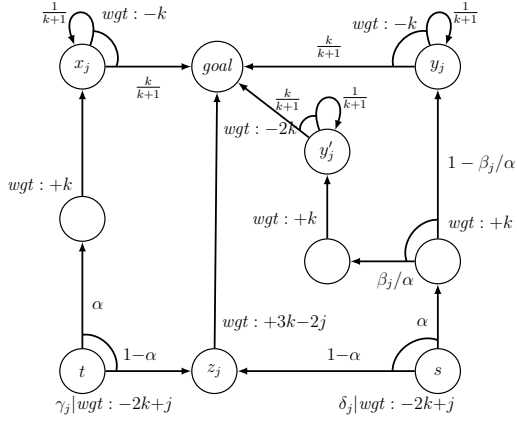
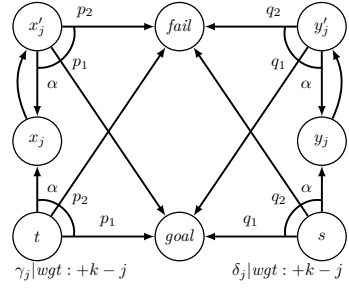
\begin{figure}[t]
     \centering
     \begin{subfigure}[b]{0.47\textwidth}
         \centering

          \resizebox{1.06\textwidth}{!}{
          
\begin{tikzpicture}[scale=1,auto,node distance=8mm,>=latex]
\Large
    \tikzstyle{round}=[thick,draw=black,circle]

    \node[round,minimum size=30pt] (t) {$t$};
      \node[round,above=25mm of t, minimum size=30pt] (xj1) {};
    \node[round,above=25mm of xj1, minimum size=30pt] (xj) {$x_j$};
 
  \node[round, right=20mm of t,minimum size=30pt] (zj) {$z_j$};

    \node[round, right=20mm of xj ,minimum size=30pt] (goal) {$\goal$};
         \node[round,right=40mm of goal, minimum size=30pt] (yj) {$y_j$};
         \node[round,below=35mm of yj,minimum size=30pt] (yj1) {};
         \node[round,left=15mm of yj1,minimum size=30pt] (yj3) {};
    \node[round,below=15mm of yj1,minimum size=30pt] (s) {$s$};
   
    \node[round,above=15mm of yj3, minimum size=30pt] (yj2) {$y_j^\prime$};

  \draw[color=black , very thick,->,loop, out=70, in=110, min distance=30pt] (xj) edge  node [ very near start, anchor=center] (h1) {} node [pos=0.4,right=5pt] {$\wgt: {-}k$}   node [pos=0.5,left=5pt] {$\frac{1}{k+1}$} (xj) ;
  \draw[color=black , very thick,->] (xj)  edge node [very near start, anchor=center] (h2) {} node [pos=0.5,below=2pt] {$\frac{k}{k{+}1}$} (goal) ;
  \draw[color=black , very thick] (h1.center) edge [bend left=55] node [pos=0.25,above=2pt] {} (h2.center);
  
  \draw[color=black , very thick, ->] (t) edge  node [pos=0.4,right=2pt] {$\alpha$} node [very near start, anchor=center] (w1) {} (xj1);
  \draw[color=black , very thick, ->] (t) edge  node [pos=0.4,above=2pt] {$1{-}\alpha$} node [very near start, anchor=center] (w2) {}   node [pos=0.1,below=12pt] {$\gamma_j|\wgt:{-}2k{+}j$}  (zj);
   \draw[color=black , very thick] (w2.center) edge [bend right=55] node [pos=0.25,above=2pt] {} (w1.center);

    \draw[color=black , very thick, ->] (xj1) edge  node [pos=0.1,left=2pt] {$\wgt:+k$} (xj);

      \draw[color=black , very thick, ->] (yj1) edge  node [pos=0.5,right=2pt] {$1-\beta_j/\alpha$} node [pos=0.1,right=2pt] {$\wgt:+k$} node [ very near start, anchor=center] (p1) {} (yj);
      \draw[color=black , very thick, ->] (yj1) edge  node [pos=0.3,below=2pt] {$\beta_j/\alpha$} node [very near start, anchor=center] (p2) {} (yj3);
        \draw[color=black , very thick] (p1.center) edge [bend right=55] node [pos=0.25,above=2pt] {} (p2.center);

   \draw[color=black , very thick, ->] (yj3) edge  node [pos=0.3,left=2pt] {$\wgt:+k$} (yj2);

    \draw[color=black , very thick, ->] (s) edge    node [pos=0.4,left=2pt] {$\alpha$} node [very near start, anchor=center] (v1) {} (yj1);
  \draw[color=black , very thick, ->] (s) edge  node [pos=0.4,above=2pt] {$1{-}\alpha$} node [very near start, anchor=center] (v2) {}  node [pos=0.1,below=12pt] {$\delta_j|\wgt:{-}2k{+}j$} (zj);
   \draw[color=black , very thick] (v1.center) edge [bend right=55] node [pos=0.25,above=2pt] {} (v2.center);

    \draw[color=black , very thick,->,loop, out=110, in=70, min distance=30pt] (yj) edge  node [ very near start, anchor=center] (g1) {}node [pos=0.2,left=2pt] {$\wgt: {-}k$}    node [pos=0.5,right=5pt] {$\frac{1}{k{+}1}$} (yj) ;
  \draw[color=black , very thick,->] (yj)  edge node [very near start, anchor=center] (g2) {} node [pos=0.5,above=1pt] {$\frac{k}{k{+}1}$} (goal) ;
  \draw[color=black , very thick] (g1.center) edge [bend right=55] node [pos=0.25,above=2pt] {} (g2.center);
  
     \draw[color=black , very thick,->,loop, out=110, in=70, min distance=30pt] (yj2) edge  node [ very near start, anchor=center] (q1) {} node [pos=0.0,left=10pt] {$\wgt: {-}2k$}   node [pos=0.5,right=5pt] {$\frac{1}{k{+}1}$} (yj2) ;
  \draw[color=black , very thick,->] (yj2)  edge node [very near start, anchor=center] (q2) {} node [pos=0.3,above=2pt] {$\frac{k}{k{+}1}$} (goal) ;
  \draw[color=black , very thick] (q1.center) edge [bend right=55] node [pos=0.25,above=2pt] {} (q2.center);

  \draw[color=black , very thick, ->] (zj) edge  node [pos=0.1,right=0pt] {$\wgt:{+}3k{-}2j$}  node [very near start, anchor=center] (v1) {} (goal);

\end{tikzpicture}
}
         
         \caption{The gadget contains the depicted states and actions for each $0\leq j \leq k-1$.
 $\alpha=\sum_{1= i}^{ k} |\alpha_i|$.
         }
         \label{fig:cvar_initial_values}
     \end{subfigure}
    \hfill
     \begin{subfigure}[b]{0.47\textwidth}
         \centering
               \resizebox{.7\textwidth}{!}{
\begin{tikzpicture}[scale=.9,auto,node distance=8mm,>=latex]
\large
    \tikzstyle{round}=[thick,draw=black,circle]

    \node[round,minimum size=30pt] (t) {$t$};
    \node[round,above=10mm of t, minimum size=30pt] (xj) {$x_j$};

    \node[round, right=20mm of t ,minimum size=30pt] (goal) {$\goal$};
     \node[round, right=20mm of goal ,minimum size=30pt] (s) {$s$};
     \node[round,above=10mm of s, minimum size=30pt] (yj) {$y_j$};

    \node[round,above=35mm of goal, minimum size=30pt] (fail2) {$\fail$};
    \node[round,left=20mm of fail2, minimum size=30pt] (xj2) {$x_j^\prime$};
    \node[round,right=20mm of fail2, minimum size=30pt] (yj2) {$y_j^\prime$};

  \draw[color=black , very thick,->] (t) edge   node [pos=0.10,right=2pt] {$p_2$} (fail2) ;
  \draw[color=black , very thick,->] (t)  edge node [very near start, anchor=center] (h2) {} node [pos=0.5,above=2pt] {$p_1$} node [pos=0.1,below=12pt] {$\gamma_j|\wgt:+k-j$} (goal) ;
   \draw[color=black , very thick, ->] (t) edge node [ near start, anchor=center] (h1) {} node [pos=0.5,right=2pt] {$\alpha$} (xj);
  \draw[color=black , very thick] (h1.center) edge [bend left=55] node [pos=0.25,above=2pt] {} (h2.center);
  
   \draw[color=black , very thick,->] (xj2) edge  node [very near start, anchor=center] (g2) {}  node [pos=0.25,above=2pt] {$p_2$} (fail2) ;
  \draw[color=black , very thick,->] (xj2)  edge node [pos=0.1,right=2pt] {$p_1$} (goal) ;
   \draw[color=black , very thick, ->] (xj2) edge node [ near start, anchor=center] (g1) {} node [pos=0.5,right=2pt] {$\alpha$} (xj);
  \draw[color=black , very thick] (g2.center) edge [bend left=55] node [pos=0.25,above=2pt] {} (g1.center);
  
  \draw[color=black , very thick, ->] (xj) edge [bend left=25] (xj2);

      \draw[color=black , very thick,->] (s) edge   node [pos=0.10,left=2pt] {$q_2$} (fail2) ;
  \draw[color=black , very thick,->] (s)  edge node [very near start, anchor=center] (j1) {} node [pos=0.5,above=2pt] {$q_1$}node [pos=0.1,below=12pt] {$\delta_j|\wgt:+k-j$} (goal) ;
   \draw[color=black , very thick, ->] (s) edge node [ near start, anchor=center] (j2) {} node [pos=0.5,left=2pt] {$\alpha$}  (yj);
  \draw[color=black , very thick] (j1.center) edge [bend left=55] node [pos=0.25,above=2pt] {} (j2.center);
  
   \draw[color=black , very thick,->] (yj2) edge  node [very near start, anchor=center] (k2) {}  node [pos=0.25,above=2pt] {$q_2$} (fail2) ;
  \draw[color=black , very thick,->] (yj2)  edge node [pos=0.1,left=2pt] {$q_1$} (goal) ;
   \draw[color=black , very thick, ->] (yj2) edge node [ near start, anchor=center] (k1) {} node [pos=0.5,left=2pt] {$\alpha$} (yj);
  \draw[color=black , very thick] (k1.center) edge [bend left=55] node [pos=0.25,above=2pt] {} (k2.center);
  
  \draw[color=black , very thick, ->] (yj) edge [bend right=25] (yj2);
    
\end{tikzpicture}

}
         
         \caption{The gadget contains the depicted states and actions for each $0\leq j \leq k-1$. The probabilities are:  $p_1=(1-\alpha) (\frac{1}{2k^{2(k-j)}}+\beta_j)$, $p_2=(1-\alpha)(1-(\frac{1}{2k^{2(k-j)}}+\beta_j))$,  $q_1=(1-\alpha) \frac{1}{2k^{2(k-j)}}$, $q_2=(1-\alpha)(1-\frac{1}{2k^{2(k-j)}})$.
 All actions except for $\gamma_j$ and $\delta_j$ have weight $0$.
        }
         \label{fig:wlf_initial_values}
     \end{subfigure}
     \vspace{12pt}
        \caption{The  gadgets encoding  initial values for (a) the conditional value-at-risk for the classical SSPP and (b) weighted long-run frequencies.}
        \label{fig:new_gadgets}
          \vspace{-4pt}
\end{figure}

\begin{mylem}\label{lem:Skolem_rawdiaminus}
The Positivity problem is polynomial-time reducible to the following problem:
Given an MDP $\cM$ and a rational $\vartheta$, decide whether $\mathbb{E}^{\max}_{\cM,\sinit} (\rawdiaminus \goal)> \vartheta$.
\end{mylem}

We  adjust the MDP used for the Skolem-hardness proof for the partial SSPP. 
So, let $k$ be a natural number, $\alpha_1,\dots,\alpha_k$ be rational coefficients of a linear recurrence sequence, and $\beta_0,\dots, \beta_{k-1}\geq 0$ the rational initial values. W.l.o.g. we again assume these values to be small, namely: $\sum_{1\leq i\leq k} |\alpha_i|\leq \frac{1}{5(k+1)}$ and for all $j$, $\beta_j\leq \frac{1}{3}\alpha$ where $\alpha=\sum_{1\leq i\leq k} |\alpha_i|$.

The first important observation is that the optimal expectation of $\rawdiaminus \goal$ for different starting states and starting weights behaves very similar to optimal partial expectations:
For each state $q$ and each integer $w$, let $e(q,w)$ be the optimal expectation of $\rawdiaminus \goal$ when starting in state $q$ with accumulated weight $w$. If an optimal scheduler chooses $\alpha$ when in $q$ with accumulated weight $w$, then $e(q,w)=\sum_{r\in S} P(q,\alpha,r)\cdot e(r,w{+}\wgt(q,\alpha))$. Reusing the MDP-gadget depicted in \ref{fig:linear_recurrence}, we observe that
if we again let $d(w)=e(t,w)-e(s,w)$, the following holds as before:
For any $w\in \mathbb{Z}$, if an optimal scheduler chooses action $\gamma$ in $t$ and $\delta$ in $s$ if the accumulated weight is $w$, then 
$d(w)=\alpha_1 d(w-1) + \dots+ \alpha_k d(w-k)$.

Now, we construct a new gadget that encodes the initial values of a linear recurrence sequence. 
The new gadget  is depicted in Figure \ref{fig:cvar_initial_values}. Besides the actions $\gamma_j$ and $\delta_j$ for $0\leq j \leq k-1$ there are no non-deterministic choices. 
Again, we glue together the two gadgets in states $s$, $t$, and $\goal$.
The main idea is that for non-negative starting weights in state $s$ or $t$ actions $\gamma_j$ and $\delta_j$ lead to a larger expected tail loss than actions $\gamma$ and $\delta$.
For $0\leq j\leq k{-}1$ and an accumulated weight ${-}k{+}j$ in state $t$ or $s$, the actions $\gamma_j$ and $\delta_j$ are, however, optimal for maximizing the expectation of $\rawdiaminus \goal$ sinve the goal is reached with non-negative weights with high probability under these actions (details in  Appendix \ref{app:cvar}). The difference of optimal values satisfies $e(t,{-}k+j)-e(s,{-}k+j)=\beta_j$ for $0\leq j \leq k{-}1$ again.
Finally, we add the same initial component as in the previous section and see that the scheduler $\sched$ always choosing $\tau$ in state $c$ is optimal iff the linear recurrence sequence stays non-negative. As the expectation of $\rawdiaminus \goal$ under $\sched$ is again a rational (see Appendix \ref{app:cvar}), this finishes the proof analogously to the previous section.


\subsection{Long-run probability and frequency-LTL} \label{sec:long_run_probabilities}

In order to transfer the Skolem-hardness results to long-run probabilities and frequency-LTL, we introduce the auxiliary notion of \emph{weighted long-run frequency}.
Let $\cM$ be an MDP with a weight function $\wgt: S\times \Act \to \mathbb{Z}$  and two disjoint  sets of states $\Goal,\Fail \subseteq S$.
On an infinite paths $\pi=s_0,\alpha_0,s_1,\dots$, we define the random variable $\mathit{wlf}$  as follows:
\[ \mathit{wlf} (\pi) = \liminf_{n\to \infty} \frac{1}{n+1}  \sum\nolimits_{i=0}^n \wgt(s_i,\alpha_i)\cdot \mathds{1}_{\pi[i\dots]\vDash \neg \Fail \Until \Goal}\]
where $\mathds{1}_{\pi[i\dots]\vDash \neg \Fail \Until \Goal}$ is $1$ if the suffix $\pi[i\dots]=s_i,\alpha_i,s_{i+1},\dots$ satisfies $\neg\Fail \Until \Goal$, and 0 otherwise.
Given a scheduler $\sched$, we  define the  weighted long-run frequency ${\PMP}_\cM^\sched= \mathbb{E}^\sched_{\cM} ( \mathit{wlf} )$ and ${\PMP}_\cM^{\max}=\sup_\sched {\PMP}_\cM^\sched$. 
This can be seen as a long-run average version of partial expectations. Weights are only received  if afterwards $\Goal$ is visited before $\Fail$ and we measure the average weight received per step according to this rule.
Note that we only consider the path property $\neg\Fail \Until \Goal$ in this paper and hence do not include this property in our notation and terminology. 
An illustrating example can be found in Appendix \ref{app:wlf}.

We  modifiy the MDP that was constructed in Section \ref{sub:Skolem_partial} for the Skolem-hardness of the partial SSPP. We replace the gadget encoding the initial values with the gadget  depicted in Figure \ref{fig:wlf_initial_values}.  This gadget differs from the gadget used for partial expectations only in the expected time it takes to reach $\goal$ or $\fail$ under $\gamma_j$ or $\delta_j$. It is constructed in a way such that the expected time to reach $\goal$ or $\fail$ from $\sinit$ does not depend on the scheduler.
 Finally, we add a transition leading back to the initial state from $\goal$ and $\fail$.
An optimal scheduler for weighted long-run frequencies in the constructed MDP $\cK$ now just has to maximize the  partial expectation leading to the Skolem-hardness result (for more details see Appendix \ref{app:wlf}).

\begin{mythm}\label{thm:Skolem_wlf}
The Positivity problem is polynomial-time reducible to the following problem: Given an MDP $\cM$ and a rational $\vartheta$, decide whether ${\PMP}^{\max}_{\cM}>\vartheta$.
\end{mythm}


This result now serves as a tool to establish analogous results for long-run probabilities. 
The key idea is to encode integer weights via a labelling of states and to use a simple regular co-safety property to mimic the reception of weights in weighted long-run frequencies. 
\begin{mythm}\label{thm:Skolem_LP}
The Positivity problem is polynomial-time reducible to the following problem:
Given an MDP $\cM$, an NFA $\cA$, and a rational $\vartheta$, decide whether  $\LP^{\max}_\cM (\cA)>\vartheta$.
\end{mythm}

\begin{wrapfigure}{R}{0.43\textwidth}

\scalebox{.65}{
\begin{tikzpicture}[scale=1,auto,node distance=8mm,>=latex]
    \tikzstyle{round}=[thick,draw=black,circle]

    \node[round,minimum size=15pt] (init) {};
    \node[color=white, left=10mm of init] (start) {};
    \node[round,right=30mm of init, minimum size=15pt] (a) {};
    
    \node[round, below=11mm of a,minimum size=15pt] (b) {};
    
     \node[round, below=11mm of b,minimum size=15pt] (c) {};
 
    \node[round, right=30mm of a, double,minimum size=15pt] (acc) {};

  \draw[color=black , very thick,->] (start) edge   (init) ;
  \draw[color=black , very thick,->] (init) edge node [pos=0.5, above] {$p\land \neg g \land \neg f$}  (a) ;
    \draw[color=black , very thick,->] (a) edge node [pos=0.4, above] {$g,  f \land c$}  (acc) ;
  
   \draw[color=black , very thick,->] (init) edge [bend right=25] node [pos=0.2, right=0mm] {$z\land \neg g\land \neg f$}  (b) ;
  \draw[color=black , very thick,->] (b) edge [bend right=25] node [pos=0.3, below=1mm] {$g \land c, f \land c$}  (acc) ;
  
    \draw[color=black , very thick,->] (init) edge [bend right=35] node [pos=0.2, left=0mm] {$n\land \neg g\land \neg f$}  (c) ;
  \draw[color=black , very thick,->] (c) edge [bend right=35] node [pos=0.2, below=2mm] {$ f \land c$}  (acc) ;
    
    \draw[color=black , very thick,->] (a) edge [loop above] node [pos=0.5, above=0mm] {$ \neg g \land \neg f$}  (a) ;
    \draw[color=black , very thick,->] (b) edge [loop above] node [pos=0.5, above=0mm] {$ \neg g \land \neg f$}  (b) ;
     \draw[color=black , very thick,->] (c) edge [loop above] node [pos=0.5, above=0mm] {$ \neg g \land \neg f$}  (c) ;

  \draw[color=black , very thick,->] (init) edge [bend left=45] node [pos=0.5, above=0mm] {$g \land p, g\land z \land c, f \land c$}  (acc) ;

\end{tikzpicture}
}
\caption{The NFA $\cA$ expressing a property of the form $d\lor \bigvee_{i=1}^3 (c_i\land (a \Until b_i))$. }\label{fig:NFA}
\end{wrapfigure}
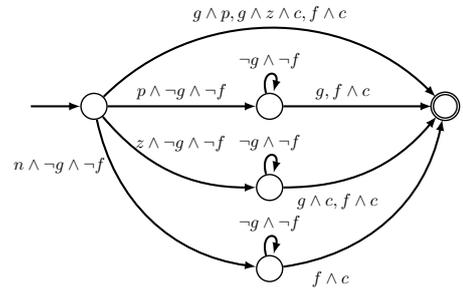

In the sequel, we consider weighted states instead of weighted state-action pairs. Further, we assume that the weights are only $-1$, $0$, and $+1$. This assumption leads to a pseudo-polynomial blow-up in the general case. The weights in the MDP $\cK$ constructed for Theorem \ref{thm:Skolem_wlf} above are, however, at most $k$. As the MDP has more than $2k$ states, transforming $\cK$ to weights $-1$, $0$, and $+1$ only leads to a polynomial blow-up.
As this MDP has no non-trivial end-components,  $\{\goal,\fail\}$ is visited infinitely often with probability $1$ under any scheduler.
Let $\AP=\{n,z,p,c,g,f\}$ be a set of atomic propositions representing \emph{negative} ($-1$), \emph{zero} ($0$), and \emph{positive} ($+1$) weight,  \emph{coin flip},  $\goal$, and $\fail$, respectively. We construct an MDP $\cL$: The states $\goal$ and $\fail$ are duplicated while one copy of each is labeled with $c$ and whenever $\goal$ or $\fail$ are entered in the MDP $\cK$, both of the two copies in $\cL$ are equally likely. For a formal definition see Appendix \ref{app:LRP}.
In Figure \ref{fig:NFA}, we depict the NFA $\cA$ used for the encoding. The NFA $\cA$ is constructed such that in $\cL$ any run starting in a state labeled \emph{zero} or reaching \emph{fail} before \emph{goal} is accepted with probability $1/2$ due to the \emph{coin flips}. A run starting in a state labeled \emph{positive}  and reaching \emph{goal} before \emph{fail} is accepted while such a path starting in a state labeled \emph{negative} is not. This leads to the following lemma that proves 
Theorem \ref{thm:Skolem_LP}.

\begin{mylem}\label{lem:encoding}
For the MDPs $\cK$ and $\cL$ constructed above, we have  
${\PMP}^{\max}_{\cK} = \frac{1}{2}+\frac{1}{2}\LP^{\max}_\cL (\cA)$.
\end{mylem}

\begin{proof}[Proof sketch]
It is quite easy to see that the claim holds for finite-memory schedulers as we can rely on steady state probabilities in the resulting Markov chain. That the supremum over all schedulers agrees with the supremum over finite-memory schedulers on both sides follows from Fatou's lemma. Details can be found in Appendix \ref{app:LRP}.
\end{proof}

 A consequence of this result is that model checking of frequency-LTL in MDPs is at least as hard as the Skolem problem. The decidability of the model-checking problem for the full logic frequency-LTL has been left open, but set as a goal in \cite{ForejtK15,ForejtKK15}. Obtaining this goal by proving the decidability of the model-checking problem hence would settle the decidability of the Skolem problem. The frequency-globally modality $G^{>\vartheta}_{\inf}(\varphi)$ is defined to hold on a path $\pi$ iff $\liminf_{n\to \infty} \frac{1}{n+1} \sum_{i=0}^n \mathds{1}_{\pi[i\dots]\vDash\varphi}>\vartheta$, i.e. iff the long-run average number of positions at which a suffix satisfying $\varphi$ starts exceeds $\vartheta$.

\begin{mythm}\label{thm:fLTL}
There is a polynomial-time reduction from the Positivity problem to the following qualitative model checking problem for frequency LTL  for a fixed LTL-formula $\varphi$: Given an MDP $\cM$ and a rational $\vartheta$, is $\Pr^{\max}_\cM (G^{>\vartheta}_{\inf} (\varphi))=1$?
\end{mythm}

\begin{proof}[Proof sketch]
The proof uses the reduction to the threshold problem for the long-run probability of the co-safety property expressed by $\cA$.
This property  is captured by a simple LTL-formula $\varphi$ (see Figure \ref{fig:NFA}). For finite-memory schedulers $\sched$ inducing a single bottom strongly connected component, we see that $G^{>\vartheta}_{\inf}(\varphi)$ holds with probability $1$ iff the expected long-run probability of $\varphi$  is greater than $\vartheta$.
That it is enough to consider such schedulers follows from the argument using  Fatou's lemma again. For more details see Appendix \ref{app:LRP}.
\end{proof}


\section{Saturation points} \label{sec:partial_mp}

Despite the inherent mathematical difficulty shown by the Skolem-hardness results so far,
 all of the problems studied here are solvable in exponential time under a natural restriction. For the problems on weighted MDPs, this restriction  only allows non-negative weights while for the long-run notions the restriction to constrained reachability properties ($a\Until b$) leads to solvability.
For the  partial and the conditional SSPP \cite{tacas2017,chen2013} and for  long-run probabilities \cite{lics2019}, the computability of optimal values under these restrictions has been shown.
The algorithms exploit the existence of \emph{saturation points}, a bound on the accumulated weight or the consecutive visits to certain states before optimal schedulers can behave memorylessly. We will extend this picture by providing a simple saturation point  for the computation of the optimal conditional value-at-risk for the classical SSPP in MDPs with non-negative weights. Afterwards, we transfer the saturation-point algorithm from \cite{lics2019} to weighted long-run frequencies in the setting of non-negative weights. As a consequence, we obtain an exponential-time algorithm for the qualitative model-checking problem of a frequency-LTL formula for which no solutions were known.
To conclude the section, we provide accompanying PSPACE lower bounds for the partial SSPP and weighted long-run frequencies with non-negative weights.

\subsection{Conditional value-at-risk for the classical SSPP}\label{cvar_pos}

Let $\cM$ be an MDP with non-negative weights.
In the classical SSPP, it is decidable in polynomial time whether the optimal expected accumulated weight before reaching the goal is bounded. If this is the case, the usual preprocessing step  removes end components \cite{deAlfaro1999,lics2018} and transforms the MDP such that exactly the schedulers reaching the goal with probability $1$ can be mimicked in the transformed MDP. So in the sequel, we assume that the absorbing state $\goal$ forms the only end component.
Given  a rational probability value $p\in (0,1)$, we are interested in the value $\CVaR^{\max}_p (\rawdiaplus \goal)$.
Note that in our formulation the worst outcomes are the paths with the lowest accumulated weight before reaching the goal.  Below we will  sketch how to treat the case where  high outcomes are considered bad.

\begin{mythm}\label{thm:cvar_pos}
Given an MDP $\cM=(S,\sinit,\Act,P,\wgt,\goal)$ with non-negative weights and no end-components except for one absorbing state $\goal$ as well as a  rational probability value $p\in(0,1)$, the value  $\CVaR^{\max}_p (\rawdiaplus \goal)$ is computable  in pseudo-polynomial time.
\end{mythm}

\begin{proof}[Proof sketch]
As there are no end components, we can provide a \emph{saturation point} $K\in\mathbb{N}$ such that paths accumulate a weight of more than $K$ with probability less than $1-p$. Then, paths reaching an accumulated weight of $K$  do  not belong to the worst $p$ outcomes. We  construct an MDP with the state space $S\times\{0,\dots, K\}$ that encodes the accumulated weight of a path up to $K$. Letting states of the form $(\goal,i)$ be terminal with weight $i$ and of the form $(s,K)$ be terminal with weight $K$, we can then rely on the algorithm computing the conditional value-at-risk for weighted reachability in \cite{kretinsky2018}. As $K$ can be chosen of pseudo-polynomial size and this algorithm runs in  time polynomial in the size of the constructed MDP, this leads to a pseudo-polynomial time algorithm. For details see Appendix \ref{cvar_pos}.
\end{proof}

Note that the behavior of a scheduler on paths with accumulated weight above $K$ does not matter at all for the conditional value-at-risk.
If we want to consider the case where long paths are considered as bad, we can multiply all weights by $-1$  and use the definitions as before. The idea here now is to compute a saturation point $-K$ such that the probability for a path to accumulate weight less than $-K$ is smaller than $p$. So, we know that a path with weight less than $-K$ belongs to the $p$ worst paths. On these paths, the best thing to do in order to maximize the conditional value-at-risk is to maximize the expected accumulated weight before reaching the goal. This can be done by a memoryless deterministic scheduler simultaneously for all states and the values are computable in polynomial time \cite{deAlfaro1999}. Then we construct the MDP $\cN$ as above but change the terminal weights as follows: states of the form $(\goal,i)$ get weight $-i$ and states of the form $(s,K)$ get weight $-K+\mathbb{E}^{\max}_{\cM,s}(\rawdiaplus \goal)$ where $\cM$ is the MDP in which all weights are already multiplied by $-1$. Afterwards the problem can be solved by the techniques for weighted reachability from \cite{kretinsky2018} again.

\subsection{Weighted long-run frequencies and frequency-LTL} \label{sec:PMP}

The existence of a saturation point for long-run probabilities of constrained reachability properties was shown in \cite{lics2019}. This result can easily be adapted to weighted long-run frequencies following the same arguments.
First, it is shown by an application of  Fatou's lemma that optimal weighted long-run frequency can be approximated by finite-memory schedulers. 
Afterwards, it is shown that the memory needed for the optimization can be restricted further: A {saturation point} $K\in \mathbb{N}$ is provided such that only scheduler keeping track of the accumulated weight up to $K$ have to be considered.
The adaptions necessary to the proof  in \cite{lics2019} are worked out in Appendix~\ref{app:PMP} and lead to the following result:

\begin{mythm}\label{thm:comp_wlf}
The maximal value ${\PMP}^{\max}_\cM$ in an MDP ${\cM}$ with non-negative weights is computable in pseudo-polynomial time.
\end{mythm}

\begin{mycor}\label{cor:frequency-LTL}
Given an MDP $\cM$ and a rational $\vartheta$, it can be checked in pseudo-polynomial time whether $\Pr^{\max}_{\cM} (G^{>\vartheta}_{\inf} (a\Until b))=1$.
\end{mycor}
\begin{proof}
The semantics of $G^{>\vartheta}_{\inf} (\neg \Fail \Until \Goal)$ on a path $\pi$ agree with the semantics of $\mathit{wlf}(\pi)>\vartheta$ if all weights are set to $+1$. Now, we can check for each end component $\cE$ of $\cM$ whether ${\PMP}^{\max}_\cE>\vartheta$. If that is the case, there is a finite memory scheduler for $\cE$ inducing only one BSCC  achieving a weighted long-run frequency greater than $\vartheta$. Under this scheduler almost all paths $\pi$ satisfy $\mathit{wlf}(\pi)>\vartheta$. Afterwards, it remains to check whether end components with such a scheduler can be reached with probability $1$ in $\cM$.
\end{proof}

In \cite{ForejtKK15}, the fragment of frequency-LTL in which no until-operators occur in the scope of a globally operator has been studied. The formula in the corollary is hence of the simplest form of frequency-LTL formulas for which no solution to the qualitative model-checking problem has been known. Remarkably, the formula used in the Skolem-hardness proof (Theorem \ref{thm:fLTL}) is only slightly more complicated as it  contains a Boolean combination of constrained reachability properties and atomic propositions under the frequency-globally operator.

\subsection{PSPACE lower bounds}\label{sec:pe_ce}

For the conditional SSPP with non-negative weights \cite{tacas2017} and the long-run probability of constrained reachability properties \cite{lics2019}, PSPACE and NP lower bounds, respectively, are known indicating that the pseudo-polynomial time algorithms for the computation can probably not be significantly improved.
The threshold problem of the conditional SSPP is already PSPACE-hard in acyclic MDPs with non-negative weights as shown in \cite{tacas2017}.
In \cite{fossacs2019}, it has been shown 
that  the threshold problem of the conditional SSPP is polynomial-time reducible to the threshold problem for the partial SSPP. This reduction generates an MDP with  negative weights, even when all weights in the original MDP are non-negative.
Here, we provide a new polynomial reduction for acyclic MDPs from the threshold problem for the conditional SSPP to the threshold problem of the partial SSPP that preserves the non-negativity of weights (see Appendix~\ref{app:PE_CE}). 

\begin{myprop} \label{prop:CE_to_PE}
The threshold problem of the  partial SSPP is PSPACE-hard in acyclic MDPs with non-negative weights. {It is contained in PSPACE for acyclic MDPs with arbitrary integer weights.}
\end{myprop}

In an acyclic MDP, we can add intermediate states on transitions such that all paths have the same length $\ell$. If we additionally add transitions form $\goal$ and $\fail$ back to the initial state, the maximal weighted long-run frequency is just the maximal partial expectation divided by $\ell$. This allows us to conclude:

\begin{myprop}\label{prop:PE_to_WLF}
The threshold problem for weighted long-run frequencies, ``Does $\PMP^{\max}_\cM\bowtie \vartheta$ hold?'', in MDPs with non-negative weights is PSPACE-hard.
\end{myprop}


\section{Conclusion} \label{sec:conclusion}

We identified a variety of optimization problems -- some of which seemed rather unrelated on first sight -- with a Skolem-hard threshold problem on MDPs.
The results show that an algorithm for the exact solution to these optimization problems would imply a major breakthrough. For the partial and conditional SSPP, however, approximation algorithms were provided in \cite{fossacs2019}.   Investigating the possibility to approximate optimal values might lead to  algorithms useful in practice  for the other objectives studied here.
Further,  the problems  have a pseudo-polynomial solution under  natural restrictions. The key result, the existence of a saturation point, has been established  in the setting of stochastic multiplayer games for partial expectations \cite{chen2013}. This raises the question to which extend the saturation point results for the other problems can be transferred to stochastic multiplayer games.

To the best of our knowledge, the  conditional value-at-risk for accumulated weights has not been addressed before. While we showed Skolem-hardness in the general setting,  the computation of the optimal value is possible in exponential time in the setting of non-negative weights. Studying 
lower bounds for the complexity of the threshold problem and the combination of constraints on the expected accumulated weight before reaching the goal, the value-at-risk, and the conditional value-at-risk in this setting are left as future work.


\bibliographystyle{abbrv}
\bibliography{references}



\newpage
\appendix
\noindent
\begin{Huge}\bf Appendix \end{Huge}


\section{Skolem-hardness: partial and conditional SSPP} \label{app:hardness_threshold_PE}

\begin{figure}[ht]
\begin{center}
\resizebox{1\textwidth}{!}{%

\begin{tikzpicture}[scale=1,auto,node distance=8mm,>=latex]
    \tikzstyle{round}=[thick,draw=black,circle]

    \node[round,minimum size=30pt] (t) {$t$};
    \node[round,below=15mm of t, minimum size=30pt] (t1) {$t_1$};
    \node[round,below=10mm of t1, minimum size=30pt] (t2) {$t_2$};
    
    \node[round, above=15mm of t,xshift=29mm ,minimum size=30pt] (goal) {$\goal$};
        
    \node[round,right=50mm of t,minimum size=30pt] (s) {$s$};
    \node[round,below=15mm of s, minimum size=30pt] (s1) {$s_1$};
    \node[round,below=10mm of s1, minimum size=30pt] (s2) {$s_2$};
    
\node[round,above=15mm of t, minimum size=30pt] (xj) {$x_j$};
    \node[round,above=15mm of goal, minimum size=30pt] (fail1) {$\fail$};

    \node[round,above=15mm of s, minimum size=30pt] (yj) {$y_j$};

     \node[round,below=80mm of goal, minimum size=30pt] (c) {$c$};
          \node[round,below=15mm of c, minimum size=30pt] (sinit) {$\sinit$};

  \draw[color=black , very thick,->] (t) edge  node [ very near start, anchor=center] (h1) {} node [pos=0.1,right=5pt] {$1-|\alpha_1|-|\alpha_2|$} (goal) ;
  \draw[color=black , very thick,->] (t)  edge node [near start, anchor=center] (h2) {} node [pos=0.7,right=2pt] {$|\alpha_1|$} (t1) ;
  \draw[color=black , very thick,->] (t)  edge node [pos=0.2,right=5pt] {$|\alpha_2|$} (s2) ;
  \draw[color=black , very thick] (h1.center) edge [bend left=55] node [pos=0.25,right=2pt] {$\gamma$} (h2.center);
  
 \draw[color=black , very thick, ->] (t1) edge [bend left=25] node [pos=0.1,left=2pt] {$\wgt:-1$} (t);
  \draw[color=black , very thick, ->] (t2) edge [bend left=75] node [pos=0.1,left=2pt] {$\wgt:-2$} (t);

  \draw[color=black , very thick,->] (s) edge  node [ very near start, anchor=center] (g1) {} node [pos=0.4,left=5pt] {$1-|\alpha_1|-|\alpha_2|$} (goal) ;
  \draw[color=black , very thick,->] (s)  edge node [near start, anchor=center] (g2) {} node [pos=0.7,left=2pt] {$|\alpha_1|$} (s1) ;
  \draw[color=black , very thick,->] (s)  edge node [pos=0.2,left=5pt] {$|\alpha_2|$} (t2) ;
  \draw[color=black , very thick] (g1.center) edge [bend right=55] node [pos=0.25,left=2pt] {$\delta$} (g2.center);
  
 \draw[color=black , very thick, ->] (s1) edge [bend right=25] node [pos=0.1,right=2pt] {$\wgt:-1$} (s);
  \draw[color=black , very thick, ->] (s2) edge [bend right=75] node [pos=0.1,right=2pt] {$\wgt:-2$} (s);  

  \draw[color=black , very thick,->] (xj) edge  node [ very near start, anchor=center] (h1) {} node [pos=0.5,left=5pt] {$1-(\frac{1}{2k^{2(k-j)}}+\beta_j)$} (fail1) ;
  \draw[color=black , very thick,->] (xj)  edge node [very near start, anchor=center] (h2) {} node [pos=0.4,below=2pt] {$\frac{1}{2k^{2(k-j)}}+\beta_j$} (goal) ;
  \draw[color=black , very thick] (h1.center) edge [bend left=55] node [pos=0.25,above=2pt] {} (h2.center);
  
  \draw[color=black , very thick, ->] (t) edge [bend left=25] node [pos=0.1,left=2pt] {$\gamma_j|\wgt:+k-j$} (xj);

    \draw[color=black , very thick,->] (yj) edge  node [ very near start, anchor=center] (g1) {} node [pos=0.5,right=5pt] {$1-\frac{1}{2k^{2(k-j)}}$} (fail1) ;
  \draw[color=black , very thick,->] (yj)  edge node [very near start, anchor=center] (g2) {} node [pos=0.2,below=2pt] {$\frac{1}{2k^{2(k-j)}}$} (goal) ;
  \draw[color=black , very thick] (g1.center) edge [bend right=55] node [pos=0.25,above=2pt] {} (g2.center);
  
  \draw[color=black , very thick, ->] (s) edge [bend right=25] node [pos=0.1,right=2pt] {$\delta_j|\wgt:+k-j$} (yj);

 \draw[color=black , very thick,->] (sinit) edge  node [  near start, anchor=center] (f1) {} node [pos=0.5,left=5pt] {$\frac{1}{2}$} (c) ;
  \draw[color=black , very thick,->] (sinit)  edge [loop, min distance=20mm, out=60, in=0] node [pos=.1, anchor=center] (f2) {} node [pos=0.4,right=7pt] {$\frac{1}{2}$} (sinit) ;
  \draw[color=black , very thick] (f1.center) edge [bend left=55] node [pos=1,xshift=2pt,above=10pt] {$\wgt: +1$} (f2.center);
  
   \draw[color=black , very thick, ->] (c) edge [loop, in=180, out=180, min distance=60mm] node [pos=0.1,above=5pt] {$\tau$} (t);
   
   \draw[color=black , very thick, ->] (c) edge [loop, in=0, out=0, min distance=60mm] node [pos=0.1,above=5pt] {$\sigma$} (s);

\end{tikzpicture}

 }
   \end{center}

\caption{The MDP contains the upper part for all $0\leq j \leq k-1$. The middle part is depicted for $k=2$, $\alpha_1\geq 0$, and $\alpha_2<0$.}\label{fig:MDP}
\end{figure}

We provide proofs to Section \ref{sub:Skolem_partial}. 
Given a linear recurrence sequence, we construct the MDP depicted in Figure \ref{fig:MDP} in which the sequence is encoded in terms of optimal partial expectations.
So let $k$ be a natural number and let $(u_n)_{n\geq 0}$ be the linear recurrence sequence given by $\alpha_i$ for $1\leq i \leq k$ and $\beta_j$ for $0\leq j \leq k{-}1$ be rational numbers. 
As
$u_{n+k} = \alpha_1 u_{n+k-1} + \dots + \alpha_k u_n $
for all $n$, we see that for any positive $\lambda\in \mathbb{Q}$ the sequence $(v_n)_{n\geq 0}$ defined by $v_n = \lambda^{n+1} u_n$ satisfies
$v_{n+k} = \lambda^1  \alpha_1 v_{n+k-1} + \dots + \lambda^k \alpha_k v_n$ for all $n$. Furthermore, $v_n$ is non-negative if and only if $u_n$ is.
W.l.o.g., we hence can assume that   $\sum_i |\alpha_i|<\frac{1}{4}$ and that $0\leq \beta_j< \frac{1}{4k^{2k+2}}$ for all $j$. 
The necessary rescaling is polynomial: Let $\mu$ be the maximum of $|\alpha_1|,\dots, |\alpha_k|,|\beta_0|,\dots, |\beta_{k{-}1}|$. Then choosing  $\lambda=\frac{1}{\mu\cdot 4k^{2k+2}}$ for the rescaling involves $\lambda^k$ as a factor. But this factor still has a polynomial size binary representation as the original input consists of $2k$ numbers one of which is $\mu$.

For each state $q$ and each integer $w$, let $e(q,w)$ be the optimal partial expectation when starting in state $q$ with accumulated weight $w$. Further, let $d(w)=e(t,w)-e(s,w)$. We use that optimal partial expectations satisfy that $e(q,w)=\sum_r P(q,\alpha,r) e(r,w{+}\wgt(q,\alpha))$ if an optimal scheduler chooses action $\alpha$ in state $q$ when the accumulated weight is $w$.
\begin{mylem} [Lemma \ref{lem:PE_recurrence}]\label{lem_app:PE_recurrence}
 Let $w\in \mathbb{Z}$. If an optimal scheduler chooses action $\gamma$ in $t$ and $\delta$ in $s$ if the accumulated weight is $w$, then 
$d(w)=\alpha_1 d(w-1) + \dots \alpha_k d(w-k)$.
\end{mylem}
\begin{proof}
If an optimal scheduler chooses action $\gamma$ and $\delta$, respectively, if the accumulated weight is $w$, we get
\begin{align*}
e(t,w)&=(1{-}|\alpha_1|{-}\ldots{-}|\alpha_k|)w+ \sum_{i:\alpha_i\geq 0} \alpha_i e(t,w-i) + \sum_{i:\alpha_i< 0} -\alpha_i e(s,w-i) \text{ and }\\
e(s,w)&=(1{-}|\alpha_1|{-}\ldots{-}|\alpha_k|)w+ \sum_{i:\alpha_i\geq 0} \alpha_i e(s,w-i) + \sum_{i:\alpha_i< 0} -\alpha_i e(t,w-i) .
\end{align*}
Now, it is straight-forward to compute
\begin{align*}
d(w) &= e(t,w)-e(s,w) \\
&= \sum_{i:\alpha_i\geq 0} \alpha_i (e(t,w{-}i)-e(s,w{-}i)) + \sum_{i:\alpha_i< 0} -\alpha_i (e(s,w{-}i)-e(t,w{-}i))\\
&= \sum_{1\leq i\leq k} \alpha_i d(w-i). \qedhere
\end{align*}
\end{proof}

\begin{mylem}[Lemma \ref{lem:optimal_scheduler}]\label{lem_app:optimal_scheduler}
Let $0\leq j \leq k-1$. Starting with accumulated weight $-(k-1)+j$ in state $t$, the action $\gamma_j$ maximizes the partial expectation. Likewise, $\delta_j$ is optimal when starting in $s$ with weight $-(k-1)+j$. For positive starting weight, $\gamma$ and $\delta$ are optimal.
\end{mylem}

\begin{proof}
Suppose action $\gamma_i$ is chosen in state $t$ when starting with weight $-(k-1)+j$. Then the partial expectation achieved from this situation is 
\[(1+j-i)(\frac{1}{2k^{2(k-i)}}+\beta_i).\]
For $i>j$ this value is $\leq 0$ and hence $\gamma_i$ is certainly not optimal. For $i=j$, we obtain a partial expectation of
 \[\frac{1}{2k^{2(k-j)}}+\beta_j.\]

For $i<j$, we know that state $x_i$ is reached with weight  $-(k-1)+j+(k-i)=1+j-i\leq k$. Further,  $\beta_i \leq \frac{1}{4k^{2k+2}}$ and  $\frac{1}{2k^{2(k-i)}}\leq \frac{1}{2k^{2(k-j)}\cdot k^2}$. So, the partial expectation obtained via $\gamma_i$ is at most
\[ \frac{k}{2k^{2(k-j)} \cdot k^2}  + \frac{k}{4k^{2k+2}}<\frac{1}{2k^{2(k-j)}}.\]

The argument for state $s$ is the same with $\beta_i=0$ for all $i$.

It is easy to see that for accumulated weight $-(k-1)+j$ with $0\leq j\leq k-1$ actions $\gamma$ or $\delta$ are not optimal in state $t$ or $s$: If $\goal$ is reached immediately, the weight is not positive and otherwise states $t$ or $s$ are reached with lower accumulated weight again. The values $\beta_j$ are chosen small enough such that also a switch from state $t$ to $s$ while accumulating negative weight does not lead to  a higher partial expectation.

For positive accumulated weight $w$, the optimal partial expectation when choosing $\gamma$ first is at least $\frac{3}{4}w$ by construction and the fact that a positive value can be achieved from any possible successor state. Choosing $\gamma_j$ on the other hands results in a partial expectation of at most $(k+w)\cdot ( \frac{1}{4k^{2k+2}} + \frac{1}{2k^2})$ which is easily seen to be less.

\end{proof}

In the proof, we see that the difference between optimal values $d(-(k{-}1){+}j)$ is equal to $\beta_j$, for $0\leq j\leq k-1$ and hence  $d({-}(k{-}1)+n)=u_n$ for all $n$.
So, the scheduler $\sched$ always choosing $\tau$ in $c$ and actions $\gamma, \gamma_0, \dots, \gamma_{k-1}, \delta, \dots$ as described in
 Lemma \ref{lem_app:optimal_scheduler} is optimal iff the given linear recurrence sequence is non-negative.

\begin{mythm}[Theorem \ref{thm:threshold_PE}]\label{app_thm:threshold_PE}
The positivity problem is reducible in polynomial time to the following problem: Given an MDP $\cM$ and a rational $\vartheta$, decide whether ${\PE}^{\max}_{\cM}>\vartheta$.
\end{mythm}

\begin{proof}
We will compute the partial expectation of scheduler $\sched$ always choosing $\tau$ in $c$ and actions $\gamma, \gamma_0, \dots, \gamma_{k-1}, \delta, \dots$ as described in
 Lemma \ref{lem_app:optimal_scheduler} in the constructed MDP $\cM$ depicted in Figure \ref{fig:MDP}: The scheduler $\sched$ chooses $\gamma$ and $\delta$, respectively, as long as the accumulated weight is positive.  For an accumulated weight of $-(k-1)+j$ for $0\leq j \leq k-1$, it chooses actions $\gamma_j$ and $\delta_j$, respectively. 

We want to recursively express the partial expectations under $\sched$ starting from $t$ or $s$ with some positive accumulated weight $n\in \mathbb{N}$ which we again denote by $e(t,n)$ and $e(s,n)$, respectively.
In order to do so, we consider the following Markov chain $\cC$ for $n\in \mathbb{N}$ (see Figure \ref{fig:Markov_chain}):

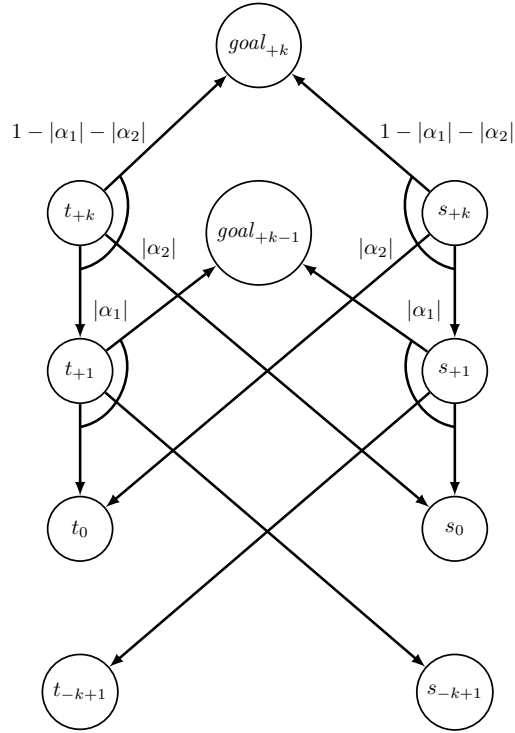
\begin{figure}[ht]
\begin{center}
\resizebox{.5 \textwidth}{!}{
\begin{tikzpicture}[scale=1,auto,node distance=8mm,>=latex]
    \tikzstyle{round}=[thick,draw=black,circle, minimum size= 6em]

    \node[round,minimum size=30pt] (t) {$t_{+k}$};
    \node[round,below=15mm of t, minimum size=30pt] (t1) {$t_{+1}$};
    \node[round,below=15mm of t1, minimum size=30pt] (t2) {$t_{0}$};
     \node[round,below=15mm of t2, minimum size=30pt] (t3) {$t_{-k+1}$};
    
    \node[round, above=15mm of t,xshift=29mm ,minimum size=30pt] (goal) {$\goal_{+k}$};
    \node[round,below=15mm of goal, minimum size=30pt] (goal2) {$\goal_{+k-1}$};
        
    \node[round,right=50mm of t,minimum size=30pt] (s) {$s_{+k}$};
    \node[round,below=15mm of s, minimum size=30pt] (s1) {$s_{+1}$};
    \node[round,below=15mm of s1, minimum size=30pt] (s2) {$s_{0}$};
     \node[round,below=15mm of s2, minimum size=30pt] (s3) {$s_{-k+1}$};

  \draw[color=black , very thick,->] (t) edge  node [ very near start, anchor=center] (h1) {} node [pos=0.5,left=5pt] {$1-|\alpha_1|-|\alpha_2|$} (goal) ;
  \draw[color=black , very thick,->] (t)  edge node [near start, anchor=center] (h2) {} node [pos=0.7,right=2pt] {$|\alpha_1|$} (t1) ;
  \draw[color=black , very thick,->] (t)  edge node [pos=0.05,right=5pt] {$|\alpha_2|$} (s2) ;
  \draw[color=black , very thick] (h1.center) edge [bend left=55] node [pos=0.25,right=2pt] {} (h2.center);
  
    \draw[color=black , very thick,->] (t1) edge  node [ very near start, anchor=center] (h11) {} node [pos=0.5,left=5pt] {} (goal2) ;
  \draw[color=black , very thick,->] (t1)  edge node [near start, anchor=center] (h21) {} node [pos=0.7,right=2pt] {} (t2) ;
  \draw[color=black , very thick,->] (t1)  edge node [pos=0.05,right=5pt] {} (s3) ;
  \draw[color=black , very thick] (h11.center) edge [bend left=55] node [pos=0.25,right=2pt] {} (h21.center);

  \draw[color=black , very thick,->] (s) edge  node [ very near start, anchor=center] (g1) {} node [pos=0.5,right=5pt] {$1-|\alpha_1|-|\alpha_2|$} (goal) ;
  \draw[color=black , very thick,->] (s)  edge node [near start, anchor=center] (g2) {} node [pos=0.7,left=2pt] {$|\alpha_1|$} (s1) ;
  \draw[color=black , very thick,->] (s)  edge node [pos=0.05,left=5pt] {$|\alpha_2|$} (t2) ;
  \draw[color=black , very thick] (g1.center) edge [bend right=55] node [pos=0.25,left=2pt] {} (g2.center);
  
      \draw[color=black , very thick,->] (s1) edge  node [ very near start, anchor=center] (g11) {} node [pos=0.5,left=5pt] {} (goal2) ;
  \draw[color=black , very thick,->] (s1)  edge node [near start, anchor=center] (g21) {} node [pos=0.7,right=2pt] {} (s2) ;
  \draw[color=black , very thick,->] (s1)  edge node [pos=0.05,right=5pt] {} (t3) ;
  \draw[color=black , very thick] (g11.center) edge [bend right=55] node [pos=0.25,right=2pt] {} (g21.center);

\end{tikzpicture}
}
\end{center}

\caption{The Markov chain $\cC$ depicted for $k=2$ with $\alpha_1\geq 0$ and $\alpha_2<0$.}\label{fig:Markov_chain}
\end{figure}

The Markov chain $\cC$ has $5k$ states named $t_{-k+1}$, \dots, $t_{+k}$, $s_{-k+1}$, \dots, $s_{+k}$, and $\goal_{+1}$, \dots, $\goal_{+k}$. 
States $t_{-k+1}$, \dots, $t_{0}$, $s_{-k+1}$, \dots, $s_{0}$, and $\goal_{+1}$, \dots, $\goal_{+k}$ are absorbing.
 For $0<i,j\leq k$, there are transitions from $t_{+i}$ to $t_{+i-j}$ with probability $\alpha_j$ if $\alpha_j>0$, to $s_{+i-j}$ 
 with probability $|\alpha_j|$ if $\alpha_j<0$, and to $\goal_{+i}$ with probability $1-|\alpha_1|-\ldots-|\alpha_k|$. Transitions from $s_{+i}$ are defined analogously.
 
 The idea behind this Markov chain is that the reachability probabilities describe how, for arbitrary $n\in \mathbb{N}$ and $1\leq i \leq k$, 
 the values $e(t,nk+i)$ and $e(s,nk+i)$ depend on $n$ and  the values $e(t,(n-1)k+j)$ and $e(s,(n-1)k+j)$ for $1\leq j \leq k$. 
 The transitions in $\cC$ behave as $\gamma$ and $\delta$ in $\cM$, but the decrease in the accumulated weight is explicitly encoded into the state space. 
 Namely, for  $n\in\mathbb{N}$ and $0<i\leq k$, we have
 \begin{align}\label{eqn:cC}
  e(t,nk+i) = \sum_{j=1}^k &\big(\Pr_{\cC, t_{+i}} (\Diamond t_{-k+j})\cdot e(t,(n{-}1)k+j) \nonumber \\
  &+\Pr_{\cC, t_{+i}} (\Diamond s_{-k+j})\cdot e(s,(n{-}1)k+j) \big) \\
  + \sum_{j=1}^k & \Pr_{\cC, t_{+i}} (\Diamond \goal_{+j}) \cdot (nk+j) \nonumber
\end{align}
and analogously for $ e(s,nk+i)$.
We now group the optimal values together in the following vectors
\[v_n = (e(t,nk+k) , e(t,nk+k-1), \ldots, e(t,nk+1), e(s,nk+k) , \ldots, e(s,nk+1))^t \]
for $n\in \mathbb{N}$. In other words, this vector contains the optimal values for the partial expectation when starting in $t$ or $s$ with an accumulated weight from  $\{nk+1,\dots,nk+k\}$. 
Further, we define the vector containing the optimal values for weights in $\{-k+1,\dots,0\}$ which are the least values of accumulated weight reachable under scheduler $\sched$.
\[v_{-1} = (e(t,0) , e(t,-1), \ldots, e(t,-k+1), e(s,0) , e(s,-1), \ldots, e(s,-k+1))^t. \] 
As we have seen, these values are given as follows:
\[e(t,-k+1+j)=\frac{1}{2k^{2(k-j)}}+\beta_j \text{ and }e(s,-k+1+j)=\frac{1}{2k^{2(k-j)}}\] for $0\leq j \leq k-1$.

As the reachability probabilities in $\cC$ are rational and computable in polynomial time, we conclude from \eqref{eqn:cC} 
that there are a matrix $A\in \mathbb{Q}^{2k\times2k}$, and vectors $a$ and $b$ in $\mathbb{Q}^{2k}$ computable in polynomial time such that 
\[ v_n = A v_{n-1} + n a + b,\]
for all $n\in\mathbb{N}$. We claim that the following explicit representation for $n\geq -1$ satisfies this recursion:
\[ v_n = A^{n+1} v_{-1} + \sum_{i=0}^n (n-i) A^i a + \sum_{i=0}^n A^i b.\]
We show this by induction:
Clearly, this representation yields the correct value for $v_{-1}$. So, assume $v_n = A^{n+1} v_{-1} + \sum_{i=0}^n (n-i) A^i a + \sum_{i=0}^n A^i b$. Then,
\begin{align*}
v_{n+1} & = A (A^{n+1} v_{-1} +  \sum_{i=0}^n (n-i) A^i a +  \sum_{i=0}^n A^i b) +(n+1) a + b \\
& = A^{n+2} v_{-1} + \left(\sum_{i=0}^{n} (n-i) A^{i+1} a\right) + (n+1) A^0 a + \left(\sum_{i=1}^{n+1} A^i b \right) + A^0 b \\
& = A^{n+2} v_{-1} + \sum_{i=0}^{n+1} (n+1-i) A^i a + \sum_{i=0}^{n+1} A^i b.
\end{align*}
So, we have an explicit representation for $v_n$.
The value we are interested in is 
\[\PE^\sched_{\cM} = \sum_{\ell=1}^\infty (1/2)^\ell e(t,\ell).\]

\noindent
Let $c=(\frac{1}{2^k} , \frac{1}{2^{k-1}} , \dots, \frac{1}{2^1} , 0 ,\dots, 0)$. 
Then, 
\[(\frac{1}{2^k})^n c\cdot v_n= \sum_{i=1}^{k}  \frac{1}{2^{nk+i}} e(t,nk+i).\]

\noindent
Hence, we can write 
\begin{align*}
&\PE^\sched_{\cM} = \sum_{n=0}^\infty (\frac{1}{2^k})^n c \cdot v_n = c \cdot \sum_{n=0}^\infty (\frac{1}{2^k})^n v_n \\
= {} & c \cdot \sum_{n=0}^\infty (\frac{1}{2^k})^n (A^{n+1} v_{-1} + \sum_{i=0}^n (n-i) A^i a + \sum_{i=0}^n A^i b) \\
= {} & c \cdot \big(  (\sum_{n=0}^\infty (\frac{1}{2^k})^n A^{n+1}) v_{-1}  + (\sum_{n=0}^\infty (\frac{1}{2^k})^n \sum_{i=0}^n (n-i) A^i) a + (\sum_{n=0}^\infty (\frac{1}{2^k})^n \sum_{i=0}^n A^i) b              \big).
\end{align*}
We claim that all of the matrix series involved converge to rational matrices. A key observation is that the maximal row sum in $A$ is at most $|\alpha_1|{+}\ldots{+}|\alpha_k|<1$ because the rows of the matrix contain exactly the probabilities to reach $t_0$, \dots $t_{-k+1}$, $s_0$, \dots, and  $s_{-k+1}$ from a state $t_{+i}$ or $s_{+i}$ in $\cC$ for $1\leq i \leq k$. But the probability to reach $\goal_{+i}$ from these states is already $1{-}|\alpha_1|{-}\ldots{-}|\alpha_k|$. Hence, $\Vert A \Vert_{\infty}$, the operator norm induced by the maximum norm $\Vert \cdot \Vert_\infty$, which equals $\max_{i} \sum_{j=1}^{2k} |A_{ij}|$, is less than $1$.

So, of course also $\Vert \frac{1}{2^k} A \Vert_{\infty}<1$ and hence the  Neumann series $\sum_{n=0}^\infty (\frac{1}{2^k} A)^{n}$ converges to $(I_{2k}-\frac{1}{2^k} A)^{-1}$ where $I_{2k}$ is the identity matrix of size $2k{\times }2k$. So,
\begin{align}
 \sum_{n=0}^\infty (\frac{1}{2^k})^n A^{n+1} = A \sum_{n=0}^\infty (\frac{1}{2^k} A)^{n} = A(I_{2k}-\frac{1}{2^k} A)^{-1}.
\end{align}
Note that $\Vert A \Vert_\infty <1$ also implies that $I_{2k}-A$ is invertible. We observe that for all $n$,
\[\sum_{i=0}^n A^i = (I_{2k}-A)^{-1} (I_{2k} - A^{n+1}) \]
which is shown by straight-forward induction.
Therefore,
\begin{align*}
\sum_{n=0}^\infty (\frac{1}{2^k})^n \sum_{i=0}^n A^i & =  (I_{2k}-A)^{-1} \left( \sum_{n=0}^\infty (\frac{1}{2^k})^n I_{2k} - A  \sum_{n=0}^\infty (\frac{1}{2^k}A)^n   \right) \\
 & =  (I_{2k}-A)^{-1} \left(\frac{2^k}{2^k{-}1} I_{2k} - A(I_{2k}-\frac{1}{2^k} A)^{-1}\right).
\end{align*}
Finally, we show by induction that \[\sum_{i=0}^n (n-i) A^i = (I_{2k} - A)^{-2} (A^{n+1}-A+n(I_{2k}-A)).\]
This is equivalent to \[ (I_{2k} - A)^2 \sum_{i=0}^n (n-i) A^i = A^{n+1}-A+n(I_{2k}-A).\]
For $n=0$, both sides evaluate to $0$. So, we assume the claim holds for $n$.
\begin{align*}
(I_{2k} - A)^2 &\sum_{i=0}^{n+1} (n+1-i) A^i = (I_{2k} - A)^2 \sum_{i=0}^{n} (n-i) A^i + (I_{2k} - A)^2 \sum_{i=0}^{n}A^i \\
 \overset{\mathrm{IH}}{=}{} &  A^{n+1}-A+n(I_{2k}-A) + (I_{2k} - A)^2 \sum_{i=0}^{n}A^i \\
 = {}&  A-A^{n+1}+n(I_{2k}-A) + (I_{2k} - A)^2 (I_{2k}-A)^{-1} (I_{2k} - A^{n+1})\\
 = {}& A-A^{n+1}+n(I_{2k}-A) + I_{2k}-A - A^{n+1} + A^{n+2} \\
 = {}& A^{n+2}- A + (n+1) (I_{2k}-A) .
\end{align*}
The remaining series is the following:
\begin{align*}
&\sum_{n=0}^\infty (\frac{1}{2^k})^n \sum_{i=0}^n (n-i) A^i  \\
={} &  \sum_{n=0}^\infty (\frac{1}{2^k})^n (I_{2k} - A)^{-2} (A^{n+1}-A+n(I_{2k}-A)) \\
={}& (I_{2k} - A)^{-2} \left(\sum_{n=0}^\infty (\frac{1}{2^k})^n A^{n+1} -  \sum_{n=0}^\infty (\frac{1}{2^k})^nA + \sum_{n=0}^\infty (\frac{1}{2^k})^n n(I_{2k}-A) \right) \\
={}&(I_{2k} - A)^{-2} \left(  A(I_{2k}-\frac{1}{2^k} A)^{-1} - \frac{2^k}{2^k{-}1}A+ \frac{2^k}{(2^k{-}1)^2}(I_{2k}-A) \right) .
\end{align*}

We conclude that all expressions in the representation of ${\PE}^\sched_{\cM}$ above are rational and computable in polynomial time. As we have seen, the originally given linear recurrence sequence is non-negative if and only if ${\PE}^{\max}_{\cM} \leq {\PE}^\sched_{\cM}$ for the MDP $\cM$ constructed from the linear recurrence sequence in polynomial time in the previous sections.
\end{proof}

To conclude the Skolem-hardness of the threshold problem of the conditional SSPP, we need the following lemma:

\begin{mylem}\label{lem_app:pe_ce_inter-reducible}
The threshold problems for the partial SSPP is polynomial-time reducible to the threshold problem of the  conditional SSPP. 
\end{mylem}
\begin{proof}

Let $\cM$ be an MDP and let $\vartheta$ be a rational number. We construct an MDP $\cN$ such that ${\PE}^{\max}_{\cM}>\vartheta$ if and only if ${\CE}^{\max}_{\cN}>\vartheta$.
We obtain $\cN$ by adding a new initial state $\sinit^\prime$, renaming the state $\goal$ to $\goal^\prime$, and adding a new state $\goal$ to $\cM$. In $\sinit^\prime$, one action with weight $0$ is enabled leading to the old initial state $\sinit$ and to $\goal$ with probability $1/2$ each. From $\goal^\prime$ there is one new action leading to $\goal$ with probability $1$ and weight $+\vartheta$. 

Each scheduler $\sched$ for $\cM$ can be seen as a scheduler for $\cN$ and vice versa.
Now, we observe that  for  any scheduler $\sched$,  
\[\CE^\sched_\cN= \frac{1/2(\PE^\sched_\cM+\Pr_\cM^\sched(\Diamond \goal) \vartheta)}{1/2+1/2\Pr^\sched_\cM(\Diamond \goal)}=\frac{\PE^\sched_\cM+\Pr_\cM^\sched(\Diamond \goal) \vartheta}{1+\Pr_\cM^\sched(\Diamond \goal)}.\]
Hence, ${\PE}^{\max}_{\cM}>\vartheta$ if and only if ${\CE}^{\max}_{\cN}>\vartheta$.
\end{proof}

\section{Skolem-hardness: conditional value-at-risk for the classical SSPP}\label{app:cvar}

This section provides the proofs for Section \ref{thm:Skolem_cvar}. We modify the MDP constructed in the previous sections.

\begin{mythm}[Theorem \ref{thm:Skolem_cvar}]
The positivity problem is polynomial-time reducible to the following problem:
Given an MDP $\cM$ and  rationals $\vartheta$ and  $p\in(0,1)$, decide whether $\CVaR^{\max}_p (\rawdiaplus \goal)> \vartheta$.
\end{mythm}

Define the random variable 
\[
\rawdiaminus \goal (\zeta)=
\begin{cases}
\rawdiaplus \goal (\zeta) & \text{if } \rawdiaplus \goal \leq 0, \\
0 & \text{otherwise.}
\end{cases}
\]

As we have shown in Section \ref{sub:Skolem_cvar}, it is sufficient to show that the positivity problem is reducible to the threshold problem ``Is $\mathbb{E}^{\max}_{\cM,\sinit}(\rawdiaminus \goal) > \vartheta$?''.

\begin{mylem}[Lemma \ref{lem:Skolem_rawdiaminus}]\label{app_lem:Skolem_rawdiaminus}
The positivity problem is polynomial-time reducible to the following problem:
Given an MDP $\cM$ and a rational $\vartheta$, decide whether $\mathbb{E}^{\max}_{\cM,\sinit} (\rawdiaminus \goal)> \vartheta$.
\end{mylem}

The rest of this section is devoted to the proof of this lemma and we will adjust the MDP used for the Skolem-hardness proof of the threshold problem of the partial SSPP. The main change we have to make  concerns the encoding of the initial values.

So, let $k$ be a natural number, $\alpha_1,\dots,\alpha_k$ be rational coefficients of a linear recurrence sequence, and $\beta_0,\dots, \beta_{k-1}\geq 0$ the rational initial values. W.l.o.g. we again assume these values to be small, namely: $\sum_{1\leq i\leq k} |\alpha_i|\leq \frac{1}{5(k+1)}$ and for all $j$, $\beta_j\leq \frac{1}{3}\alpha$ where $\alpha=\sum_{1\leq i\leq k} |\alpha+i|$.

\begin{figure}[h]
\begin{center}
\resizebox{.6\textwidth}{!}{

\begin{tikzpicture}[scale=1,auto,node distance=8mm,>=latex]
\large
    \tikzstyle{round}=[thick,draw=black,circle]

    \node[round,minimum size=30pt] (t) {$t$};
    \node[round,below=15mm of t, minimum size=30pt] (t1) {$t_1$};
    \node[round,below=15mm of t1, minimum size=30pt] (t2) {$t_2$};
    
    \node[round, above=15mm of t,xshift=29mm ,minimum size=30pt] (goal) {$\goal$};
        
    \node[round,right=50mm of t,minimum size=30pt] (s) {$s$};
    \node[round,below=15mm of s, minimum size=30pt] (s1) {$s_1$};
    \node[round,below=15mm of s1, minimum size=30pt] (s2) {$s_2$};

  \draw[color=black , very thick,->] (t) edge  node [ very near start, anchor=center] (h1) {} node [pos=0.5,left=5pt] {$1-|\alpha_1|-|\alpha_2|$} (goal) ;
  \draw[color=black , very thick,->] (t)  edge node [near start, anchor=center] (h2) {} node [pos=0.7,right=2pt] {$|\alpha_1|$} (t1) ;
  \draw[color=black , very thick,->] (t)  edge node [pos=0.2,right=5pt] {$|\alpha_2|$} (s2) ;
  \draw[color=black , very thick] (h1.center) edge [bend left=55] node [pos=0.25,right=2pt] {$\gamma$} (h2.center);
  
 \draw[color=black , very thick, ->] (t1) edge [bend left=25] node [pos=0.1,left=0pt] {$\wgt:-1$} (t);
  \draw[color=black , very thick, ->] (t2) edge [bend left=105] node [pos=0.1,left=4pt] {$\wgt:-2$} (t);

  \draw[color=black , very thick,->] (s) edge  node [ very near start, anchor=center] (g1) {} node [pos=0.5,right=5pt] {$1-|\alpha_1|-|\alpha_2|$} (goal) ;
  \draw[color=black , very thick,->] (s)  edge node [near start, anchor=center] (g2) {} node [pos=0.7,left=2pt] {$|\alpha_1|$} (s1) ;
  \draw[color=black , very thick,->] (s)  edge node [pos=0.2,left=5pt] {$|\alpha_2|$} (t2) ;
  \draw[color=black , very thick] (g1.center) edge [bend right=55] node [pos=0.25,left=2pt] {$\delta$} (g2.center);
  
 \draw[color=black , very thick, ->] (s1) edge [bend right=25] node [pos=0.1,right=0pt] {$\wgt:-1$} (s);
  \draw[color=black , very thick, ->] (s2) edge [bend right=105] node [pos=0.1,right=4pt] {$\wgt:-2$} (s);

\end{tikzpicture}
}
\end{center}
\caption{In the depicted example the recurrence depth is $2$, $\alpha_1>0$ and $\alpha_2<0$.}\label{fig:linear_recurrence2}
\end{figure}
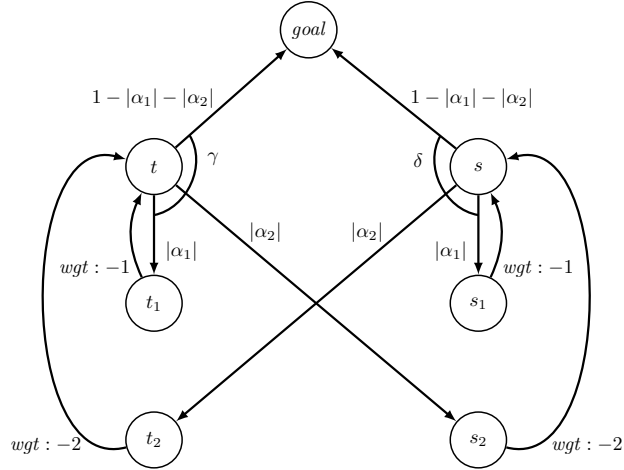

We reuse the MDP-gadget encoding the linear recurrence relation in the case of partial expectations.
This part of an MDP is again depicted in Figure \ref{fig:linear_recurrence2}. 

The first important observation is that the optimal expectation of $\rawdiaminus \goal$ for different starting states and starting weights behaves very similar to optimal partial expectations:
For each state $q$ and each integer $w$, let $e(q,w)$ be the optimal expectation of $\rawdiaminus \goal$ when starting in state $q$ with accumulated weight $w$. If an optimal scheduler chooses $\alpha$ when in $q$ with accumulated weight $w$, then $e(q,w)=\sum_{r\in S} P(q,\alpha,r)\cdot e(r,w{+}\wgt(q,\alpha))$.
If we again let $d(w)=e(t,w)-e(s,w)$, the following lemma  follows as before (see Lemma \ref{lem_app:PE_recurrence}).
\begin{mylem} \label{lem_app:recurrence}
 Let $w\in \mathbb{Z}$. If an optimal scheduler chooses action $\gamma$ in $t$ and $\delta$ in $s$ if the accumulated weight is $w$, then 
$d(w)=\alpha_1 d(w-1) + \dots \alpha_k d(w-k)$.
\end{mylem}

Now, we construct a new gadget that encodes the initial values of a linear recurrence sequence. 
The new gadget  is depicted in Figure \ref{fig:initial_values2}. Besides the actions $\gamma_j$ and $\delta_j$ for $0\leq j \leq k-1$ there are no non-deterministic choices. Recall that $\alpha=\sum_{1\leq i \leq k} |\alpha_i|$.

\begin{figure}[ht]
\begin{center}

\scalebox{.8}{
\begin{tikzpicture}[scale=1,auto,node distance=8mm,>=latex]
\Large
    \tikzstyle{round}=[thick,draw=black,circle]

    \node[round,minimum size=30pt] (t) {$t$};
      \node[round,above=30mm of t, minimum size=30pt] (xj1) {};
    \node[round,above=30mm of xj1, minimum size=30pt] (xj) {$x_j$};
 
  \node[round, right=30mm of t,minimum size=30pt] (zj) {$z_j$};

    \node[round, right=30mm of xj ,minimum size=30pt] (goal) {$\goal$};
         \node[round,right=30mm of goal, minimum size=30pt] (yj) {$y_j$};
         \node[round,below=30mm of yj,minimum size=30pt] (yj1) {};
         \node[round,right=30mm of yj,minimum size=30pt] (yj3) {};
    \node[round,below=30mm of yj1,minimum size=30pt] (s) {$s$};
   
    \node[round,above=30mm of yj, minimum size=30pt] (yj2) {$y_j^\prime$};

  \draw[color=black , very thick,->,loop, out=70, in=110, min distance=30pt] (xj) edge  node [ very near start, anchor=center] (h1) {} node [pos=0.4,right=5pt] {$\wgt: {-}k$}   node [pos=0.5,left=5pt] {$\frac{1}{k+1}$} (xj) ;
  \draw[color=black , very thick,->] (xj)  edge node [very near start, anchor=center] (h2) {} node [pos=0.5,below=2pt] {$\frac{k}{k{+}1}$} (goal) ;
  \draw[color=black , very thick] (h1.center) edge [bend left=55] node [pos=0.25,above=2pt] {} (h2.center);
  
  \draw[color=black , very thick, ->] (t) edge  node [pos=0.1,left=2pt] {$\gamma_j|\wgt:{-}2k{+}j$}  node [pos=0.4,right=2pt] {$\alpha$} node [very near start, anchor=center] (w1) {} (xj1);
  \draw[color=black , very thick, ->] (t) edge  node [pos=0.3,below=2pt] {$1-\alpha$} node [very near start, anchor=center] (w2) {} (zj);
   \draw[color=black , very thick] (w2.center) edge [bend right=55] node [pos=0.25,above=2pt] {} (w1.center);

    \draw[color=black , very thick, ->] (xj1) edge  node [pos=0.1,left=2pt] {$\wgt:+k$} (xj);

      \draw[color=black , very thick, ->] (yj1) edge  node [pos=0.5,left=2pt] {$1-\beta_j/\alpha$} node [pos=0.1,left=2pt] {$\wgt:+k$} node [ very near start, anchor=center] (p1) {} (yj);
      \draw[color=black , very thick, ->] (yj1) edge  node [pos=0.3,right=2pt] {$\beta_j/\alpha$} node [very near start, anchor=center] (p2) {} (yj3);
        \draw[color=black , very thick] (p2.center) edge [bend right=55] node [pos=0.25,above=2pt] {} (p1.center);

   \draw[color=black , very thick, ->] (yj3) edge  node [pos=0.3,right=2pt] {$\wgt:+k$} (yj2);

    \draw[color=black , very thick, ->] (s) edge  node [pos=0.1,right=2pt] {$\delta_j|\wgt:{-}2k{+}j$}  node [pos=0.4,left=2pt] {$\alpha$} node [very near start, anchor=center] (v1) {} (yj1);
  \draw[color=black , very thick, ->] (s) edge  node [pos=0.3,below=2pt] {$1-\alpha$} node [very near start, anchor=center] (v2) {} (zj);
   \draw[color=black , very thick] (v1.center) edge [bend right=55] node [pos=0.25,above=2pt] {} (v2.center);

    \draw[color=black , very thick,->,loop, out=110, in=70, min distance=30pt] (yj) edge  node [ very near start, anchor=center] (g1) {}node [pos=0.2,left=2pt] {$\wgt: {-}k$}    node [pos=0.5,right=5pt] {$\frac{1}{k{+}1}$} (yj) ;
  \draw[color=black , very thick,->] (yj)  edge node [very near start, anchor=center] (g2) {} node [pos=0.5,below=2pt] {$\frac{k}{k{+}1}$} (goal) ;
  \draw[color=black , very thick] (g1.center) edge [bend right=55] node [pos=0.25,above=2pt] {} (g2.center);
  
     \draw[color=black , very thick,->,loop, out=110, in=70, min distance=30pt] (yj2) edge  node [ very near start, anchor=center] (q1) {} node [pos=0.2,left=2pt] {$\wgt: {-}2k$}   node [pos=0.5,right=5pt] {$\frac{1}{k{+}1}$} (yj2) ;
  \draw[color=black , very thick,->] (yj2)  edge node [very near start, anchor=center] (q2) {} node [pos=0.3,above=2pt] {$\frac{k}{k{+}1}$} (goal) ;
  \draw[color=black , very thick] (q1.center) edge [bend right=55] node [pos=0.25,above=2pt] {} (q2.center);

  \draw[color=black , very thick, ->] (zj) edge  node [pos=0.1,right=2pt] {$\wgt:{+}3k{-}2j$}  node [very near start, anchor=center] (v1) {} (goal);

\end{tikzpicture}
}
\end{center}
\caption{The gadget contains the depicted states and actions for each $0\leq j \leq k-1$.
The probability $\alpha=\sum_{1\leq i \leq k} |\alpha_i|$.
}\label{fig:initial_values2}
\end{figure}
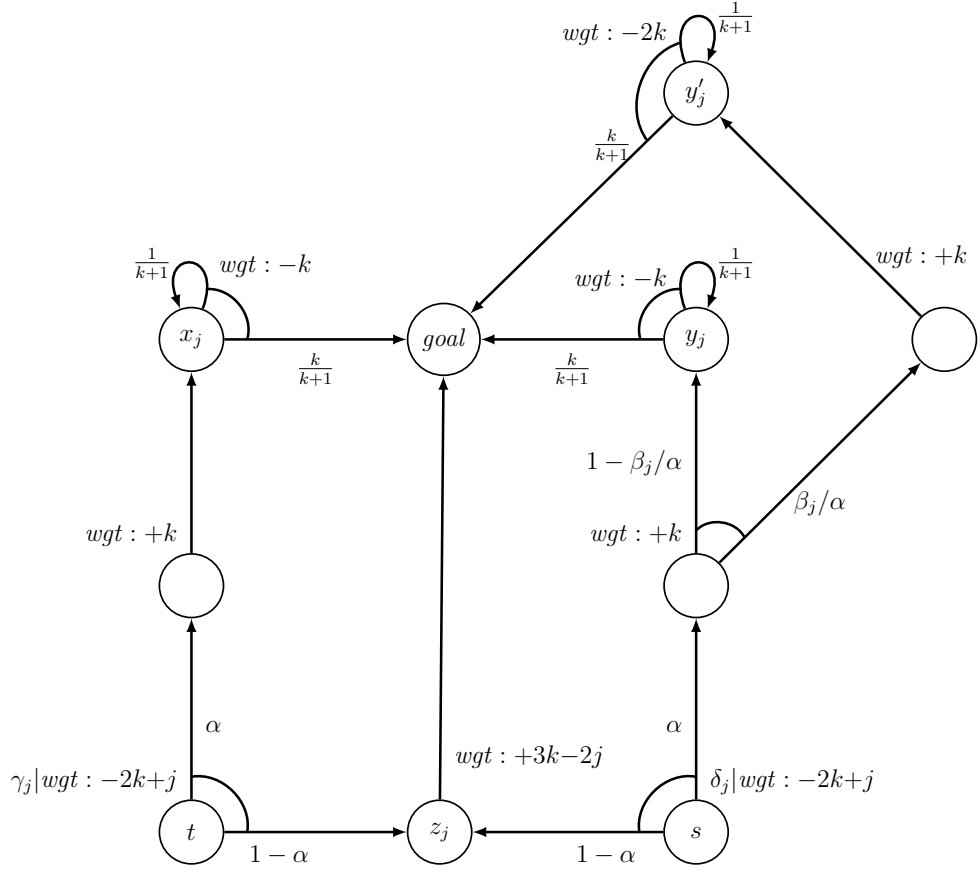

Again, we glue together the two gadgets in states $s$, $t$, and $\goal$.
The key observation is that for $0\leq j\leq k{-}1$ and an accumulated weight ${-}k{+}j$ in state $t$ or $s$ the actions $\gamma_j$ and $\delta_j$, respectively, are optimal for maximizing the expectation of $\rawdiaminus \goal$. For non-negative accumulated weights $\gamma$ and $\delta$ are optimal:

\begin{mylem}\label{lem:initial_values2}
Let $0\leq j \leq k-1$. Starting with accumulated weight ${-}k{+}j$ in state $t$, the action $\gamma_j$ maximizes the partial expectation among the actions $\gamma_0,\dots,\gamma_{k-1}$. Likewise, $\delta_j$ is optimal when starting in $s$ with weight ${-}k{+}j$. If the accumulated weight is non-negative in state $s$ or $t$, then $\gamma$ or $\delta$ are optimal.
\end{mylem}

\begin{proof}
First, we estimate the expectation of $\rawdiaminus \goal$ when choosing $\delta_i$ while the accumulated weight is ${-}k{+}j$. 
If $i> j$, then $\delta_i$ and $\delta$ lead to $\goal$ directly with probability $1{-}\alpha$ and weight $\leq -1$. 
So, the expectation is less than ${-}(1-\alpha)\leq {-}1{+}\frac{1}{4(k{+}1)}$. 

If $i\leq j$, then with probability $1{-}\alpha$ $\goal$ is reached with positive weight, hence $\rawdiaminus \goal$ is $0$ on these paths. 
With probability $\beta_i$, goal is reached via $y_j^\prime$. In this case all runs reach $\goal$ with negative weight. On the way to $y_j^\prime$ weight $2k$ is added, but afterwards subtracted again at least once. 
In expectation weight $2k$ is substracted $\frac{k{+}1}{k}$ many times. Furthermore, ${-}2k{+}i$ is added to the starting weight of ${-}k{+}j$. So, these paths contribute $\beta_i\cdot(2k-2k\frac{k{+}1}{k}{-}3k{+}j{+}i)=({-}3k{+}j{+}i{-}2)\cdot\beta_i$ to the expectation of $\rawdiaminus \goal$. 
With analogous reasoning, we see that the remaining paths contribute $({-}3k{+}j{+}i{-}1)\cdot(\alpha-\beta_i)$. So, all in all the expectation of $\rawdiaminus \goal$ in this situation is $\alpha {\cdot} ({-}3k{+}j{+}i{-}1){-}\beta_i$.
Now, as $\alpha\leq \frac{1}{5(k{+}1)}$ and $\beta_i\leq \frac{\alpha}{3}$ for all $i$, indeed $\delta_j$ is the optimal action. For $\gamma_j$ the same proof with $\beta_i=0$ for all $i$ leads to the same result.

Now assume that the accumulated weight in $t$ or $s$ is $\ell\geq 0$. 
Then, all actions lead to $\goal$ with a positive weight with probability $1-\alpha$. 
In this case $\rawdiaminus \goal$ is $0$. 
However, a scheduler $\sched$ which always chooses $\gamma$ and $\delta$ is better than a scheduler choosing $\gamma_j$ or $\delta_j$ for any $j\leq k{-}1$. 
Under scheduler $\sched$ starting from $s$ or $t$ a run returns to $\{s,t\}$ with probability $\alpha$ while accumulating weight 
$\geq {-}k$ and the process is repeated. After choosing $\gamma_j$ or $\delta_j$ the run moves to $x_j$, $y_j$ or $y_j^\prime$ while accumulating a negative weight. From then on, in each step it will stay in that state with probability greater than $\alpha$ and accumulate weight $\leq {-}k$. Hence, the expectation of $\rawdiaminus \goal$ is lower under $\gamma_j$ or $\delta_j$ than under $\sched$. Therefore indeed $\gamma$ and $\delta$ are the best actions for non-negative accumulated weight in states $s$ and $t$.
\end{proof}

From the proof we also learn the following:

\begin{mycor}
The difference $d(-k{+}j)=e(t,-k{+}j){-}e(s,-k{+}j)$ is equal to $\beta_j$, for $0\leq j\leq k-1$ in the combination of the gadgets presented above.
\end{mycor}

Put together this shows that $d({-k}+\ell) = u_\ell$ where $(u_n)_{n\in\mathbb{N}}$ is the linear recurrence sequence specified by the $\alpha_i$, $\beta_j$, $1\leq i\leq k$, and $0\leq j \leq k{-}1$. Further, we know the optimal behaviour for all accumulated weights $\geq {-}k$ in states $s$ and $t$.


\begin{figure}[ht]
\begin{center}
\resizebox{1\textwidth}{!}{%

\begin{tikzpicture}[scale=1,auto,node distance=8mm,>=latex]
    \tikzstyle{round}=[thick,draw=black,circle]

    \node[round,minimum size=30pt] (t) {$t$};
    \node[round,below=15mm of t, minimum size=30pt] (t1) {$t_1$};
    \node[round,below=10mm of t1, minimum size=30pt] (t2) {$t_2$};
    
    \node[round, above=15mm of t,xshift=29mm ,minimum size=30pt] (goal) {$\goal$};
        
    \node[round,right=50mm of t,minimum size=30pt] (s) {$s$};
    \node[round,below=15mm of s, minimum size=30pt] (s1) {$s_1$};
    \node[round,below=10mm of s1, minimum size=30pt] (s2) {$s_2$};
    
\node[above=15mm of t, minimum size=30pt] (xj) {\vdots};

    \node[above=15mm of s, minimum size=30pt] (yj) {\vdots};

     \node[round,below=80mm of goal, minimum size=30pt] (c) {$c$};
          \node[round,below=15mm of c, minimum size=30pt] (sinit) {$\sinit$};

  \draw[color=black , very thick,->] (t) edge  node [ very near start, anchor=center] (h1) {} node [pos=0.1,right=5pt] {$1-|\alpha_1|-|\alpha_2|$} (goal) ;
  \draw[color=black , very thick,->] (t)  edge node [near start, anchor=center] (h2) {} node [pos=0.7,right=2pt] {$|\alpha_1|$} (t1) ;
  \draw[color=black , very thick,->] (t)  edge node [pos=0.2,right=5pt] {$|\alpha_2|$} (s2) ;
  \draw[color=black , very thick] (h1.center) edge [bend left=55] node [pos=0.25,right=2pt] {$\gamma$} (h2.center);
  
 \draw[color=black , very thick, ->] (t1) edge [bend left=25] node [pos=0.1,left=2pt] {$\wgt:-1$} (t);
  \draw[color=black , very thick, ->] (t2) edge [bend left=75] node [pos=0.1,left=2pt] {$\wgt:-2$} (t);

  \draw[color=black , very thick,->] (s) edge  node [ very near start, anchor=center] (g1) {} node [pos=0.4,left=5pt] {$1-|\alpha_1|-|\alpha_2|$} (goal) ;
  \draw[color=black , very thick,->] (s)  edge node [near start, anchor=center] (g2) {} node [pos=0.7,left=2pt] {$|\alpha_1|$} (s1) ;
  \draw[color=black , very thick,->] (s)  edge node [pos=0.2,left=5pt] {$|\alpha_2|$} (t2) ;
  \draw[color=black , very thick] (g1.center) edge [bend right=55] node [pos=0.25,left=2pt] {$\delta$} (g2.center);
  
 \draw[color=black , very thick, ->] (s1) edge [bend right=25] node [pos=0.1,right=2pt] {$\wgt:-1$} (s);
  \draw[color=black , very thick, ->] (s2) edge [bend right=75] node [pos=0.1,right=2pt] {$\wgt:-2$} (s);

  \draw[color=black , very thick, ->] (t) edge [bend left=25] node [pos=0.1,left=2pt] {$\gamma_j$} (xj);

  \draw[color=black , very thick, ->] (s) edge [bend right=25] node [pos=0.1,right=2pt] {$\delta_j$} (yj);

 \draw[color=black , very thick,->] (sinit) edge  node [  near start, anchor=center] (f1) {} node [pos=0.5,left=5pt] {$\frac{1}{2}$} (c) ;
  \draw[color=black , very thick,->] (sinit)  edge [loop, min distance=20mm, out=60, in=0] node [pos=.1, anchor=center] (f2) {} node [pos=0.4,right=7pt] {$\frac{1}{2}$} (sinit) ;
  \draw[color=black , very thick] (f1.center) edge [bend left=55] node [pos=1,xshift=2pt,above=10pt] {$\wgt: +1$} (f2.center);
  
   \draw[color=black , very thick, ->] (c) edge [loop, in=180, out=180, min distance=60mm] node [pos=0.1,above=5pt] {$\tau$} (t);
   
   \draw[color=black , very thick, ->] (c) edge [loop, in=0, out=0, min distance=60mm] node [pos=0.1,above=5pt] {$\sigma$} (s);

\end{tikzpicture}

 }
   \end{center}

\caption{The MDP contains the upper part as depicted in Figure \ref{fig:initial_values2} for all $0\leq j \leq k-1$. The middle part is depicted for $k=2$, $\alpha_1\geq 0$, and $\alpha_2<0$.}\label{fig:MDP2}
\end{figure}
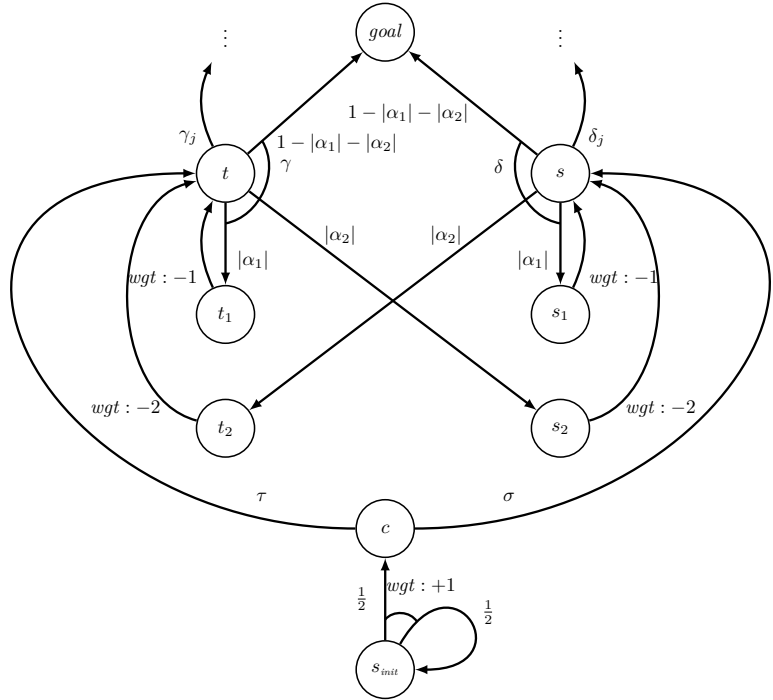

Finally, we again equip the MDP with an initial component as depicted in Figure~\ref{fig:MDP2}: From the initial state $\sinit$, one action with weight $+1$ is enabled. This action leads to a state $c$ with probability $\frac{1}{2}$ and loops back to $\sinit$ with probability $\frac{1}{2}$. In $c$, the decision between action $\tau$  leading to state $t$ and action $\sigma$ leading to state $s$ has to be made. 
We have shown that action $\tau$ is optimal in $c$ for accumulated weight $w$ if and only if $u_{w+k-1}\geq 0$.
Further, the scheduler $\sched$ always choosing $\tau$ in $c$ and actions $\gamma, \gamma_0, \dots, \gamma_{k-1}, \delta, \dots$ as described in Lemma \ref{lem:optimal_scheduler} is optimal if and only if $u_n\geq 0$ for all $n$.

In order to complete the proof of the theorem, we compute the expectation $\vartheta=\mathbb{E}^{\sched}_{\sinit}(\rawdiaminus \goal)$. 
This is done analogously to the computation in the proof of Theorem \ref{app_thm:threshold_PE}. We describe the necessary modifications here:
The weight levels are shifted by $1$ compared to the proof for the partial SSPP. Hence, we again define a vector containing the optimal values in $s$ and $t$ for the weight levels encoding the intial values:
\[v_{-1} = (e(t,{-}1) , e(t,{-}2), \ldots, e(t,{-}k), e(s,{-}1) , e(s,{-}2), \ldots, e(s,-k))^t. \] 
Then, the optimal values on higher weight levels can again be computed in terms of this vector. We define the vectors $v_n$ for all $n$ as
\[v_n = (e(t,nk{+}k{-}1) , e(t,nk{+}k{-}2), \ldots, e(t,nk), e(s,nk{+}k{-}1) , \ldots, e(s,nk))^t. \]
Again the weight levels are shifted by $1$ compared to the proof of Theorem \ref{app_thm:threshold_PE}.
Using the
Markov chain in Figure \ref{fig:Markov_chain}, we obtain a matrix $A\in \mathbb{Q}^{2k\times 2k}$ as before (in fact the same matrix as in the proof of Theorem \ref{app_thm:threshold_PE}) such that $v_{n}=Av_{n{-}1}$ for all $n\geq 0$. As $\rawdiaminus$ evaluates to $0$ on all paths reaching $\goal$ with positive weight, this is considerably simpler than in the case of the partial SSPP.
In particular, this time the explicit representation for $v_n$ takes the simple form $v_n=A^{n+1}v_{-1}$. The remaining argument is now completely analogous to the proof of Theorem \ref{app_thm:threshold_PE}.
So, $\vartheta=\mathbb{E}^{\sched}_{\sinit}(\rawdiaminus \goal)$ is a rational computable in polynomial time.
 We conclude that $\mathbb{E}^{\max}_{\sinit}(\rawdiaminus \goal)>\vartheta$ if and only if there is an $n$ such that $u_n<0$. This finishes the proof of Lemma \ref{app_lem:Skolem_rawdiaminus}.

\section{Skolem-hardness: weighted long-run frequency}\label{app:wlf}

We illustrate the notion of weighted long-run frequency in the following example, which already shows that memoryless schedulers are not sufficient to solve the optimization problem.

\begin{myex} \label{ex:intro_PMP}
Consider the  example MDP $\cM$ depicted in Figure \ref{fig:ex-MDP}. The only non-deterministic choice is the choice between actions $\alpha$ and $\beta$ in state $\sinit$. So, there are two memoryless deterministic schedulers, $\sched_\alpha$ choosing $\alpha$, and $\sched_\beta$ choosing $\beta$. We compute their weighted long-run frequencies by taking the quotient of the expected accumulated weight on paths satisfying $\neg\Fail \Until \Goal$ and the expected return time from the initial situation to the initial situation, i.e. starting from $\sinit$ until $\sinit$ is reached from $\Goal \cup \Fail$. 
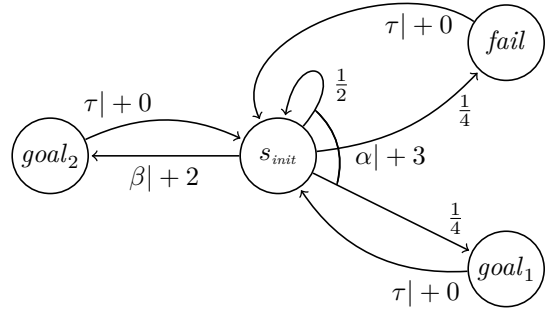
\begin{figure}[ht]
\vspace{-12pt}
\begin{center}
\scalebox{1}{
\begin{tikzpicture}
  [shorten >=1pt,node distance=12mm and 1.5cm,on grid,auto,semithick]
  \node[state,inner sep=1pt,minimum size=10mm] (s) {$\sinit$};
\node[state,inner sep=1pt,minimum size=10mm] (t) [right = 30mm of s,yshift=15mm] {$\fail$};
\node[state,inner sep=1pt,minimum size=10mm] (u) [right = 30mm of s,yshift=-15mm] {$\goal_1$};
\node[state,inner sep=1pt,minimum size=10mm] (w) [left = 30mm of s] {$\goal_2$};

 \path[->]
(s) edge [loop above, distance= 10mm, out=50, in=80 ]   node [pos= 0.1, anchor=center]  (h1) {}	 node [pos=.3,right] {$\frac{1}{2}$}  (s) 
 (s) edge [bend right=20] node [pos=.4,below] {$\alpha | +3$} node [pos=.9,below] {$\frac{1}{4}$} (t)
 (s) edge  node  [pos=0.15, anchor=center] (h2) {} node [pos=.9,above] {$\frac{1}{4}$}  (u)
(t) edge [bend right, out=-60, in =-90] node [pos=.2, below] {$\tau|+0$} (s)
 (u) edge[bend left] node [pos=.2,below] {$\tau|+0$} (s)
 (s) edge node [below] {$\beta | +2$} (w)
 (w) edge[bend left,out=25,in=155]   node [pos=.2,above] {$\tau|+0$} (s)
 ;
 \path
 (h1.center) edge  [bend left] (h2.center)
 (h2.center) edge  [bend right] (h1.center)
;

\end{tikzpicture}
}
\end{center}

\caption{Example MDP $\cM$: The weights associated to the actions are stated after the bar and non-trivial probability values are denoted as fractions next to the arrows. $\Goal=\{\goal_1,\goal_2\}$ and $\Fail=\{\fail\}$.}\label{fig:ex-MDP}
\end{figure}

Under $\sched_\alpha$ it takes $2$ steps in expectation to reach $\goal_1$ or $\fail$. So, the expected return time is $3$. The expected accumulated weight when reaching $\goal_1$ or $\fail$ is $6$. As the path only satisfies $\neg\Fail \Until \Goal$ if $\goal_1$ is reached and as reaching $\goal_1$ and reaching $\fail$ are equally likely, the expected accumulated weight for our calculation is $3$. Hence, $\PMP^{\sched_\alpha}_{\cM} = 1$.
For $\sched_\beta$ the calculation is simple: We always receive weight $2$ in $2$ steps. Hence, $\PMP^{\sched_\beta}_{\cM}=1$ as well.

However, the following scheduler $\sched$ using memory achieves a higher value:
The scheduler $\sched$ chooses $\alpha$ in the initial situation. Only in case the self-loop to $\sinit$ is taken, it afterwards chooses $\beta$. Under this scheduler, there are only three paths from the initial situation to the initial situation. Therefore, the computation of the weighted long-run frequency is easy:
\[ \PMP^\sched_{\cM} = \frac{\text{expected accumulated weight}}{\text{expected return time}}= \frac{1/4 \cdot 0 + 1/4 \cdot 3 + 1/2 \cdot 5}{1/4\cdot 2 + 1/4\cdot 2 +1/2 \cdot 3} =\frac{13}{10} .\]

This example already demonstrates that  memoryless schedulers are not sufficient for the optimization of weighted long-run frequencies. Further, it indicates that the weight already accumulated since the last visit to $\Goal$ or $\Fail$ is an important information from the history of a run. Also for partial and conditional expectations and long-run probabilities, similar examples show  the analogous results (cf. \cite{tacas2017,fossacs2019,lics2019}). {\hfill $\blacksquare$}
\end{myex}

\begin{figure}[h]
\begin{center}
\scalebox{.9}{
\begin{tikzpicture}[scale=.9,auto,node distance=8mm,>=latex]
\large
    \tikzstyle{round}=[thick,draw=black,circle]

    \node[round,minimum size=30pt] (t) {$t$};
    \node[round,above=10mm of t, minimum size=30pt] (xj) {$x_j$};

    \node[round, right=20mm of t ,minimum size=30pt] (goal) {$\goal$};
     \node[round, right=20mm of goal ,minimum size=30pt] (s) {$s$};
     \node[round,above=10mm of s, minimum size=30pt] (yj) {$y_j$};

    \node[round,above=35mm of goal, minimum size=30pt] (fail2) {$\fail$};
    \node[round,left=20mm of fail2, minimum size=30pt] (xj2) {$x_j^\prime$};
    \node[round,right=20mm of fail2, minimum size=30pt] (yj2) {$y_j^\prime$};

      \draw[color=black , very thick,->] (t) edge   node [pos=0.10,right=2pt] {$p_2$} (fail2) ;
  \draw[color=black , very thick,->] (t)  edge node [very near start, anchor=center] (h2) {} node [pos=0.25,below=2pt] {$p_1$} (goal) ;
   \draw[color=black , very thick, ->] (t) edge node [ near start, anchor=center] (h1) {} node [pos=0.5,right=2pt] {$p_0$} node [pos=0.1,left=2pt] {$\gamma_j|\wgt:+k-j$} (xj);
  \draw[color=black , very thick] (h1.center) edge [bend left=55] node [pos=0.25,above=2pt] {} (h2.center);
  
   \draw[color=black , very thick,->] (xj2) edge  node [very near start, anchor=center] (g2) {}  node [pos=0.25,above=2pt] {$p_2$} (fail2) ;
  \draw[color=black , very thick,->] (xj2)  edge node [pos=0.1,right=2pt] {$p_1$} (goal) ;
   \draw[color=black , very thick, ->] (xj2) edge node [ near start, anchor=center] (g1) {} node [pos=0.5,right=2pt] {$p_0$} (xj);
  \draw[color=black , very thick] (g2.center) edge [bend left=55] node [pos=0.25,above=2pt] {} (g1.center);
  
  \draw[color=black , very thick, ->] (xj) edge [bend left=25] (xj2);

      \draw[color=black , very thick,->] (s) edge   node [pos=0.10,left=2pt] {$q_2$} (fail2) ;
  \draw[color=black , very thick,->] (s)  edge node [very near start, anchor=center] (j1) {} node [pos=0.25,below=2pt] {$q_1$} (goal) ;
   \draw[color=black , very thick, ->] (s) edge node [ near start, anchor=center] (j2) {} node [pos=0.5,left=2pt] {$q_0$} node [pos=0.1,right=2pt] {$\delta_j|\wgt:+k-j$} (yj);
  \draw[color=black , very thick] (j1.center) edge [bend left=55] node [pos=0.25,above=2pt] {} (j2.center);
  
   \draw[color=black , very thick,->] (yj2) edge  node [very near start, anchor=center] (k2) {}  node [pos=0.25,above=2pt] {$q_2$} (fail2) ;
  \draw[color=black , very thick,->] (yj2)  edge node [pos=0.1,left=2pt] {$q_1$} (goal) ;
   \draw[color=black , very thick, ->] (yj2) edge node [ near start, anchor=center] (k1) {} node [pos=0.5,left=2pt] {$q_0$} (yj);
  \draw[color=black , very thick] (k1.center) edge [bend left=55] node [pos=0.25,above=2pt] {} (k2.center);
  
  \draw[color=black , very thick, ->] (yj) edge [bend right=25] (yj2);
    
\end{tikzpicture}

}
\end{center}

\caption{The new gadget contains the depicted states and actions for each $0\leq j \leq k-1$. The probabilities are: $p_0=q_0=|\alpha_1|+\dots+|\alpha_k|  $, $p_1=(1-p_0) (\frac{1}{2k^{2(k-j)}}+\beta_j)$, $p_2=(1-p_0)(1-(\frac{1}{2k^{2(k-j)}}+\beta_j))$, and $q_1=(1-q_0) \frac{1}{2k^{2(k-j)}}$, $q_2=(1-q_0)(1-\frac{1}{2k^{2(k-j)}})$.
 All actions except for $\gamma_j$ and $\delta_j$ have weight $0$.
 }\label{app_fig:gadget_wlf}
\end{figure}
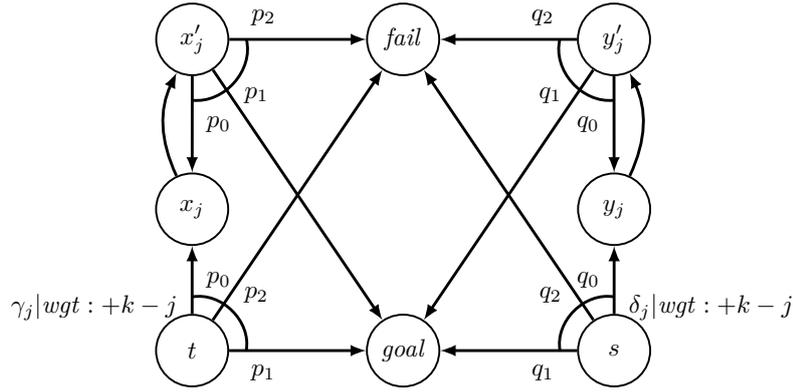


\begin{mythm}[Theorem \ref{thm:Skolem_wlf}]\label{thm_app:Skolem_wlf}
The positivity problem is polynomial-time reducible to the following problem: Given an MDP $\cM$ and a rational $\vartheta$, decide whether ${\PMP}^{\max}_{\cM}>\vartheta$.
\end{mythm}

\begin{proof}
Given the parameters of a linear recurrence sequence, we again construct the MDP depicted in Figure \ref{fig:MDP}. To obtain $\cM$, we then replace the gadget encoding
the initial values of the linear recurrence sequence depicted in  Figure \ref{fig:initial_values} by  the gadget depicted in Figure \ref{app_fig:gadget_wlf} in which the probabilities to reach $\goal$ are the same as before, but the expected number of steps changes. Further, we add  transitions from states $\goal$ and $\fail$ to $\sinit$ with probability $1$. In the MDP $\cM$, the expected time from the moment $\sinit$ is entered from $\goal$ or $\fail$ to the next time this happens does not depend on the scheduler. In fact, it takes $3$ steps in expectation until $t$ or $s$ is reached and from there on it takes $2\frac{1}{1{-}|\alpha_1|{-} \ldots {-} |\alpha_k|}$ many steps until $\goal$ or $\fail$ is reached no matter which actions are chosen. So the expected return time is $4+2\frac{1}{1{-}|\alpha_1|{-} \ldots {-} |\alpha_k|}$. This however means that a scheduler $\sched$ achieving ${\PMP}^\sched_{\cM}>\vartheta$ also achieves ${\PE}^\sched_{\cM}>\vartheta(4+2\frac{1}{1{-}|\alpha_1|{-} \ldots {-} |\alpha_k|})$.
 So, we can use the rational  threshold computed in Section \ref{subsec:hardness_threshold_PE} and divide it by this constant expected return time in order to establish the Skolem-hardness of the threshold problem for partial mean pay-offs as well.
 \end{proof}

The rational weight introduced in the reduction can easily be transformed to an integer weight by multiplying all weights with the denominator of $\vartheta$. Partial and conditional expectations simply scale accordingly.

\section{Skolem-hardness: long-run probabilities and frequency-LTL}\label{app:LRP}

Suppose that $\cM$ is a strongly connected MDP
with state space $S$ and two designated set of states $\Goal$ and $\Fail$.
An important result we will use in the sequel states that optimal weighted long-run frequencies can be approximated by finite memory schedulers.

\begin{mylem}
\label{app_lem:FM-scheduler}
  For each scheduler $\tsched$ for $\cM$, each $\varepsilon>0$ and 
  each state $s$ of $\cM$,
  there is a finite memory scheduler $\fsched$ for $\cM$ such that: 
  \[
   \PMP^{{\fsched}}_{{\cM,s}} \ \geqslant \ 
    \PMP^{{\tsched}}_{{\cM,s}} -\varepsilon
  \]
\end{mylem}

\begin{proof}
Let $\tsched$ be an arbitrary scheduler for $\cM$.
By Fatou's lemma, we have:
\begin{eqnarray*}
   \PMP^{\tsched}_{\cM,s} 
   & = & 
   \mathbb{E}^{\tsched}_{\cM,s}
     \left(  \liminf_{n\to \infty} \frac{1}{n+1} \sum_{i=0}^n \wgt(s_i,\alpha_i)\cdot \mathds{1}_{\pi[i\dots]\vDash \neg \Fail \Until \Goal}
     \right) 
   \\
   & \leqslant &
   \liminf_{n\to \infty} \
   \mathbb{E}^{\tsched}_{\cM,s} 
   \left( 
       \frac{1}{n+1} \sum_{i=0}^n \wgt(s_i,\alpha_i)\cdot \mathds{1}_{\pi[i\dots]\vDash \neg \Fail \Until \Goal}
   \right)
\end{eqnarray*}
So, there exists $k_0\in \Nat$ such that for all $k\geqslant k_0$:
\[
  \mathbb{E}^{\tsched}_{\cM,s} 
  \left( 
     \frac{1}{k{+}1} 
     \sum_{i=0}^{k}
         \wgt(s_i,\alpha_i)\cdot \mathds{1}_{\pi[i\dots]\vDash \neg \Fail \Until \Goal}
  \right) 
  \ \geqslant \ 
  \PMP^{\tsched}_{\cM,s} - \frac{\varepsilon}{2}
\]

Let $\qsched_s$ be the following finite memory scheduler with two modes. 
If the current state is not in $\Goal$ or $\Fail$, it starts in the first mode, 
in which it behaves like an MD-scheduler maximizing the 
probability of $\neg \Fail \Until \Goal$. 
As soon as  a state in $\Goal \cup \Fail$ has been reached,
scheduler $\qsched_s$ operates in the second mode, 
in which it memorylessly minimizes 
the expected number of steps until reaching $s$. 
Let $f_{t,s} = \mathbb{E}^{\min}_{\cM,t}(\text{``steps until $s$''})$ 
denote the expected number of steps 
this scheduler $\qsched_s$ needs to reach $s$ 
in the second mode starting from state $t$. 
We then define $f_s=\max_{t\in S} f_{t,s}$ and  $f=\max_{s\in S} f_s$.

We now construct a finite-memory scheduler $\fsched$ satisfying the claim of the lemma. 
First, choose a natural number $k$ with $k \geqslant k_0$ and
$k{+}1> \frac{2\cdot W\cdot f_s}{\varepsilon}$ where $W$ is the maximal weight appearing in $\cM$. 
The behavior of scheduler $\fsched$ is as follows.
In its first mode, it starts in $s$ and behaves like $\tsched$ 
in the first $k$ steps. 
Then, it switches to the second mode and behaves like $\qsched_s$ 
until it reaches $s$ (in the second mode of $\qsched_s$). 
Afterwards, it switches back to the first mode.

As $\qsched_s$ maximizes the probability of $\neg \Fail \Until \Goal$
whenever it starts in a state not in $\Goal$ or $\Fail$, we obtain:
\begin{eqnarray*}
 && \frac{1}{k{+}1} \cdot \mathbb{E}^\tsched_{\cM,s}( \sum_{i=0}^{k}
     \wgt(s_i,\alpha_i)\cdot \mathds{1}_{\pi[i\dots]\vDash \neg \Fail \Until \Goal}) \\
  & \leqslant &  
  \frac{1}{k{+}1} \cdot \mathbb{E}^\fsched_{\cM,s}( \sum_{i=0}^{k}
     \wgt(s_i,\alpha_i)\cdot \mathds{1}_{\pi[i\dots]\vDash \neg \Fail \Until \Goal}).
\end{eqnarray*}
Furthermore, the expected number of steps which $\fsched$ 
takes to follow $\tsched$ for $k{+}1$ steps and 
to return to $s$ via $\qsched_s$ 
is at most $k{+}1{+}f_s$.

Expressing the weighted long-run frequency of $\fsched$ as a quotient, we obtain:
\begin{eqnarray*}
 \PMP^{\fsched}_{\cM,s}
  &\geqslant & 
  \frac { \mathbb{E}^\fsched_{\cM,s}( \sum_{i=0}^{k}
     \wgt(s_i,\alpha_i)\cdot \mathds{1}_{\pi[i\dots]\vDash \neg \Fail \Until \Goal})}
    {k+1+f_s} 
  \\
  \\[-1ex]
  & \geqslant & 
  \frac { \mathbb{E}^\tsched_{\cM,s}( \sum_{i=0}^{k}
     \wgt(s_i,\alpha_i)\cdot \mathds{1}_{\pi[i\dots]\vDash \neg \Fail \Until \Goal})}
  {(k+1)\cdot (1+\varepsilon/2W)} 
  \\
  \\[-1ex]
  & \geqslant & 
  \frac {  \mathbb{E}^\tsched_{\cM,s}( \sum_{i=0}^{k}
     \wgt(s_i,\alpha_i)\cdot \mathds{1}_{\pi[i\dots]\vDash \neg \Fail \Until \Goal})}
  {k+1} \cdot (1-\varepsilon/2W) 
  \\
  \\[-1ex]
  & \geqslant &
  ( \PMP^{\tsched}_{\cM,s} -\varepsilon/2) \cdot (1-\varepsilon/2W) 
\end{eqnarray*}
by the choice of $k$. Using the fact that 
$\PMP^{\tsched}_{\cM,s}$ is bounded by W
we obtain:
\[
 \begin{array}{lcl}
   \multicolumn{3}{l}{\PMP^{\fsched}_{\cM,s}}
   \\[1ex]
   & \geqslant &
   ( \PMP^{\tsched}_{\cM,s} -\varepsilon/2) \cdot (1-\varepsilon/2W) 
   \\[1ex]
   & \geqslant &
   \PMP^{\tsched}_{\cM,s} -\varepsilon.
 \end{array}
\]
This completes the proof.
\end{proof}

The analogous result for long-run probabilities was shown in \cite{lics2019}.

We now provide the proof to Lemma \ref{lem:encoding}. The key idea is to encode integer weights via a labelling of states and to use a simple regular co-safety property to mimic the reception of weights in weighted long-run frequencies. 
In the sequel, we will work with weighted states instead of weighted state-action pairs. Further, we assume that the weights are only $-1$, $0$, and $+1$. This assumption leads to a pseudo-polynomial blow-up in the general case. The weights in the MDP $\cK$ constructed for Theorem \ref{thm:Skolem_wlf} above are, however, at most $k$. As the MDP has more than $2k$ states, transforming $\cK$ to weights $-1$, $0$, and $+1$ only leads to a polynomial blow-up.
As this MDP has no non-trivial end-components,  $\{\goal,\fail\}$ is visited infinitely often with probability $1$ under any scheduler.
Let $\AP=\{n,z,p,c,g,f\}$ be a set of atomic propositions representing \emph{negative}, \emph{zero}, and \emph{positive} weight, as well as  \emph{coin flip},  $\goal$, and $\fail$, respectively. 

We construct a new labelled MDP $\cL=\langle S^\prime, \Act, \Pr^\prime, \sinit, L \rangle$ with a labeling function $L:S^\prime \to 2^\AP$.
The state space $S^\prime=S\setminus (\Goal \cup \Fail) \cup (\Goal \cup \Fail) \times \{0,1\}$. The set of actions stays the same. 
For any action $\alpha\in\Act$, states $s,t \in S\setminus (\Goal\cup\Fail)$  and $(u,i),(v,j) \in  (\Goal\cup\Fail)\times \{0,1\}$, we define
$
\Pr^\prime(s,\alpha,t)=\Pr(s,\alpha,t)$, 
$\Pr^\prime(s,\alpha,(u,i))=\frac{1}{2}\Pr(s,\alpha,u)$,
$\Pr^\prime((u,i),\alpha,s)=\Pr(u,\alpha,s)$, and
$\Pr^\prime((u,i),\alpha,(v,j))=\frac{1}{2}\Pr(u,\alpha,v)$.
So, intuitively the only change is that states in $\Goal$ or $\Fail$ are duplicated and whenever they are entered each of the copies is visited with probability $\frac{1}{2}$.
The labeling function $L$ does the following: For a state $s\in S\setminus (\Goal\cup \Fail)$, we have $L(s)=\{n\}$ iff $\wgt(s)=-1$,  $L(s)=\{z\}$ iff $\wgt(s)=0$, and $L(s)=\{p\}$ iff $\wgt(s)=+1$.
For states $(t,i)\in \Fail\times\{0,1\}$, we have $L((t,i))=\{f\}$ iff $i=0$ and  $L((t,i))=\{f,c\}$ iff $i=1$. For states $(u,j)\in \Goal\times\{0,1\}$, we have $L((u,j))=\{g,x\}$ iff $j=0$ and $L((u,j))=\{g,x,c\}$ iff $i=1$ where $x$ is $n$, $z$, or $p$ depending on $\wgt(u)$ as above.

\begin{mylem}[Lemma \ref{lem:encoding}]
For the MDPs $\cK$ and $\cL$ constructed above, we have  
\[{\PMP}^{\max}_{\cK} = \frac{1}{2}+\frac{1}{2}\LP^{\max}_\cL (\cA).\]
\end{mylem}

\begin{proof}
Let $\fsched$ be a finite memory scheduler for $\cK$ and $\cL$ which induces a single BSCC. It is clear that it is enough to consider such scheduler for the maximization as in a strongly connected MDP a scheduler could always move to the best of multiple BSCCs.

In this single BSCC $\cB^\fsched$ where states are again enriched with memory modes of $\fsched$, we can compute the probability to satisfy $\neg \Fail \Until \Goal$ from each state. For $\mathfrak{s}\in\cB^{\fsched}$, let $p_{\mathfrak{s}}$ be this probability. Furthermore, let $x_{\mathfrak{s}}$ be the steady state probability of state $\mathfrak{s}$ in this single BSCC.
Then,
\[{\PMP}^{\fsched}_{\cK} = \sum_{\mathfrak{s}\in \cB^\fsched} \wgt(\mathfrak{s}) \cdot p_\mathfrak{s} \cdot x_\mathfrak{s}.\]

By the assumption that $\Goal \cup \Fail$ intersects all end components, we can conclude that the probability to satisfy $\neg \Goal \Until \Fail$ is $1-p_\mathfrak{s}$ in each state $\mathfrak{s}$.
So, we can compute the long-run probability of $\cA$ in $\cL$ as follows. We  use that $c$ holds with probability $1/2$ whenever  a state in $\Goal$ or $\Fail$ is reached. 
\begin{eqnarray*}
\LP^{\fsched}_\cL (\cA) & = & \sum_{p\in L(\mathfrak{s})} x_\mathfrak{s} \cdot (p_\mathfrak{s} + 1/2 (1- p_{\mathfrak{s}})) \\
 && + \sum_{z\in L(\mathfrak{s})} x_\mathfrak{s} \cdot (1/2 \cdot p_\mathfrak{s} + 1/2 (1- p_{\mathfrak{s}})) \\
 && + \sum_{n\in L(\mathfrak{s})} x_\mathfrak{s} \cdot  1/2 (1- p_{\mathfrak{s}}) \\
 && + \sum_{f\in L(\mathfrak{s})} x_\mathfrak{s} \cdot 1/2
\end{eqnarray*}
Now, it is easy to conclude that ${\PMP}^{\fsched}_{\cK} = \frac{1}{2}+\frac{1}{2}\LP^{\fsched}_\cL (\cA)$.

That the maximum agrees with the supremum over finite-memory schedulers on the left-hand side was shown in Lemma \ref{app_lem:FM-scheduler} using Fatou's lemma. We sketch the proof for the long-run frequency of the co-safety property $\phi$ given by $\cA$ in the MDP $\cL$ following the ideas of \cite{lics2019}. 
Note that by the fact that there are no non-trivial end components, we can conclude that the states labelled with $g$ or $f$ are reached infinitely often with probability $1$. Further, there is a bound $d$ on the expected time to the next visit to $f$ pr $g$ under any scheduler and from any starting point.
Note that in the automaton $\cA$ any run is accepted or rejected as soon as $g$ or $f$ is read. Furthermore, there is a bound $r$ on the expected time to return to the initial state $\sinit$ form any other state under a scheduler $\rsched$ minimizing this time.
Now, let $\tsched$ be any scheduler for $\cL$ and $\varepsilon>0$. By Fatou's lemma, we have:
\begin{eqnarray*}
   \LP^{\tsched}_{\cL,\sinit} (\cA)
   & = & 
   \mathbb{E}^{\tsched}_{\cL,\sinit}
     \left(  \liminf_{n\to \infty} \frac{1}{n+1} \sum_{i=0}^n \mathds{1}_{\pi[i\dots]\vDash \phi}
     \right) 
   \\
   & \leqslant &
   \liminf_{n\to \infty} \
   \mathbb{E}^{\tsched}_{\cL,\sinit} 
   \left( 
     \frac{1}{n+1} \sum_{i=0}^n \mathds{1}_{\pi[i\dots]\vDash \phi}
   \right)
\end{eqnarray*}
So, there exists $k_0\in \Nat$ such that for all $k\geqslant k_0$:
\[
  \mathbb{E}^{\tsched}_{\cL,\sinit} 
\left(  \frac{1}{k+1} \sum_{i=0}^k \mathds{1}_{\pi[i\dots]\vDash \phi}  \right) 
  \ \geqslant \ 
\LP^{\tsched}_{\cL,\sinit} (\cA) - \frac{\varepsilon}{2}
\]

Pick $N\geq k_0$ such that $N+1>\frac{2\cdot d \cdot r}{\varepsilon}$.
We now provide a finite memory scheduler $\sched$ with $\LP^{\sched}_{\cL,\sinit} (\cA) \geq \LP^{\tsched}_{\cL,\sinit} (\cA) - \varepsilon$:
The scheduler $\sched$ behaves like $\tsched$ for the first $N+1$ steps. Then, it maximizes the probability for $\neg f \Until g$ if more states labelled $p$ than states labelled $n$ have been visited since the last visit to a state labelled $f$ or $g$. Otherwise, it maximizes the probability for $\neg g \Until f$. As there have only been $N+1$ steps since the beginning this can be tracked with finite memory.
As soon as it reaches a state labelled $f$ or $g$ now, it returns to $s$ using the choices of $\rsched$. Then, it restarts behaving like $\tsched$ for $N+1$ steps and so on.

First, we see that 
\[  \mathbb{E}^{\tsched}_{\cL,\sinit} 
\left(  \frac{1}{N+1} \sum_{i=0}^N \mathds{1}_{\pi[i\dots]\vDash \phi}  \right) \leq 
\mathbb{E}^{\sched}_{\cL,\sinit} 
\left(  \frac{1}{N+1} \sum_{i=0}^N \mathds{1}_{\pi[i\dots]\vDash \phi}  \right)
\]
due to the optimization of the probabilities of  $\neg f \Until g$ or $\neg g \Until f$ depending on the number of states labelled $p$ or $n$ since the last visit to $f$ or $g$:
If $\ell$ suffixes which have not yet been accepted or rejected by $\cA$  started with $p$ and $m$ started with $n$, then the expected number of those runs which will be  accepted  under some $\qsched$ is:
$\Pr^{\qsched} (\neg f \Until g) \cdot \ell + \Pr^{\qsched} (\neg g \Until f) \cdot (1/2 (\ell + m))$. And, $\sched$ behaves such that this value is maximized as $\Pr^{\qsched} (\neg f \Until g)=1-\Pr^{\qsched} (\neg g \Until f)$.
The expected return time to the initial state in the initial memory mode under $\sched$ is at most $N+1+d+r$.
So, the long-run probability under $\sched$ satisfies:
\begin{eqnarray*}
\LP^{\sched}_{\cL,\sinit} (\cA) & \geq & \frac{\mathbb{E}^{\sched}_{\cL,\sinit} 
\left(  \sum_{i=0}^N \mathds{1}_{\pi[i\dots]\vDash \phi}  \right)}{N+1+d+r} \\
& \geq & \frac{\mathbb{E}^{\sched}_{\cL,\sinit} 
\left(  \sum_{i=0}^N \mathds{1}_{\pi[i\dots]\vDash \phi}  \right)}{(N+1)\cdot (1+\varepsilon/2)}\\
&\geq & \frac{\mathbb{E}^{\sched}_{\cL,\sinit} 
\left(  \sum_{i=0}^N \mathds{1}_{\pi[i\dots]\vDash \phi}  \right)}{(N+1)} (1-\varepsilon/2) \\
&\geq & (\LP^{\tsched}_{\cL,\sinit} (\cA) - \frac{\varepsilon}{2})(1-\varepsilon/2)\\
& \geq & \LP^{\tsched}_{\cL,\sinit} (\cA) - \varepsilon.
\end{eqnarray*}

Now, it follows that the maximal long-run probability of $\phi$ is obtained by taking the supremum over all finite memory schedulers as well and this finishes the proof.
\end{proof}

\begin{mythm}[Theorem \ref{thm:fLTL}]
There is a polynomial-time reduction from the positivity problem to the following qualitative model checking problem for frequency LTL  for a fixed LTL-formula $\phi$: Given an MDP $\cM$ and a rational $\vartheta$, is $\Pr^{\max}_\cM (G^{>\vartheta}_{\inf} (\varphi))=1$?
\end{mythm}

\begin{proof}
We provided a polynomial reduction from the positivity problem to the following problem: given a strongly connected MDP $\cN$ in which each end component contains a state labelled $f$ or $g$ and a rational $\vartheta$, decide whether there is a scheduler $\sched$ such that $\LP_{\cN}^{\sched}(\cA)>\vartheta$. The property expressed by $\cA$ is captured by the following LTL-formula $\phi$:
\begin{eqnarray*}
((g \land p) \lor (g\land z \land c )\lor (f \land c)) &\lor& (p\land ((\neg g \land \neg f)\Until (g\lor (f\land c)))) \\
&\lor& (z\land ((\neg g \land \neg f)\Until ((g\land c)\lor (f\land c)))) \\
&\lor& (n\land ((\neg g \land \neg f)\Until (f\land c))).
\end{eqnarray*}
We claim that there is such a scheduler $\sched$ if and only if there is a scheduler $\tsched$ such that $G^{>\vartheta}_{\inf}(\phi)$ holds with probability $1$ under $\tsched$ in $\cN$.

The semantics of $G^{>\vartheta}_{\inf}(\phi)$ as given in \cite{ForejtKK15} is the following: An infinite path $\infpath$ satisfies $G^{>\vartheta}_{\inf}(\phi)$ if 
\[
\liminf_{n\to \infty } \frac{1}{n+1} \sum_{i=0}^n \mathds{1}_{\infpath[i\dots]\vDash \phi} > \vartheta.
\]

Suppose there is  a scheduler with $\sched$ with $\LP_{\cN}^{\sched}(\cA)>\vartheta$.
We have seen in the previous proof, that we can assume that $\sched$ is a finite memory scheduler. As $\cN$ is strongly connected, we can further assume that $\sched$ induces only one BSCC. We claim that under this  scheduler $\sched$ also $G^{>\vartheta}_{\inf}(\phi)$ holds with probaility $1$. For finite memory schedulers it is easy to check that the expected long-run probability equals the expected long-run frequency: As above let $x_\mathfrak{s}$ be the steady state probability of states $\mathfrak{s}$ enriched with memory modes in the single BSCC $\cB^\sched$ induced by $\sched$.
Further, let $p_\mathfrak{s}$ be the probability that a run starting in $\mathfrak{s}$ under $\sched$ satisfies $\phi$.
Then, $\LP_{\cN}^{\sched}(\cA) = \sum_{\mathfrak{s}\in \cB^{\sched}} x_\mathfrak{s}\cdot p_{\mathfrak{s}}$. But the same expression also computes the expected frequency with which $\phi$ holds on suffixes, i.e. the expected value $\mathbb{E}_{\cN}^{\sched}(\liminf_{n\to \infty } \frac{1}{n+1} \sum_{i=0}^n \mathds{1}_{\infpath[i\dots]\vDash \phi})$. Furthermore, this can be seen as a mean-payoff in a strongly connected Markov chain where the weights are $p_\mathfrak{s}$ in each state. But, for a mean-payoff in a strongly connected Markov chain it is well known that the mean pay-off of almost all paths agrees with the expected value. So, $\liminf_{n\to \infty } \frac{1}{n+1} \sum_{i=0}^n \mathds{1}_{\infpath[i\dots]\vDash \phi} > \vartheta$ almost surely.

Conversely, If there is a scheduler $\tsched$ such that $G^{>\vartheta}_{\inf}(\phi)$ holds with probability $1$ under $\tsched$ in $\cN$, the expected value $\mathbb{E}_{\cN}^{\sched}(\liminf_{n\to \infty } \frac{1}{n+1} \sum_{i=0}^n \mathds{1}_{\infpath[i\dots]\vDash \phi})>\vartheta$. By an ananlogue  Fatou's lemma argument, we can find a finite memory scheduler with expected long-run frequency, and hence long-run probability, greater than $\vartheta$.
\end{proof}

\section{Saturation points: conditional value-at-risk for the classical SSPP}\label{cvar_pos}

\begin{mythm}[Theorem \ref{thm:cvar_pos}]
Given an MDP $\cM=(S,\sinit,\Act,P,\wgt,\goal)$ with non-negative weights and no end-components except for one absorbing state $\goal$ as well as a  rational probability value $p\in(0,1)$, the value  $\CVaR^{\max}_p (\rawdiaplus \goal)$ is computable  in pseudo-polynomial time.
\end{mythm}

\begin{proof}
Let $N$ be the number of states of $\cM$, $\delta$ be the minimal non-zero transition probability, and $W$ the maximal weight occuring in $\cM$. As there are no end components except for $\goal$, the state $\goal$ is reached within $N$ steps from any other state under any scheduler with probability at least $\delta^N$. Let $\ell$ be such that $(1-\delta^N)^\ell \leq 1-p$. Note that $\ell$ simply has to be chosen bigger than $\frac{\log(1-p)}{\log(1-\delta^N)}$ and hence can be computed in polynomial time. So, its numerical value is  at most of pseudo-polynomial size. Then, the probability that a path accumulates a weight higher that $K=\ell\cdot N\cdot W$ is less than $1-p$ under any scheduler. So, the value-at-risk $\VaR^\sched_p(\rawdiaplus \goal)$ is less than $K$ under any scheduler $\sched$.
This means that the value-at-risk and hence also the conditional value-at-risk are not affected if we simply assign weight $K$ to all paths accumulating weight at least $K$.
We can achieve this by explicitly encoding the accumulated weight into the state space:

We define a new MDP $\cN$ with a set of weighted target states as follows: The state space $S^\prime$ is $S\times \{0,\dots,K\}$. The initial state $\sinit^\prime$ is $(\sinit,0)$. The set of actions stays the same. The transition probability function $P^\prime$ is defined by $P^{\prime}((s,i),\alpha,(t,j)=P(s,\alpha,t)$ if $i+\wgt(s,\alpha)=j<K$ or $i+\wgt(s,\alpha)\geq j=K$, and  $P^{\prime}((s,i),\alpha,(t,j)=0$ otherwise. There is no weight function in $\cN$, but instead a set of weighted target states. The target states are $(\goal,i)$ with weight $i$ for all $i<K$ and $(s,K)$ with weight $K$ for all $s\in S$.
In this way, each path $\zeta$ reaching $\goal$ in $\cM$ (i.e. almost all paths) corresponds to a path $\zeta^\prime$ in $\cN$ and if $\wgt(\zeta)<K$, then $\zeta^\prime$ reaches a terminal state with weight $\wgt(\zeta)$. If $\wgt(\zeta)\geq K$, then $\zeta^\prime$ reaches a terminal state with weight $K$.

Now, we can compute the optimal  conditional value-at-risk with the probability value $p$ for the random variable assigning the terminal weight to a path  with the methods for weighted reachability presented in \cite{kretinsky2018} in polynomial time in the size of $\cN$ to obtain the value $\CVaR^{\max}_p (\rawdiaplus \goal)$ in $\cM$. 
The linear program presented there requires a guess of the value-at-risk. However, the value-at-risk in our setting is a natural number between $0$ and $K$, so there are only pseudo-polynomially many candidates.
This results in an exponential (pseudo-polynomial) time algorithm for our problem.
\end{proof}




\section{Computation of maximal weighted long-run frequency}\label{app:PMP}

In this section, we present the full proof of Theorem \ref{thm:comp_wlf}.
The proof modifies the proof of \cite[Lemma IV.7]{lics2019} stating the existence of a saturation point for long-run probabilities of constrained reachability properties to account for the weights.

Let us first recall the relevant notations. Let $\cM=\langle S, \Act, \Pr, \sinit, \wgt, \Goal, \Fail \rangle$ be a strongly connected MDP with non-negative weights.

The saturation point is computed as follows: Let $S^\prime=S\setminus (\Goal \cup \Fail)$.
For each state $s$ let $p_s^{\max}=\Pr^{\max}_{\cM,s} (\neg \Fail \Until \Goal)$. 
Further, let \[p_{s,\alpha}=\sum_{t\in S} \Pr(s,\alpha,t) \cdot p^{\max}_t\] for all states $s\in S^\prime$. 
We write $\Act(s)$ for the set of actions that are enabled in $s$,
i.e., $\act \in \Act(s)$ iff $\sum_{t\in S}P(s,\act,t) =1$.
We define $\Act^{\max}(s)= \{ s\in \Act(s) | p_{s,\alpha}^{\max}=p_s^{\max}\}$ for all $s\in S^\prime$. Further, define
\[\delta = \min \{ p_s^{\max}-p_{s,\alpha}^{\max} | s\in S^\prime, \alpha \in \Act(s)\setminus \Act^{\max}(s)\}.\]
If this set is empty, we set $\delta=1$.

Further, we fix a memoryless deterministic scheduler $\qsched$ which maximizes the probability of $\neg \Fail \Until \Goal$ from all states in $S^\prime$.
Then, we define $e_{s,t}$ for all states $s,t\in \Goal\cup \Fail$: The value $e_{s,t}=\min_\rsched \mathbb{E}_{\cM,s} (\text{``steps until $t$''})$ where $\rsched$ ranges over all schedulers which behave like $\qsched$ whenever $|S^\prime|+1$ states in $|S^\prime|$ have been visited consecutively. We let $e=\min_{s,t\in S^{\prime}} e_{s,t}$.

Finally, we let $W$ be the maximal weight occuring in $\cM$ and define

\[ K = \max (W\cdot e / \delta, W\cdot (|S^\prime|+1)). \qedhere\]

The value $e_{t,s}$ for some $t,s\in \Goal \cup \Fail$ can be computed as follows: We construct an MDP $\cN$, by taking $N$-many copies of each state in $S^\prime$. So, the state space of $\cN$ is $\Goal \cup \Fail \cup S^\prime \times \{1,\dots, N\}$. From states in $\Goal\cup\Fail$, the transitions are as in $\cM$ and if a state $s$ in $S^\prime$ would be reached in $\cM$, the copy $(s,1)$ is reached in $\cN$ instead. From a copy $(s,i)$ with $i<N$, the same actions as in $s$ are enabled and if a state $t\in S^\prime$ would be  reached in $\cM$, we move to the copy $(t,i+1)$ in $\cN$ instead. Finally, in states of the form $(s,N)$, only the action that $\qsched$ chooses in $s$ is enabled and if the process moves to another state in $S^\prime$ we also move to the $N$th copy of this state.

The MDP $\cN$ is of polynomial size in the size of $\cM$ and 
\[e_{s,t}=\min_\rsched \mathbb{E}_{\cN,s} (\text{``steps until $t$''})\]
 where $\rsched$ now ranges over all schedulers for $\cN$. For each $t$, fix a scheduler $\rsched_{\cN,t}$ minimizing this expected value.  So, $e_{t,s}$ can be computed by standard techniques for stochastic shortest path problems in time polynomial in the size of $\cN$.


Let $ \FM(\SatPoint)$ be the class of all finite memory schedulers which choose the actions according to some memoryless  scheduler maximizing the probability of $\neg \Fail \Until \Goal$  whenever the accumulated weight since the last visit to $\Goal$ or $\Fail$ is at least $K$.
\begin{mythm}
 \label{appendix:sat-point}
 For each finite-memory scheduler $\tsched$, there is a
 scheduler $\sched \in \FM(\SatPoint)$
 with 
 $\PMP^{\sched}_{\cM}\ \geqslant \
     \max_{s\in S} \PMP^{\tsched}_{\cM,s}$.
\end{mythm}

\begin{proof}
Let $\tsched$ be an finite-memory scheduler for $\cM$
with modes (memory cells) in the finite set $X$. 
Let $\cC^\tsched$ denote the Markov chain induced by $\tsched$. 
We can think of the states in $\cC^{\tsched}$ as pairs $(s,x)$ consisting of
a state $s$ in $\cM$ and a mode $x\in X$.
We may assume w.l.o.g. that $\cC^{\tsched}$ has a single BSCC, say 
$\cB^{\tsched}$. This yields that all states of $\cC^{\tsched}$ 
have the same weighted long-run frequency.
Let us simply write $\PMP^{\sched}_{\cM}$ for this value.
From now on, we refer to the set of  states $\Goal\times X$ in $\cC^{\tsched}$  as $\Goal$,  to $\Fail\times X$ as $\Fail$, and $S^\prime \times X$ as $S^\prime$.

If $\cB^{\tsched}$ contains no state from $\Goal$ then
$\PMP^{\tsched}_{\cM}=0$ and the claim is trivial as
we can deal with any $\FM(\SatPoint)$-scheduler.

Suppose now that $\cB^{\tsched}$ contains at least one goal state.
Then, almost all $\tsched$-paths visit 
infinitely often some  goal state.

We now explain how to modify $\tsched$'s decision for generating a scheduler
in $\FM(\SatPoint)$ with the desired property.
Our procedure works by induction on the number $k^{\tsched}$ of 
$(\Goal \cup \Fail)$-states
$\fmstate{s}{x}=(s,x)$ in $\cB^{\tsched}$ 
where 
\[\Pr_{\cB^{\tsched},\fmstate{s}{x}}
               (\neXt (S^\prime \Until^{\geqslant \SatPoint} D^{\tsched}))>0.
\]
Here, $D^{\tsched}$ denotes the set of states $\fmstate{t}{y}=(t,y)$ in $S^\prime$ in
the BSCC $\cB^{\tsched}$ where $\tsched(\fmstate{t}{y})(\alpha) >0$
for some action $\alpha$ with $p^{\max}_t\not = p^{\max}_{t,\alpha}$. The formula $S^\prime \Until^{\geqslant \SatPoint} D^{\tsched}$ holds if  a path only visits states in $S^\prime$ as long as the accumulated weight is less than $\SatPoint$ and afterwards satisfies $S^\prime \Until D^{\tsched}$.

If $k^{\tsched}=0$ then for every path 
$\finpath = \fmstate{s_0}{}\alpha_0\ldots \alpha_{n-1}\fmstate{s_n}{}$ in 
$\cB^{\tsched}$ with accumulated weight  $n \geqslant \SatPoint$ where all but the
first state are in $S^\prime$, the states
$\fmstate{s_{\text{\tiny $\SatPoint$}}}{}, 
 \fmstate{s_{\text{\tiny $\SatPoint{+}1$}}}{},\ldots, \fmstate{s_n}{}$
do not belong to $D^{\tsched}$.
That is, $\tsched$ schedules only actions in $\Act^{\max}$ for these states.
But then the probability for $\neg \Fail \Until \Goal$ in $\cB^{\tsched}$ from 
each of the states $\fmstate{s_i}{}=(s_i,x_i)$ with $i\geqslant \SatPoint$
equals $p_{s_i}^{\max}$.%
\footnote{For general MDPs, schedulers that
  only select actions in $\Act^{\max}$ might not achieve the maximal
  probability for $\neg \Fail \until \Goal$. 
  This, however, is only possible if the schedulers
  under consideration realize an end component consisting of state in $S^\prime$.
  As $\cB^{\tsched}$ is a BSCC with at least one $\Goal$-state,
  this case does not apply to scheduler $\tsched$.}
This implies that 
if $k^{\tsched}=0$ then $\tsched$ is an $\FM(\SatPoint)$-scheduler
and we can deal with $\sched=\tsched$.

Suppose now that $k^{\tsched}\geqslant 1$.
We show how to transform $\tsched$ into a new finite-memory scheduler $\sched$
with a single BSCC  such that
$\PMP^{\sched}_{\cM} \geqslant \PMP^{\tsched}_{\cM}$
and $k^{\sched} < k^{\tsched}$.

Given states $\fmstate{s}{x}=(s,x)$ and $\fmstate{t}{y}=(t,y)$ 
in $\cB^{\tsched}$, 
where $\fmstate{s}{x}$ is in $\Goal \cup \Fail$
and $\fmstate{t}{y}$ in $S^\prime$,
let $\Gamma_{\fmstate{s}{x},\fmstate{t}{y}}$ 
denote the set of of finite $\tsched$-paths 
$\finpath=\fmstate{s_0}{}\alpha_0\ldots \alpha_{n-1}\fmstate{s_n}{}$ 
such that 
\begin{itemize}
\item 
  $\wgt(\finpath) \geqslant \SatPoint$,
\item
  $\fmstate{s_0}{}=\fmstate{s}{x}$,
  $\fmstate{s_n}{}=\fmstate{t}{y}$,
\item
  $\fmstate{s_1}{},\ldots,\fmstate{s_n}{}$ are in $S^\prime$, and 
\item
  $\tsched(\fmstate{s_n}{})(\alpha) > 0$ for some
  action $\alpha \notin \Act^{\max}(t)$.
\end{itemize}
Let $\Pi_{\fmstate{s}{x},\fmstate{t}{y}}$ denote the set of
paths
$\finpath \in \Gamma_{\fmstate{s}{x},\fmstate{t}{y}}$ such that
no proper prefix of $\finpath$ belongs to 
$\Gamma_{\fmstate{s}{x},\fmstate{t}{y}}$,
and let
$\Pi_{\fmstate{s}{x}}$ denote the union of the sets
$\Pi_{\fmstate{s}{x},\fmstate{t}{y}}$.
As $k^{\tsched}$ is positive, we can pick some state  $\Goal \cup \Fail$-state
$\fmstate{s}{x}=(s,x)$ in $\cB^{\tsched}$ where
$\Pi_{\fmstate{s}{x}}$ is nonempty.

The definition of finite-memory scheduler $\sched$ is as follows.
 Scheduler $\sched$ operates in two phases.
 Its first phase starts in $\fmstate{s}{x}=(s,x)$ and
 uses additional memory cells to keep track of the accumulated weight since the last visit of
 $\fmstate{s}{x}$.
 As long as this accumulated weight is smaller than $\SatPoint$ 
 or if a $\Goal \cup \Fail$-state has been reached along a path
 where the counter value is always smaller than $\SatPoint$,
 scheduler $\sched$ just behaves like $\tsched$. 
 As soon as the counter value exceeds $\SatPoint$, 
 scheduler $\sched$ switches to the second phase and behaves as 
 scheduler $\rsched_{N,s}$. The scheduler $\rsched_{N,s}$ has been defined above where also the values $e_{t,s}$ have been defined.

  More precisely, if $\sched$'s current state $t$ in $\cM$ belongs to
  $S^\prime$ then
  $\sched$ mimics the behavior of $\rsched_{N,s}$ from 
  state $(t,N)$ in $\cN$ described above as the current accumulated weight exceeds $N\cdot W$ because $\SatPoint\geq N\cdot W$, and so at least $N$ steps have been taken since the last visit to $\Goal$ or $\Fail$.
 Thus,
 by following
 $\rsched_{N,s}$'s decisions, $\sched$ 
 will only choose actions in $\Act^{\max}$ until $S^\prime$ is left.
 As soon as state $s$ is reached in $\sched$'s second phase
 (this will happen with probability 1 as $\rsched_{N,s}$ minimizes
  the expected number of steps to $s$ from every state in 
  the strongly connected MDP $\cN$),
 $\sched$ switches back to the first phase 
 and restarts to mimic $\tsched$ from state $s$ in mode $x$, 
 i.e., from state $\fmstate{s}{x}$ in $\cB^{\tsched}$.
 For all states that are not reachable from $\fmstate{s}{x}$ in this way,
  $\sched$ behaves as $\tsched$.

As $\tsched$ has a single BSCC, so does $\sched$, although
  the BSCC $\cB^{\sched}$ induced by $\sched$ can be different
  from $\cB^{\tsched}$.
As $\fmstate{s}{x}$ belongs to both $\cB^{\tsched}$ and $\cB^{\sched}$,
$\fmstate{s}{x}$ is visited infinitely often almost surely 
with finite expected return time under both schedulers 
$\sched$ and $\tsched$.

Let us first observe that we indeed have $k^{\sched} < k^{\tsched}$.
This is thanks to the fact that (1) $\rsched_{N,s}$ maximizes the probability
for $a \Until b$ whenever $N$ or more consecutive $S^\prime$-states have
been visited, so in particular if the accumulated weight since the last visit to $\Goal$ or $\Fail$ is at least $\SatPoint$, and
(2) the reference state $\fmstate{s}{x}$ is not an $S^\prime$-state.
Thus, for each $(\Goal\cup \Fail)$-state $\fmstate{u}{z}$ 
visisted by $\rsched_{N,s}$ in the return (second)
phase of $\sched$ we have:
\[
  \Pr_{\cB^{\sched},\fmstate{u}{z}}
               (\neXt (S^\prime \Until^{\geqslant \SatPoint} D^{\sched}))=0
\]
Hence, whenever $\fmstate{u}{z}$ is a $(\Goal \cup \Fail)$-state in
$\cB^{\sched}$ where
$\Pr_{\cB^{\sched},\fmstate{u}{z}}
               (\neXt (S^\prime \Until^{\geqslant \SatPoint} D^{\sched}))$
is positive then
\begin{itemize}
\item
   $\fmstate{u}{z}\not=\fmstate{s}{x}$,
\item
   $\fmstate{u}{z}$ also belongs to $\cB^{\tsched}$ and
\item
  the
  $\sched$-paths  from $\fmstate{u}{z}$ satisfying
  $\neXt (S^\prime \Until^{\geqslant \SatPoint} D^{\sched})$
  are also $\tsched$-paths
  and satisfy $\neXt (S^\prime \Until^{\geqslant \SatPoint} D^{\tsched})$.
\end{itemize}
The last item yields 
$\Pr_{\cB^{\tsched},\fmstate{u}{z}}
               (\neXt (S^\prime \Until^{\geqslant \SatPoint} D^{\tsched}))>0$.
This completes the proof that $k^{\sched}$ is smaller than $k^{\tsched}$.

We now show that 
$\PMP^{\sched}_{\cM} \geqslant \PMP^{\tsched}_{\cM}$. 
To simplify the calculations, we present the proof for the
case where $\Pi_{\fmstate{s}{x}}$ is a singleton, say
$\Pi_{\fmstate{s}{x}}=\{\finpath\}$.

Furthermore, let $\lgth$ be the length of $\finpath$, and
$\fmstate{t}{y}=(t,y)=\last(\finpath)$. Further, we know that the accumulated weight on $\finpath$ is $\geqslant \SatPoint$.
Again, to simplify the calculations, let us suppose that
there is a single action $\alpha \in \Act(t)\setminus \Act^{\max}(t)$
that $\tsched$ schedules for $\fmstate{t}{y}$ with
positive probability $p$.%
\footnote{%
  At the end of the proof, we briefly explain how to treat the general case
  where $\Pi_{\fmstate{s}{x}}$ is a (prefix-free) countable set of paths,
  for which $\tsched$ can schedule multiple actions not in
  $\Act^{\max}$ with positive probability.}
So,
\[
  p \ = \ 
    \tsched(\fmstate{t}{y})(\alpha) \ > \ 0
\]
(Note that $p=1$ if $\tsched$ is a deterministic finite-memory scheduler.)

The long run probabilities of the two schedulers $\sched$ and $\tsched$
can be expressed as follows.

Given a state $\fmstate{u}{z}$ in $\cB^{\tsched}$, let
$\ens^\tsched_{\fmstate{u}{z},\fmstate{s}{x}}$ be the expected 
number of steps from $\fmstate{u}{z}$ to $\fmstate{s}{x}$ under $\tsched$
(via paths from $\fmstate{u}{z}$ to $\fmstate{s}{x}$
of length at least 1 where all intermediate states are different from
$\fmstate{s}{x}$).
Let $\eap^{\tsched}_{\fmstate{u}{z},\fmstate{s}{x}}$ 
denote the expected received weight
that $\tsched$ gains during this period. This is to be understood as follows: The weight at some step is received only if the suffix of the path starting there satisfies $\neg\Fail \Until \Goal$.
The value $\ens^\tsched_{\fmstate{s}{x},\fmstate{s}{x}}$ can be understood as 
the expected return time from and to $\fmstate{s}{x}$ under $\tsched$.
Then: 
\begin{equation}
  \label{lrp-tsched}
  \PMP^{\tsched}_{\cM} \ = \ 
  \frac{\eap^{\tsched}_{\fmstate{s}{x},\fmstate{s}{x}}}
       {\ens^{\tsched}_{\fmstate{s}{x},\fmstate{s}{x}}}
  \tag{$\dagger$}
\end{equation}
For the scheduler $\sched$ we express 
$\PMP^{\sched}_{\cM}$ 
as the fraction of the expected accumulated weight along return paths 
and 
the expected return time from $\fmstate{s}{x}$ to $\fmstate{s}{x}$ as well:
\[
  \LPr{\sched}{\cM}{a\Until b} \ = \ 
  \frac{\eap^{\sched}_{\fmstate{s}{x},\fmstate{s}{x}}}
       {\ens^{\sched}_{\fmstate{s}{x},\fmstate{s}{x}}}
\]
To provide 
an upper bound for $\ens^{\sched}_{\fmstate{s}{x},\fmstate{s}{x}}$
and a lower bound for
$\eap^{\sched}_{\fmstate{s}{x},\fmstate{s}{x}}$,
we need several auxiliary notations.

Recall that $e_{t,s}$ is the expected number of steps that 
$\rsched_{N,s}$ needs from $(t,N)$ to $s$ in $\cN$.
Hence, $e_{t,s}$ is an upper bound for the expected number
of steps $\ens^{\sched}_{\fmstate{t}{y},\fmstate{s}{x}}$
that $\sched$ needs from state $\fmstate{t}{y}=\last(\finpath)$
to the reference state $\fmstate{s}{x}$.
The value $e$ has been defined as the maximum of the values
$e_{t,s}$. Hence, we obtain:
\[
  e \ \ \geqslant \ \ 
  e_{t,s} \ \geqslant \ \ens^{\sched}_{\fmstate{t}{y},\fmstate{s}{x}}
\]
Let $\ens^\tsched_{\fmstate{t}{y},\alpha,\fmstate{s}{x}}$ denote the
the expected number of steps 
that $\tsched$ needs from $\fmstate{t}{y}$ 
to $\fmstate{s}{x}$, under the assumption that action
$\alpha$ is scheduled in $\fmstate{t}{y}$ (which happens with
probability $p$).
So, if
$(u_1,z_1),\ldots,(u_\ell,z_\ell)$ denote the $\alpha$-successors
of $\fmstate{t}{y}$ in $\cB^{\tsched}$
then:
\[
  \ens^\tsched_{\fmstate{t}{y},\alpha,\fmstate{s}{x}}
  \ \ = \ \
  1 +
  \sum_{i=1}^{\ell} P(t,\alpha,u_i)
      \cdot \ens^{\tsched}_{(u_i,z_i),\fmstate{s}{x}}
\]
Similarly, we define: 
\[
  \eap^{\tsched}_{\fmstate{t}{y},\alpha,\fmstate{s}{x}}
  \ =  \
  \Pr^{\tsched}_{\fmstate{t}{y}}(\neg \Fail \until \Goal) \cdot \wgt(t,\alpha)+
  \sum_{i=1}^{\ell} P(t,\alpha,u_i)
      \cdot \eap^{\tsched}_{(u_i,z_i),\fmstate{s}{x}}
\]
For $0\leqslant i \leqslant \lgth$, let $\fragment{\wp}{0}{i}$ denote the 
probability under $\tsched$ for generating the path fragment 
$\prefix{\finpath}{i}$, the first $i$ steps, from state $s$ in mode $x$. 
So, $\fragment{\wp}{0}{\lgth}$ is the probability under $\tsched$ for
generating the full path $\finpath$ from $\fmstate{s}{x}$.

For the expected number of steps 
$\ens^{\sched}_{\fmstate{s}{x},\fmstate{s}{x}}$ that $\sched$ needs from 
$\fmstate{s}{x}$ to $\fmstate{s}{x}$ along paths of length at least 1,
we get:
\begin{equation}
  \label{ens-sched}
  \ens^{\sched}_{\fmstate{s}{x},\fmstate{s}{x}}
  \ \ \leqslant  \ \
  \ens^{\tsched}_{\fmstate{s}{x},\fmstate{s}{x}} + 
    \fragment{\wp}{0}{\lgth} 
    \cdot p \cdot 
      (e-\ens^{\tsched}_{\fmstate{t}{y},\alpha,\fmstate{s}{x}})
  \tag{*}
\end{equation}
The reason is that $\sched$ and $\tsched$ only differ when $\finpath$ has been generated and $\alpha$ has been chosen by $\tsched$. This happens with probability $ \fragment{\wp}{0}{\lgth}  \cdot p  $. From there on, $\sched$ takes at most $e$ steps in expectation to return to $\fmstate{s}{x}$ while $\tsched$ needs $\ens^{\tsched}_{\fmstate{t}{y},\alpha,\fmstate{s}{x}}$ steps in expectation.

The next goal is to provide
 a lower bound for the
expected accumulated probability 
$\eap^{\sched}_{\fmstate{s}{x},\fmstate{s}{x}}$. 
The claim is:
\begin{equation}
  \label{eap-sched}
  \eap^{\sched}_{\fmstate{s}{x},\fmstate{s}{x}} 
  \ \ \geqslant \ \
   \eap^{\tsched}_{\fmstate{s}{x},\fmstate{s}{x}} + 
    \fragment{\wp}{0}{\lgth} \cdot p    
      \cdot (e\cdot W-\eap^{\tsched}_{\fmstate{t}{y},\alpha,\fmstate{s}{x}})
  \tag{**}
\end{equation}
{\it Proof of \eqref{eap-sched}.}

As the schedulers $\sched$ and $\tsched$ agree on all paths except for extensions of $\finpath$ in case $\tsched$ chooses $\alpha$ after $\finpath$, we only have to compare the expected received weight on these extensions to get an estimation for $ \eap^{\sched}_{\fmstate{s}{x},\fmstate{s}{x}} 
-
   \eap^{\tsched}_{\fmstate{s}{x},\fmstate{s}{x}} $.
   This situation occurs with probability $ \fragment{\wp}{0}{\lgth} \cdot p $. The accumulated weight when reaching this situation is $\wgt(\finpath)\geq \SatPoint$. As $\sched$ maximizes the probability of $\neg \Fail \Until \Goal$ in this situation, these extensions of $\finpath$ contribute to the expected received weight under $\sched$ by at least $ \fragment{\wp}{0}{\lgth} \cdot p \cdot \wgt(\finpath) \cdot p_t^{\max}$. Under $\tsched$ if the scheduler chooses $\alpha$, however, these extensions contribute to the partial expectation by at most $\fragment{\wp}{0}{\lgth} \cdot p \cdot \wgt(\finpath) \cdot p_{t,\alpha}^{\max} + \eap^{\tsched}_{\fmstate{t}{y},\alpha,\fmstate{s}{x}}$. Therefore,
   \begin{align*}
    \eap^{\sched}_{\fmstate{s}{x},\fmstate{s}{x}} -    \eap^{\tsched}_{\fmstate{s}{x},\fmstate{s}{x}} & \geq \fragment{\wp}{0}{\lgth} \cdot p \cdot (\wgt(\finpath) \cdot (p_t^{\max}-p_{t,\alpha}^{\max} ) - \eap^{\tsched}_{\fmstate{t}{y},\alpha,\fmstate{s}{x}}) \\
    & \geq \fragment{\wp}{0}{\lgth} \cdot p \cdot ( \SatPoint \cdot \delta - \eap^{\tsched}_{\fmstate{t}{y},\alpha,\fmstate{s}{x}}) \\
    & \geq \fragment{\wp}{0}{\lgth} \cdot p \cdot ( e\cdot W - \eap^{\tsched}_{\fmstate{t}{y},\alpha,\fmstate{s}{x}}).
   \end{align*}

%
%

With $q=\fragment{\wp}{0}{\lgth} \cdot p$,
we obtain by \eqref{ens-sched} and \eqref{eap-sched}:
\[
 \begin{array}{rcl}
  \ens^{\sched}_{\fmstate{s}{x},\fmstate{s}{x}}
  &  \leqslant &
  \ens^{\tsched}_{\fmstate{s}{x},\fmstate{s}{x}} + 
    q \cdot
      (e-\ens^{\tsched}_{\fmstate{t}{y},\alpha,\fmstate{s}{x}})
  \\[1ex]
  \eap^{\sched}_{\fmstate{s}{x},\fmstate{s}{x}}
    &  \geqslant &
   \eap^{\tsched}_{\fmstate{s}{x},\fmstate{s}{x}} + 
        q \cdot (e\cdot W-\eap^{\tsched}_{\fmstate{t}{y},\alpha,\fmstate{s}{x}})
 \end{array} 
\]
and therefore:
\begin{equation}
 \label{lrp-sched}
 \PMP^{\sched}_{\cM}
 \ \geqslant  \
 \begin{array}{r@{\hspace*{0.1cm}}c@{\hspace*{0.15cm}}l}
       \eap^{\tsched}_{\fmstate{s}{x},\fmstate{s}{x}} & + & 
       q\cdot (e\cdot W- 
        \eap^{\tsched}_{\fmstate{t}{y},\alpha,\fmstate{s}{x}}) 
       \\
       \hline
       \\[-2.2ex]
       \ens^{\tsched}_{\fmstate{s}{x},\fmstate{s}{x}} & + & 
        q \cdot (e-\ens^{\tsched}_{\fmstate{t}{y},\alpha,\fmstate{s}{x}})
  \end{array} 
  \tag{$\ddagger$}
\end{equation}
We now use \eqref{lrp-tsched} and \eqref{lrp-sched}
to show that $\PMP^{\sched}_{\cM} \geqslant
  \PMP^{\tsched}_{\cM}$.

As $W$ is the maximal weight occurring in $\cM$, we get
\[
  \eap^{\tsched}_{\fmstate{t}{y},\alpha,\fmstate{s}{x}}
  \ \ \leqslant \ \ 
 W\cdot  \ens^{\tsched}_{\fmstate{t}{y},\alpha,\fmstate{s}{x}}
\]
Hence, if $e=\ens^{\tsched}_{\fmstate{t}{y},\alpha,\fmstate{s}{x}}$ then
$e \cdot W -  \eap^{\tsched}_{\fmstate{t}{y},\alpha,\fmstate{s}{x}} \geqslant 0$ 
and therefore:
\[
 \LPr{\sched}{\cM}{a\Until b}
 \ \geqslant  \
 \frac{\eap^{\tsched}_{\fmstate{s}{x},\fmstate{s}{x}}}
      {\ens^{\tsched}_{\fmstate{s}{x},\fmstate{s}{x}}}
 \ = \ 
 \LPr{\tsched}{\cM}{a\Until b}
\]
Suppose now that 
$e\not=\ens^{\tsched}_{\fmstate{s}{x},\alpha,\fmstate{s}{x}}$.

If $e > \ens^{\tsched}_{\fmstate{s}{x},\alpha,\fmstate{s}{x}}$
then
\[
  \frac{e\cdot W-\eap^{\tsched}_{\fmstate{s}{x},\alpha,\fmstate{s}{x}}}
       {e-\ens^{\tsched}_{\fmstate{s}{x},\alpha,\fmstate{s}{x}}}
  \ \geqslant \ 
  \frac{e\cdot W-\ens^{\tsched}_{\fmstate{s}{x},\alpha,\fmstate{s}{x}}\cdot W}
       {e-\ens^{\tsched}_{\fmstate{s}{x},\alpha,\fmstate{s}{x}}}
  \ = \ W
  \ \geqslant \ 
  \frac{\eap^{\tsched}_{\fmstate{s}{x},\fmstate{s}{x}}}
       {\ens^{\tsched}_{\fmstate{s}{x},\fmstate{s}{x}}}
\]
We now use the fact that if $x,y,z,w$ are non-negative real numbers
with $w,y>0$ then:
\[
 \begin{array}{lcl}
  \frac{x+z}{y+w} \geqslant \frac{x}{y} 
  & \text{iff} &
  \frac{z}{w} \geqslant \frac{x}{y} 
 \end{array}
\]
This yields:
\[
  \PMP^{\sched}_{\cM} 
  \ \geqslant \ 
  \frac{\eap^{\tsched}_{\fmstate{s}{x},\fmstate{s}{x}}}
       {\ens^{\tsched}_{\fmstate{s}{x},\fmstate{s}{x}}}
  \ = \  
  \PMP^{\tsched}_{\cM}
\]
It remains to consider the case 
$e < \ens^{\tsched}_{\fmstate{s}{x},\alpha,\fmstate{s}{x}}$.
Here, we use the fact for that all non-negative rational numbers $x,y,z,w$ 
with $y > w>0$ then:
\[
 \begin{array}{lcl}
  \frac{x-z}{y-w} \geqslant \frac{x}{y} 
  & \text{iff} &
  \frac{z}{w} \leqslant \frac{x}{y} 
 \end{array}
\]
In particular, if $x \leqslant y$ and $0< z < y$ then
$(x{-}z)/(y{-}z) \leqslant x/y$. Hence:
\[
  \frac{\eap^{\tsched}_{\fmstate{s}{x},\alpha,\fmstate{s}{x}}-e\cdot W}
       {\ens^{\tsched}_{\fmstate{s}{x},\alpha,\fmstate{s}{x}}-e}
       =
     W   \frac{\eap^{\tsched}_{\fmstate{s}{x},\alpha,\fmstate{s}{x}}-e\cdot W}
       {\ens^{\tsched}_{\fmstate{s}{x},\alpha,\fmstate{s}{x}}\cdot W-e\cdot W}
  \ \leqslant \ 
 W \frac{\eap^{\tsched}_{\fmstate{s}{x},\fmstate{s}{x}}}
       {\ens^{\tsched}_{\fmstate{s}{x},\fmstate{s}{x}}\cdot W}
       =
       \frac{\eap^{\tsched}_{\fmstate{s}{x},\fmstate{s}{x}}}
       {\ens^{\tsched}_{\fmstate{s}{x},\fmstate{s}{x}}}
\]
which again yields:
\[
 \begin{array}{lcl}
 \PMP^{\sched}_{\cM}
 & \geqslant  &
 \begin{array}{r@{\hspace*{0.1cm}}c@{\hspace*{0.1cm}}l}
       \eap^{\tsched}_{\fmstate{s}{x},\fmstate{s}{x}} & - & 
       q\cdot 
        (\eap^{\tsched}_{\fmstate{t}{y},\alpha,\fmstate{s}{x}}-e\cdot W) 
       \\
       \hline
       \\[-2.2ex]
       \ens^{\tsched}_{\fmstate{s}{x},\fmstate{s}{x}} & - & 
        q \cdot (\ens^{\tsched}_{\fmstate{t}{y},\alpha,\fmstate{s}{x}}-e)
  \end{array} 
  \\
  \\[-1ex]
  & \geqslant  &
 \begin{array}{c}
       \eap^{\tsched}_{\fmstate{s}{x},\fmstate{s}{x}} 
       \\
       \hline
       \ens^{\tsched}_{\fmstate{s}{x},\fmstate{s}{x}} 
  \end{array} 
  \ = \ 
  \PMP^{\tsched}_{\cM}
  \end{array}
\]
For the general case where $\Pi_{\fmstate{s}{}}$ is a 
(countable and prefix-free) set 
of paths $\finpath$, for which $\tsched$ schedules several actions 
not in $\Act^{\max}$, the argument is fairly the same.
The essential difference is that we have to consider
all state-action pairs $(\fmstate{t}{y},\alpha)$ in $\cB^{\tsched}$
such that state $\fmstate{t}{y}=(t,y)$ is reachable from $\fmstate{s}{x}$ via
a path in $\Pi_{\fmstate{s}{}}$ and 
$\alpha \in \Act(t)\setminus \Act^{\max}(t)$
such that
$\tsched$ schedules action $\alpha$ for $\fmstate{t}{y}$ 
with positive probability.
So, the lower bound for $\PMP^{\sched}_{\cM,s}$
in \eqref{lrp-sched} then has the form:
\[
 \PMP^{\sched}_{\cM}
 \ \geqslant  \
 \begin{array}{r@{\hspace*{0.1cm}}c@{\hspace*{0.15cm}}l}
       \eap^{\tsched}_{\fmstate{s}{x},\fmstate{s}{x}} & + & 
       \sum\limits_{\fmstate{t}{y},\alpha} 
          q_{\fmstate{t}{y},\alpha} \cdot 
          (e \cdot W- \eap^{\tsched}_{\fmstate{t}{y},\alpha,\fmstate{s}{x}}) 
       \\
       \hline
       \\[-2.2ex]
       \ens^{\tsched}_{\fmstate{s}{x},\fmstate{s}{x}} & + & 
       \sum\limits_{\fmstate{t}{y},\alpha} 
          q_{\fmstate{t}{y},\alpha} \cdot 
        (e-\ens^{\tsched}_{\fmstate{t}{y},\alpha,\fmstate{s}{x}})
  \end{array} 
\]
where $q_{\fmstate{t}{y},\alpha}$ is the product of the
probability in $\cB^{\tsched}$ 
of reaching $\fmstate{t}{y}$ from $\fmstate{s}{x}$ along
a path in $\Pi_{\fmstate{s}{x}}$ and
the probability for taking action $\alpha$ in state $\fmstate{t}{y}$ 
of $\cB^{\tsched}$. 
\end{proof}

The existence of a saturation point can now be used to show the following result. As all arguments from now on are completely analogous to the proof in \cite{lics2019}, we do not repeat them here.

\begin{mythm}
The maximal value ${\PMP}^{\max}_{\cM}$ in an MDP ${\cM}$ with non-negative weights is equal to the maximal mean-payoff in an at most exponentially (pseudo-polynomially) bigger MDP.
\end{mythm}

This theorem now immediately implies Theorem \ref{thm:comp_wlf}.

\section{Reductions between threshold problems}\label{app:PE_CE}


\newcommand{\CEmax}[1]{\mathit{CE}^{\max}_{#1}}
\newcommand{\PEmax}[1]{\mathit{PE}^{\max}_{#1}}
\newcommand{\CEmin}[1]{\mathit{CE}^{\min}_{#1}}
\newcommand{\PEmin}[1]{\mathit{PE}^{\min}_{#1}}

This section provides the proofs to Section \ref{sec:pe_ce}.

We briefly sketch the proof that the threshold problem for the conditional SSPP is polynomial-time reducible to the threshold problem for the partial SSPP which was given in \cite{fossacs2019}.

\begin{myprop}[\cite{fossacs2019}]
The threshold problem for the conditional SSPP is polynomial time reducible to the threshold problem for the partial SSPP.
\end{myprop}

\begin{proof}
Let $\cM$ be an MDP and $\vartheta$ a rational number. W.l.o.g. we assume that any scheduler for $\cM$ reaches the goal state with positive probability. In \cite{fossacs2019}, it is shown why we can make this assumption.
We construct a new MDP $\cN$ by adding a new state $\goal^\prime$ which is the new goal state in $\cN$ and a transition with probability $1$ from the old goal state $\goal$ to $\goal^\prime$ with weight $-\vartheta$. 
We claim that $\CE^{\max}_{\cM}\bowtie \vartheta$ if and only if $\PE^{\max}_{\cN} \bowtie 0$ for ${\bowtie} \in \{<,\leq,\geq,>\}$.

In fact, we show that the claim even holds scheduler-wise and by the existence of optimal schedulers which is shown in \cite{fossacs2019} the claim then follows.
Clearly any scheduler for $\cM$ can be seen as a scheduler for $\cN$ and vice versa. 
Let $\sched$ be such  a scheduler. Then, we have
\begin{align*}
& \PE^\sched_{\cN} = \PE^\sched_{\cM} - \vartheta \cdot \Pr^{\sched}_{\cM}(\Diamond \goal) \bowtie 0 \\
\text{iff } \ & \frac{\PE^\sched_{\cM}}{\Pr^{\sched}_{\cM}(\Diamond \goal)} \bowtie \vartheta \
\text{ iff } \  \CE^\sched_{\cM}\bowtie \vartheta . \qedhere
\end{align*}

\end{proof}

This reduction, however, introduces a negative weight to the MDP. We show in the sequel that for acyclic MDPs with non-negative weights, there is a reduction from the threshold problem of the conditional SSPP to the partial SSPP that does not introduce a negative weight.

\tudparagraph{1ex}{PSPACE-hardness of the  
             threshold problems for the partial SSPP in non-negative acyclic MDPs (see Proposition \ref{prop:CE_to_PE})}

In this section, we address the threshold problems for maximal and minimal partial expectations 
  with strict or non-strict inequalities. These are the following four problems:
Given a weighted MDP $\cM$ and
a rational threshold $\threshold$, decide whether
\begin{center}   
  \begin{tabular}{l@{\hspace*{0.2cm}}c@{\hspace*{0.2cm}}l@{\hspace*{1cm}}l}
    $\PEmax{\cM}$ & $>$ & $\threshold$ , \\[0.7ex]
    $\PEmax{\cM}$ & $\geqslant$ & $\threshold$ , \\[0.7ex]
    $\PEmin{\cM}$ & $<$ & $\threshold$ , \\[0.7ex]
    $\PEmin{\cM}$ & $\leqslant$ & $\threshold$. 
  \end{tabular}
\end{center}
We now show that the four  threshold problems are PSPACE-complete for
acyclic MDPs with non-negative integer weights.
The proof goes via a polynomial-time reduction from the threshold problem for conditional expectations with
strict threshold conditions:

  \begin{tabular}{ll}
    given: & 
    an acyclic MDP $\cM$ with non-negative integer weights and \\ 
    & a positive rational threshold $\threshold$,
    \\[0ex]
    question: & does $\CEmax{\cM} > \threshold$ hold?
 \end{tabular}
 
and the analogous problem where the task is to check whether
$\CEmin{\cM} < \threshold$.
These threshold problems are known to be PSPACE-complete for acyclic MDPs with
non-negative integer (or rational) weights \cite{tacas2017}.

In the sequel, we assume $\cM$ has distinguished states $s_0$ (initial state)
and two traps 
$\goal$ (target state) and $\fail$ such that
(i) all states of $\cM$ are reachable from $s_0$ and
(ii) $\goal$ is reachable from
all non-trap states in $\cM$.

\tudparagraph{1ex}{Additional assumption.}
For the threshold problem for the conditional SSPP with strict threshold conditions, 
it is no restriction to assume that 
$\Pr^{\min}_{\cM,s_0}(\Diamond \goal)>0$.

To see this, let us suppose that the given MDP $\cM$ has schedulers
under which $\goal$ is not reachable from $s_0$.
Let $\cN$ denote the (rational weighted) 
MDP that extends $\cM$ by a fresh initial state
$s_0'$ with a single enabled action $\tau$ where
$P_{\cN}(s_0',\tau,s_0)= P_{\cN}(s_0',\tau,\goal)=\frac{1}{2}$
and $\wgt_{\cN}(s_0',\tau) = \threshold$ (see Figure \ref{fig:ipad1}).

\begin{figure}[ht]
\begin{center}
\scalebox{1}{
\begin{tikzpicture}
  [shorten >=1pt,node distance=12mm and 1.5cm,on grid,auto,semithick]
  
  \node[state,inner sep=1pt,minimum size=10mm] (s) {$s_0^\prime$};
\node[state,inner sep=1pt,minimum size=10mm] (t) [right = 30mm of s] {$s_0$};
\node[state,inner sep=1pt,minimum size=10mm] (u) [below = 30mm of t] {$\goal$};
\node[state,inner sep=1pt,minimum size=10mm] (w) [right = 30mm of u] {$\fail$};
\node (label) [right=30mm of t] {MDP $\cM$};

 \path[->]
 (s) edge [bend right=20] node [pos=.4,below] {$\tau | +\vartheta$} node [pos=.2,above] {$\frac{1}{2}$}  node [pos= 0.1, anchor=center]  (h1) {} (t)
 (s) edge [bend right=20] node  [pos=0.15, anchor=center] (h2) {} node [pos=.3,below] {$\frac{1}{2}$}  (u)
 ;
 \path
 (h1.center) edge  [bend left] (h2.center)
 (h2.center) edge  [bend right] (h1.center)
;
\draw[draw=black,rounded corners,fill= gray, opacity=0.2] (t.135)+(-.5,.5) rectangle ++(4.5,-4.5);

\end{tikzpicture}
}
\end{center}

\caption{Construction of the MDP $\cN$.}\label{fig:ipad1}
\end{figure}
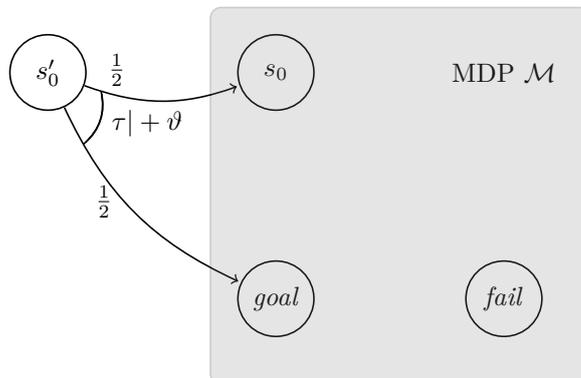

Obviously, each scheduler for $\cM$  can be viewed as
a scheduler for $\cN$, and vice versa. Moreover, for each scheduler $\sched$
we have:
\[
\begin{array}{l}
  \Pr^{\sched}_{\cN,s_0}(\Diamond \goal)
  \ \ = \ \ 
  \frac{1}{2} + \frac{1}{2} \Pr^{\sched}_{\cM,s_0}(\Diamond \goal)
  \\[1ex]
  \PE^{\sched}_{\cN} \ \ = \ \
  \frac{1}{2} \threshold \, + \, \frac{1}{2}  \PE_{\cM}^{\sched}
\end{array}
\]
With $y=\Pr^{\sched}_{\cM,s_0}(\Diamond \goal)$ we get:
\[
  \CE^{\sched}_{\cN} \ \ = \ \ 
  \frac{\, \frac{1}{2} \threshold \, + \, \frac{1}{2} \PE_{\cM}^{\sched} \, }
       {\frac{1}{2} + \frac{1}{2} y}
\]
Note that 
$\CE^{\sched}_{\cN} \not= \threshold$ implies $y>0$. Likewise,
$\PE^{\sched}_{\cM}>0$ implies $y>0$.
But then:
\[
\begin{array}{lll}
    & & \CE^{\sched}_{\cN} \ > \ \threshold
    \\[1ex]
    \text{iff} & & 
    \frac{1}{2} \threshold \, + \, \frac{1}{2} \PE_{\cM}^{\sched} 
    \ \ > \ \ (\frac{1}{2}  + \frac{1}{2} y)\threshold
    \ \ = \ \ \frac{1}{2} \threshold\, + \, \frac{1}{2} \threshold y
    \\[1ex]
    \text{iff} & & 
    \PE_{\cM}^{\sched} \ > \  \threshold y
    \\[1ex]
    \text{iff} & & 
    \CE_{\cM}^{\sched} \ > \  \threshold 
\end{array}
\]
and similary, $\CE^{\sched}_{\cN} < \threshold$ iff
$\CE_{\cM}  < \threshold$.
This yields: 
\begin{center}
   $\CEmax{\cN} > \threshold$ \ iff \ $\CEmax{\cM} > \threshold$
   \qquad and \qquad
   $\CEmin{\cN} < \threshold$ \ iff \ $\CEmin{\cM} < \threshold$
\end{center}

An integer weighted MDP enjoying an analogous property can be obtained
as follows. Let $\threshold = a/b$ where $a,b$ are co-prime positive integers.
and let $\cK$ denote the MDP arising from $\cN$ by multiplying all weights
with $b$. (Thus, $\wgt_{\cK}(s_0',\tau)=a$ and 
$\wgt_{\cK}(s,\alpha)=b \cdot \wgt_{\cM}(s,\alpha)$ for all state-action pairs
$(s,\alpha)$ in $\cM$.)
Then, $\PE^{\sched}_{\cK} = b\cdot \PE^{\sched}_{\cN}$ for each scheduler
$\sched$. Hence, $\CEmax{\cK} > a$ iff $\CEmax{\cM} > \threshold$,
and similarly $\CEmax{\cK} < a$ iff $\CEmax{\cM} < \threshold$.

\tudparagraph{1ex}{Reduction.}
We now describe a polynomial reduction from the threshold problem for the conditional SSPP 
with strict
threshold condition in acyclic MDPs under the additional assumption that
the given acyclic MDP $\cM$ enjoys the property
$\Pr^{\min}_{\cM,s_0}(\Diamond \goal)>0$.
We are going to construct an acyclic MDP $\cN$ and threshold values
$\threshold_1$ and $\threshold_2$ such that
$\CEmax{\cM} > \threshold$ iff
$\PEmax{\cN} > \threshold_1$ iff
$\PEmax{\cN} \geqslant \threshold_2$ (see Lemma \ref{PEmax})
and an analogous statement for minimal conditional/partial expectations
(see Lemma \ref{PEmin}).

The structure of this MDP $\cN$ is sketched in the following Figure \ref{fig:ipad2}:

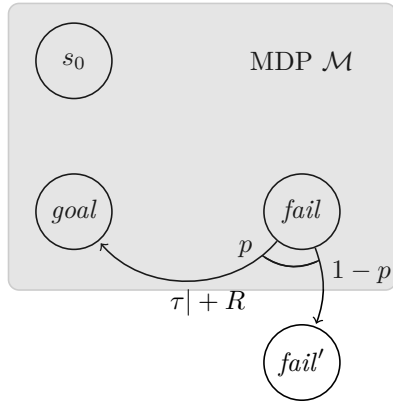
\begin{figure}[h]
\begin{center}
\scalebox{1}{
\begin{tikzpicture}
  [shorten >=1pt,node distance=12mm and 1.5cm,on grid,auto,semithick]

\node[state,inner sep=1pt,minimum size=10mm] (t) {$s_0$};
\node[state,inner sep=1pt,minimum size=10mm] (u) [below = 20mm of t] {$\goal$};
\node[state,inner sep=1pt,minimum size=10mm] (w) [right = 30mm of u] {$\fail$};
 \node[state,inner sep=1pt,minimum size=10mm] (s) [below = 20mm of w] {$\fail^\prime$};
\node (label) [right=30mm of t] {MDP $\cM$};

 \path[->]
 (w) edge [bend left=50] node [pos=.4,below] {$\tau | +R$} node [pos=.2,above] {$p$}  node [pos= 0.1, anchor=center]  (h1) {} (u)
 (w) edge [bend left=20] node  [pos=0.15, anchor=center] (h2) {} node [pos=.3,right] {$1-p$}  (s)
 ;
 \path
 (h2.center) edge  [bend left] (h1.center)
 (h1.center) edge  [bend right] (h2.center)
;
\draw[draw=black,rounded corners,fill= gray, opacity=0.2] (t.135)+(-.5,.4) rectangle ++(4.6,-3.4);

\end{tikzpicture}
}
\end{center}

\caption{Construction of the MDP $\cN$. The probability $p$ and the weight $R$ are chosen such that $pR=\vartheta$.}\label{fig:ipad2}
\end{figure}

So, the essential task is to find appropriate values for $p$ and $R$
and corresponding threshold values for $\cN$.
For this we need to tackle the problem that different
paths from $s_0$ to $\fail$ in $\cM$ might have different
accumulated weights.

For each non-trap state $s$
in $\cM$, let
 $m_s$ denote the least common multiple
of the denominators of the transition probabilities $P(s,\alpha,t)$
and let $m$ be the product of the values $m_s$ for all non-trap states $s$.
The number of digits of an binary (or decimal) representation 
of $m$ is polynomially bounded in the size of $\cM$.
To see this, we observe that $m \leqslant D^{|S|}$ where 
$D=\max_{s,\alpha,t} \mathit{denom}(s,\alpha,t)$
and $\mathit{denom}(s,\alpha,t)=1$ if $P(s,\alpha,t)=0$,
while for $P(s,\alpha,t)>0$,
$\mathit{denom}(s,\alpha,t)$ denotes the denominator of 
the unique representation
of $P(s,\alpha,t)$ as the quotient of co-prime positive integers.
This yields
$\log m \leqslant d \cdot |S|$ where $d$ is 
the maximal number of digits of the binary representations of 
the weights in $\cM$, i.e., $d$ is the least natural number with
$D < 2^d$.

As $\cM$ is acyclic, the probability of each path $\pi$ from $s_0$ to
$\goal$ or $\fail$ is a rational number of the form $\ell/m$ for
for some natural number $\ell$.
The same applies to the probabilities
for reaching $\goal$ from $s_0$ 
and the partial expectations 
under deterministic schedulers.
That is, if $\sched$ is a deterministic scheduler then
\[
  \Pr^{\sched}_{\cM,s_0}(\Diamond \goal), \ \
  \PE^{\sched}_{\cM,s_0}
  \ \in \ \Bigl\{\, \frac{\ell}{m} \, : \, \ell \in \Nat \, \Bigr\}
\]
Consider a representation of the threshold $\threshold$ as the quotient
$a/b$ of two positive integers $a,b$.
Let
\[
   \delta \ \ \eqdef \ \ \frac{1}{bm}
\]
Then, for each deterministic scheduler $\sched$,
the value
$\threshold \cdot \Pr^{\sched}_{\cM,s_0}(\Diamond \goal)$
and the partial expectation $\PE^{\sched}_{\cM,s_0}$
are integer-multiples of $\delta$.
This yields:
\begin{equation}
  \label{magic}
  \text{If $y=\Pr^{\sched}_{\cM,s_0}(\Diamond \goal)$ and
         $\PE^{\sched}_{\cM} > 
            \threshold y$ 
        then
        $\PE^{\sched}_{\cM} \geqslant
         \threshold y + \delta$.}
  \tag{*}
\end{equation}
Let now 
\[
  w \ \ \eqdef \ \ 
  1 + \max \, \bigl\{ \, \wgt(\pi) \, : \, 
    \text{$\pi$ is a path from $s_0$ to $\fail$ in $\cM$} \, \bigr\}
\]
and define
\[
   p \ \eqdef \ \frac{\delta}{2w}
   \quad \text{and} \quad
   R \ \eqdef \ \frac{2w\threshold}{\delta}
\]
Note that $w$ is finite (recall that $\cM$ is acyclic) and
computable in polynomial time 
and that the logarithmic length of the numerator and denominator of the
rational numbers $p$ and $R$ is polynomial in the sizes 
of the given MDP $\cM$ and the threshold value $\threshold$.
Obviously: 
\begin{equation}
  \label{p-and-R}
  pw \ = \ \frac{\delta}{2} 
  \quad \text{and} \quad
  pR = \threshold 
  \tag{$\dagger$}
\end{equation} 
We now construct a new MDP $\cN$ that extends $\cM$ by a fresh state $\fail'$
and an action $\tau$ that is enabled in 
state $\fail$ 
with weight $\wgt_{\cN}(\fail,\tau_i)=R$ and 
the transition probabilities 
$P_{\cN}(\fail,\tau,\goal)=p$,
$P_{\cN}(\fail,\tau,\fail')=1{-}p$.
For all other states, the enabled actions and their transition probabilities
and weights are the same as in $\cM$.


\begin{mylem}
\label{PEmax}
$
  \CEmax{\cM} > \threshold 
  \quad \ \text{iff} \ \quad
  \PEmax{\cN} \geqslant \threshold +\delta
  \quad \ \text{iff} \ \quad
  \PEmax{\cN} > \threshold + \frac{\delta}{2}
$
\end{mylem}

\begin{proof}
Let us first observe that $\cM$ and $\cN$ have the same schedulers.
Moreover, if $\sched$ is a scheduler and
$y=\Pr^{\sched}_{\cM,s_0}(\Diamond \goal)$ then:
\begin{enumerate}
\item [(1)]
  $\Pr^{\sched}_{\cN,s_0}(\Diamond \goal) \ = \ y + p(1{-}y)$
\item [(2)]
  $
     \PE^{\sched}_{\cM} +  \threshold(1{-}y) 
     \ \ \ \leqslant \ \ \
     \PE^{\sched}_{\cN}
     \ \ \ < \ \ \
     \PE^{\sched}_{\cM}  +  \threshold(1{-}y)
         + \frac{\delta}{2}
  $

  Proof:
  The claim is obvious if $y=1$ because then
  $\CE^{\sched}_{\cM}=\PE^{\sched}_{\cM}=\PE^{\sched}_{\cN}$.

  Suppose now that $y<1$.
  As the accumulated weight of all paths from $s_0$ to $\fail$
  is at most $w{-}1$ we have:
  \[
     \PE^{\sched}_{\cM} + pR(1{-}y) 
     \ \ \leqslant \ \
     \PE^{\sched}_{\cN}
     \  \ <  \ \
     \PE^{\sched}_{\cM} + p(R{+}w)(1{-}y)
  \]
  The claim then follows from \eqref{p-and-R}.

\end{enumerate}
Suppose now that $\CEmax{\cM} > \threshold$.
Pick a deterministic scheduler $\sched$ such that
$\CE^{\sched}_{\cM} > \threshold$.
Thus, with $y=\Pr^{\sched}_{\cM,s_0}(\Diamond \goal)$ we have $y>0$ and
\[
  \PE^{\sched}_{\cM} \ \ > \ \ 
  \threshold y
\]
By \eqref{magic} we have 
\[
  \PE^{\sched}_{\cM} \ \ \geqslant \ \ 
  \threshold y + \delta
\]
Using the first inequality of statement (2) we obtain:
\[
  \PE^{\sched}_{\cN}
  \ \ \ \stackrel{\text{(2)}}{\geqslant} \ \ \
  \PE^{\sched}_{\cM} +  \threshold(1{-}y) 
  \ \ \ \geqslant \ \ \
  \threshold y \, + \, \delta \, + \, \threshold \, - \, \threshold y
  \ \ = \ \ 
  \threshold +\delta
\]
Hence, $\PEmax{\cN} \geqslant \threshold + \delta$.

Suppose now that $\PEmax{\cN} > \threshold + \frac{\delta}{2}$.
Pick a deterministic scheduler $\sched$ such that
$\PE^{\sched}_{\cN} > \threshold + \frac{\delta}{2}$.
Let
\[
  y \ \ \eqdef \ \ 
  \Pr^{\sched}_{\cM,s_0}(\Diamond \goal) 
\]
The assumption $\Pr^{\min}_{\cM,s_0}(\Diamond \goal) >0$
yields $y>0$.
Using the second inequality of statement (2) we obtain:
\[
  \PE^{\sched}_{\cM} + \threshold (1{-}y) + \frac{\delta}{2}
  \ \ \ \stackrel{\text{(2)}}{>} \ \ \
  \PE^{\sched}_{\cN} 
  \ \ > \ \ \threshold + \frac{\delta}{2}
\]
This yields:
\[
  \PE^{\sched}_{\cM}  \ \ > \ \ 
  \threshold y 
\]
But then 
$\CE^{\sched}_{\cM} > \threshold$,
and therefore $\CEmax{\cN} > \threshold$. 
\end{proof}


\begin{mylem}
\label{PEmin}
$
  \CEmin{\cM} < \threshold 
  \quad \ \text{iff} \ \quad
  \PEmin{\cN} \leqslant \threshold - \frac{\delta}{2}
  \quad \ \text{iff} \ \quad
  \PEmin{\cN} < \threshold - \frac{\delta}{2}
$
\end{mylem}

\begin{proof}
The argument is fairly the same as in the proof of Lemma \ref{PEmax}.
Instead of \eqref{magic} we use here the following fact:
\begin{equation}
  \label{magic-PEmin}
  \text{If $y=\Pr^{\sched}_{\cM,s_0}(\Diamond \goal)$ and
         $\PE^{\sched}_{\cM} < 
            \threshold y$ 
        then
        $\PE^{\sched}_{\cM} \leqslant
         \threshold y - \delta$.}
  \tag{**}
\end{equation}
Suppose first that $\CEmin{\cM} < \threshold$.
Pick a deterministic scheduler $\sched$ such that
$\CE^{\sched}_{\cM} < \threshold$ and let
$y=\Pr^{\sched}_{\cM,s_0}(\Diamond \goal)$. 
Then, $y>0$ and
$\PE^{\sched}_{\cM} < \threshold y$.
By \eqref{magic-PEmin} we get:
\[
  \PE^{\sched}_{\cM} \ \ \leqslant \ \ 
  \threshold y - \delta
\]
We now rely on
the second inequality of statement (2) 
in the proof of Lemma \ref{PEmax} 
and obtain:
\[
  \PE^{\sched}_{\cN}
  \ \ \stackrel{\text{(2)}}{<} \ \
  \PE^{\sched}_{\cM} +  \threshold(1{-}y) +\frac{\delta}{2}
  \ \ \leqslant \ \
  \threshold y \, - \, \delta \, + \, \threshold \, - \, \threshold y
  \, + \, \frac{\delta}{2}
  \ \ = \ \ 
  \threshold - \frac{\delta}{2}
\]
Hence, $\PEmin{\cN} < \threshold - \frac{\delta}{2}$.

Suppose now that $\PEmin{\cN} \leqslant \threshold - \frac{\delta}{2}$.
Pick a deterministic scheduler $\sched$ such that
$\PE^{\sched}_{\cN} \leqslant \threshold - \frac{\delta}{2}$.
By assumption
\[
  y \ \ \eqdef \ \ 
  \Pr^{\sched}_{\cM,s_0}(\Diamond \goal) \ \ > \ 0 
\]
The first inequality of statement (2) 
in the proof of Lemma \ref{PEmax} yields:
\[
  \PE^{\sched}_{\cM} + \threshold (1{-}y) 
  \ \ \ \stackrel{\text{(2)}}{\leqslant} \ \ \
  \PE^{\sched}_{\cN} 
  \ \ \leqslant \ \ \threshold - \frac{\delta}{2}
\]
Hence:
\[
  \PE^{\sched}_{\cM}  \ \ \leqslant \ \ 
  \threshold y  - \frac{\delta}{2} \ \ < \ \ \threshold y
\]
But then 
$\CE^{\sched}_{\cM} < \threshold$, which implies
$\CEmin{\cN} < \threshold$. 
\end{proof}


By the PSPACE-hardness of the threshold problem
``does $\CEmax{\cM}> \threshold$ hold?'' or 
``does $\CEmin{\cM}< \threshold$ hold?'' 
\cite{tacas2017}
and Lemma \ref{PEmax} and Lemma \ref{PEmin},
the four PE-problems are PSPACE-hard too.
Membership of the four threshold problems for the partial SSPP to PSPACE
follows from the inter-reducibility with  threshold problems for maximal conditional expectations
for acyclic MDPs.
Hence:

\begin{mycor}
  The four threshold problems for the partial SSPP in acyclic MDPs with non-negative integer weights
  are PSPACE-complete.
\end{mycor}

\begin{myprop}[Proposition \ref{prop:PE_to_WLF}]\label{prop_app:PE_to_WLF}
The threshold problem for weighted long-run frequencies, ``Does $\PMP^{\max}_\cM\bowtie \vartheta$ hold?'', in MDPs with non-negative weights is PSPACE-hard.
\end{myprop}

\begin{proof}
We reduce from the threshold problem of the  partial SSPP in acyclic MDPs with non-negative weights: Given such an MDP $\cM$ and a threshold $\vartheta$,  we can transform $\cM$ to an MDP $\cN$ in which all paths to $\fail$ or $\goal$ have the same length $\ell$ in polynomial time by simply adding (polynomially many) intermediate states and transitions with no weight. In $\cN$, we let $\Goal=\{\goal\}$ and $\Fail=\{\fail\}$ and add transitions with porbability $1$ and weight $0$ from $\fail$ and $\goal$ to the initial state. Then $\PMP^{\max}_\cN\bowtie \vartheta/(\ell+1)$ if and only if $\PE^{\max}_\cM \bowtie \vartheta$ for ${\bowtie} \in \{<,\leq,\geq,>\}$.
\end{proof}






\end{document}